\documentclass[a4paper,10pt]{article}

\usepackage{amsfonts}
\usepackage{amsmath}
\usepackage{amsthm}
\usepackage{bbm}

\usepackage{nicefrac}

\usepackage{pgfplots}
\usepackage{pgfplotstable}
\usetikzlibrary{pgfplots.groupplots}
\usetikzlibrary{matrix}

\definecolor{color1}{RGB}{0,  113.9850,  188.9550}
\definecolor{color2}{RGB}{216.7500,   82.8750,   24.9900}
\definecolor{color3}{RGB}{236.8950,  176.9700,   31.8750}
\definecolor{color4}{RGB}{125.9700,   46.9200,  141.7800}
\definecolor{color5}{RGB}{161.9250,   19.8900,   46.9200}

\newlength{\plotWidth}
\newlength{\plotHeight}
\newlength{\vSpace}
\newlength{\hSpace}
\newcommand{\total}[2]{{\textstyle\sum\limits_{#1}^{#2}}}
\newcommand{\ttotal}[2]{{\textstyle\sum\hspace{-0.2em}{}_{#1}^{#2}}}

\newcommand{\tproduct}[2]{{\textstyle\prod\hspace{-0.2em}{}_{#1}^{#2}}}
\newcommand{\union}[2]{{\textstyle\bigcup\limits_{#1}^{#2}}}

\newcommand{\process}[4]{(#1_{#2})_{#2=#3}^{#4}} 
\newcommand{\processdef}[4]{#1=\process{#1}{#2}{#3}{#4}} 

\usepackage[shortlabels]{enumitem}
\setlist[enumerate]{leftmargin=.5in}
\setlist[itemize]{leftmargin=.5in}

\theoremstyle{definition}
\newtheorem{theorem}{Theorem}
\numberwithin{theorem}{section}
\newtheorem{proposition}{Proposition}
\numberwithin{proposition}{section}
\newtheorem{construction}{Construction}
\numberwithin{construction}{section}

\newtheorem{example}{Example}
\numberwithin{example}{section}

\theoremstyle{remark}
\newtheorem{remark}{Remark}
\numberwithin{remark}{section}

\numberwithin{equation}{section}

\title{\textbf{Optimal investment and contingent claim valuation with exponential disutility under proportional transaction costs}}

\author{Alet Roux\thanks{Department of Mathematics, University of York, Heslington, YO10 5DD, United Kingdom. Email address: alet.roux@york.ac.uk}
\and Zhikang Xu\thanks{Willis Towers Watson, 51 Lime St, London, EC3M 7DQ, United Kingdom. Email: zk.xu@outlook.com. Most of the research presented in this paper was conducted while this author was a PhD student in the Department of Mathematics, University of York.}}

\usepackage{amsopn}
\DeclareMathOperator{\conv}{conv}
\DeclareMathOperator{\relint}{ri}

\DeclareMathOperator{\dom}{dom}
\DeclareMathOperator{\epigraph}{epi}
\DeclareMathOperator{\cone}{cone}
\DeclareMathOperator{\cl}{cl}
\DeclareMathOperator{\interior}{int}

\usepackage{mathtools}

\usepackage{dcolumn}
\newcolumntype{d}[1]{D{.}{.}{#1}}

\usepackage[longnamesfirst]{natbib}

\begin{document}

\maketitle

\begin{abstract}
  We consider indifference pricing of contingent claims consisting of payment flows in a discrete time model with proportional transaction costs and under exponential disutility. This setting covers utility maximisation as a special case. A dual representation is obtained for the associated disutility minimisation problem, together with a dynamic procedure for solving it. This leads to efficient and convergent numerical procedures for indifference pricing, optimal trading strategies and shadow prices that apply to a wide range of payoffs, a large range of time steps and all magnitudes of transaction costs.

\emph{Keywords}: transaction costs, option pricing, utility maximisation, entropy, indifference pricing, generalised convex hull, dynamic programming
\end{abstract}


\section{Introduction}

The price of a contingent claim in a complete market is uniquely determined by the principle of replication: it is the discounted expectation of the claim price under the (unique) martingale measure. However, the presence of transaction costs can lead to the curious contradiction that superreplicating a claim may involve less trading (and lower transaction costs) than exact replication, and therefore be less expensive, so that the replication price can in fact lead to arbitrage. Furthermore, financial markets with transaction costs generally admit many different martingale measures, leading to intervals of no-arbitrage claim prices. This means that subjective factors, such as an investor's risk appetite, come into play when determining the price of a claim. The indifference principle offers a compelling alternative to replication and arbitrage pricing: it states that the seller of a claim will charge (at least) a price that will allow him to sell the claim without increasing the risk of his existing financial position. This is called the \emph{indifference price}. As a special case, the \emph{reservation price} is a price that would have allowed the seller to cover a claim at an acceptable level of risk, had their existing position been zero (in other words, not taking it into account). This is often associated with the terms ``economic capital'' in banking, and ``technical provisions'' or ``reserving'' in insurance. 

Indifference pricing based on utility maximisation has been well studied in the literature on proportional transaction costs. Work in continuous time has mostly focused on adapting stochastic optimal control and other techniques from friction-free models (such as the Black-Scholes model), and in recent years have led to numerical approximation and asymptotics for small transaction costs; see the works by \citet{bichuch2014pricing,Pricing1997Davis,davis1993european,hodges1989optimal,kallsen2015option,monoyios2003efficient,monoyios2004option,whalley1997asymptotic}, for example. Results obtained in continuous time models typically assume continuous trading, which limits their applicability in realistic settings \citep{Dorfleitner_Gerer2016}, hence motivating the need for continued theoretical and numerical work in the discrete time setting. 

The present paper is motivated by the work of \citet{pennanen2014optimal}, who studied indifference pricing in a very general discrete time setting, including proportional transaction costs. In view of the fact that financial liabilities in banking and insurance often consist of sequences of payment streams, such as swaps, coupon paying bonds, insurance premia, etc, \cite{pennanen2014optimal} extended the classical utility maximisation framework, which focuses on the expected disutility of hedging shortfall at the expiration date of the liability faced by an investor (and insists on self-financing trading at other times), to a more flexible framework which allows hedging to fall short at intermediate steps, takes into account the expected total disutility of hedging shortfall at all steps, and presents theoretical results for contingent claims consisting of cash payment streams and a very general class of disutility functions.

Allowing hedging to fall short at intermediate time steps means that there is also a connection between the current work and another important strand in the transaction cost literature, namely maximising utility from consumption. An important notion in the study of these problems is the \emph{shadow price}, which is a price process taking values in the bid-ask spread of the model with proportional transaction costs, with the property that maximising expected utility from consumption in the friction-free model with this price process, leads to the same maximal utility as in the original market with transaction costs. \citet{Kallsen_Muhle-Karbe2011} and \citet{Rogala_Stettner2015} showed that shadow prices exist in discrete time in a similar (though incompatible) technical setting to the current paper. Working in general discrete time models, \citet{Czichowsky_Muhle-Karbe_Schachermayer2014} demonstrated that there is a link between the solution to the dual problem, and the existence of a shadow price. The existence of shadow processes in more general models is by no means guaranteed. Additionally, shadow prices may not be tractable, leading to the use of asymptotic expansions and/or restrictions in the magnitude of transaction costs. In the context of continuous-time models, see the earlier paper of \cite{cvitanic1996hedging}, as well as more recent contributions by \cite{kallsen2010using}, \cite{Gerhold_Muhle-Karbe_Schachermayer2013}, \cite{Gerhold_Guasoni_Muhle-Karbe_Schachermayer2014}, \cite{Herczegh_Prokaj2015}, \cite{Czichowsky_Schachermayer_Yang2017}, \cite{Czichowsky_Schachermayer2016,Czichowsky_Schachermayer2017}, \cite{Lin_Yang2017} and \cite{Gu_Lin_Yang2017}.

The present paper specialises the model of \citet{pennanen2014optimal} to exponential utility and proportional transaction costs, which allows the use of powerful dual methods, and finite state space, motivated by the need for numerical results. Our results apply to contingent claims with physical delivery (in other words, streams of portfolios rather than just cash). We propose a backward recursive procedure that can be used to solve the utility maximisation problem and compute indifference prices, together with an efficient and convergent numerical approximation method (with error bounds). This procedure has polynomial running time in recombinant models and for path-independent claims, and does not require the construction of a shadow price process, which is in general path-dependent (a known difficulty in models with proportional transaction costs). Nevertheless, the outputs from this procedure can be used to construct a shadow price process and accompanying martingale measure, together with an optimal hedging strategy. This latter construction is performed by (forward) induction, which makes it practical for studying individual scenarios, despite the path-dependence of the objects that are being studied. Our results apply to all magnitudes of transaction costs, and our numerical methods work for a large range of time steps; \cite{Xu2018Pricng} reported a number of more demanding numerical results that have not been included in this paper for lack of space.

The results reveal interesting features of disutility minimisation problems and indifference prices. In particular, because asset holdings in the model can be carried over between different time periods, the value of the disutility minimisation problem of an investor faced with delivering a portfolio stream depends only on the total payment involved in the stream (suitably discounted), which implies that indifference prices also depend only on the total payment due. Nevertheless, the additional flexibility offered by allowing hedging to fall short at time periods other than the final time leads to smaller spreads in indifference prices, when compared to utility indifference pricing spreads. Our numerical results further suggest that there is a complex relationship between disutility indifference prices and the real-world measure.

The numerical methods and examples work reported in this paper extend and complement the limited work in the literature for discrete time models with proportional transaction costs. The results on disutility minimisation generalise the results of \citet{castaneda2011utility} in a one-step binomial model with proportional transaction costs. To put the power of the numerical methods into context, previously reported numerical results are limited to European put options in a $3$-step Cox-Ross-Rubinstein binomial model with convex transaction costs and exponential utility \citep{cetin2007modeling}, utility indifference prices of a European call option under exponential utility in a binomial tree model with $6$ steps and proportional transaction costs \citep{quek2012portfolio}, and numerical solution of utility maximisation problems under power utility with multiple assets and proportional transaction costs \citep{cai2013numerical}.

Whilst we restrict our attention to indifference prices (payable at time $0$ in cash) rather than indifference swap rates \citep[used by][]{pennanen2014optimal} for brevity, we believe that the extension is straightforward \citep[preliminary work reported by][]{Xu2018Pricng}. We believe that our work can be generalised to include measuring hedging shortfall in terms of portfolios rather than just cash; this is the subject of ongoing research, as is application of these methods to other classes of utility functions and multi-asset models.

The paper is arranged as follows. Background information on arbitrage and superhedging in discrete time models with proportional transaction costs is collected in Section~\ref{sec:intro}. The disutility minimisation problem that forms the basis of the indifference pricing framework is introduced in Section~\ref{sec:Disutility-minimisation}; this includes utility maximisation as a special case. In Section~\ref{sec:The-dual-problem} we derive a Lagrangian dual formulation for the disutility minimisation problem. Indifference prices are introduced in Section~\ref{sec:Indifference-pricing}, together with arbitrage pricing bounds. A dynamic procedure for solving the disutility minimisation problem and computing indifference prices is presented in Section~\ref{sec:Solution-construction}, together with a procedure for constructing the shadow price. A procedure for constructing optimal hedging strategies is presented in Section \ref{sec:optimal-injection-investment}. Section~\ref{sec:numerical-examples} contains a number of illustrative numerical examples. \ref{sec:conv-hull} summarises a number of properties of a generalisation of the convex hull of convex functions that appears in the dynamic procedure of Section~\ref{sec:Solution-construction}; this includes a numerical approximation by piecewise linear functions, complete with error bound. Proofs of all results in the main part of the paper appear in Appendix \ref{app:proofs}.

\section{Preliminaries}
\label{sec:intro}

\subsection{Discrete-time model with proportional transaction\\costs}
\label{sec:intro:model}

In this paper we consider a discrete-time financial market model with a finite time horizon $T\in\mathbb{N}$ and trading dates $t=0,\ldots,T$ on a finite probability space $(\Omega,\mathcal{F},\mathbb{P})$ equipped with a filtration $\process{\mathcal{F}}{t}{0}{T}$. We assume without loss of generality that $\mathcal{F}_0=\{ \Omega,\emptyset\} $, $\mathcal{F}_T=\mathcal{F}=2^{\Omega}$ and $\mathbb{P}(\omega)>0$ for all $\omega\in\Omega$. For each $t$, the collection of atoms of $\mathcal{F}_t$ is denoted by $\Omega_t$. The elements of $\Omega_t$ are called the \emph{nodes} of the model at time~$t$, and they form a partition of $\Omega$. For each $\omega$ and $t$, denote by $\omega_t$ the unique node $\nu\in\Omega_t$ such that $\omega\in\nu$. A node $\nu\in\Omega_{t+1}$ is said to be a \emph{successor} of a node $\mu\in\Omega_t$ if $\nu\subseteq\mu$. Denote the collection of successors of any given node $\mu\in\Omega_t$ by $\mu^+$, and define the transition probability from $\mu$ to any successor node $\nu\in\mu^+$ by $p_{t+1}^\nu\coloneqq\tfrac{\mathbb{P}(\nu)}{\mathbb{P}(\mu)}$.

For every $t$ and $d=1,2$, let $\mathcal{L}_t^{d}$ be the space of $\mathbb{R}^{d}$-valued $\mathcal{F}_t$-measurable random variables. Every random variable $x\in\mathcal{L}_t^{d}$ satisfies $x(\omega)=x(\omega^{\prime})$ for all $\omega,\omega^{\prime}\in\nu$ on every node $\nu\in\Omega_t$, and this common value is denoted $x^\nu$. A similar convention applies to $\mathcal{F}_t$-measurable random functions $f:\Omega\times\mathbb{R}^d\rightarrow\mathbb{R}$. Let $\mathcal{N}^{d}$ be the space of adapted $\mathbb{R}^{d}$-valued processes. We write $\mathcal{L}_t=\mathcal{L}_t^1$ and $\mathcal{N}=\mathcal{N}^1$ for convenience. For $d=2$ we will adopt the convention that the first and second components of any random variable $c\in\mathcal{L}_t^2$ or process $c\in\mathcal{N}^2$ are denoted $c^b$ and~$c^s$, respectively.

The financial market model consists of a risky and risk-free asset. The price of the risk-free asset, \emph{cash}, is constant and equal to $1$ at all times. This is equivalent to assuming that interest rates are zero, or that asset prices are discounted. Trading in the
risky asset, the stock, is subject to proportional transaction
costs. At any time step $t$, a share of the stock can be bought for the ask price $S_t^a$ and sold for the bid price $S_t^b$, where $S_t^a\geq S_t^b>0$. We assume that $\processdef{S^a}{t}{0}{T}\in\mathcal{N}$ and $\processdef{S^b}{t}{0}{T}\in\mathcal{N}$.

The cost of creating a portfolio $x=(x^b,x^{s})\in\mathcal{L}^2_t$ at any time $t$ is
\begin{equation} \label{def:phi}
\phi_t(x)\coloneqq x^b+x_{+}^{s}S_t^a-x_{-}^{s}S_t^b,
\end{equation}
where $z_{+}\coloneqq\max\{z,0\} $ and $z_{-}\coloneqq-\min\{z,0\}$ for all $z\in\mathbb{R}$. The liquidation value of the portfolio $x$ is
$x^b+x_{+}^{s}S_t^b-x_{-}^{s}S_t^a=-\phi_t(-x).$
Define the \emph{solvency cone} $\mathcal{K}_t$ at any time $t$ as the collection of portfolios that can be liquidated into a nonnegative cash amount, in other words,
\[
 \mathcal{K}_t \coloneqq \big\{x\in\mathcal{L}^2_t:-\phi_t(-x)\ge0\big\} = \big\{(x^b,x^s)\in\mathcal{L}^2_t:x^b+x^{s}S_t^b\ge0, x^b+x^{s}S_t^a\ge0\big\}.
\]

A trading strategy $\processdef{y}{t}{-1}{T}$ is an adapted sequence of portfolios, where $y_{-1}\in\mathcal{L}^2_0$ denotes the initial endowment at time $0$, the portfolio $y_t\in\mathcal{L}^2_t$ is held between time steps $t$ and $t+1$ for every $t=0,\ldots,T-1$, and $y_T\in\mathcal{L}^2_T$ is the terminal portfolio created at time $T$. Denote the collection of trading strategies by~$\mathcal{N}^{2\prime}$, and define
\[
 \Delta y_t \coloneqq y_t - y_{t-1} \text{ for all }t\ge0.
\]
A trading strategy $y\in\mathcal{N}^{2\prime}$ is called \emph{self-financing} if $-\Delta y_t\in\mathcal{K}_t$ for all $t$. The collection of self-financing trading strategies is defined as
\[
 \Phi \coloneqq \left\{y\in\mathcal{N}^{2\prime}:-\Delta y_t\in\mathcal{K}_t\ \forall t\right\}.
\]
We will also frequently consider the class of trading strategies that start and end with zero holdings (and are not necessarily self-financing). This class of trading strategies is denoted by
 \[
 \Psi \coloneqq \left\{y\in\mathcal{N}^{2\prime}:y_{-1}=0,y_T=0\right\}.
\]

\subsection{Arbitrage and duality}

There is a connection between the absence of arbitrage and the existence of classes of objects that appear in the study of disutility minimisation problems. To this end, define
\begin{align}
\bar{\mathcal{P}} & \coloneqq\big\{ (\mathbb{Q},S):\mathbb{Q}\ll\mathbb{P},\thinspace S\text{ a }\mathbb{Q}\text{-martingale},\thinspace S_t^b\leq S_t\leq S_t^a\thinspace\forall t\big\}, \label{eq:def_Pbar} \\
\mathcal{P} & \coloneqq\big\{ (\mathbb{Q},S):\mathbb{Q}\sim\mathbb{P},\thinspace S\text{ a }\mathbb{Q}\text{-martingale},\thinspace S_t^b\leq S_t\leq S_t^a\thinspace\forall t\big\}. \nonumber
\end{align}
We shall refer to the elements of $\bar{\mathcal{P}}$ ($\mathcal{P}$)
as (\emph{equivalent}) \emph{martingale pairs}. Observe that $\mathcal{P}\subseteq\bar{\mathcal{P}}$.

The \emph{no-arbitrage condition} is equivalent to the existence of a martingale pair. The definition~\eqref{eq:NA} is consistent with that of \citet[Def.~1.6]{schachermayer2004fundamental} and equivalent, though formally different, to the notion of weak no-arbitrage introduced by \cite{kabanov2001harrison}. 

\begin{proposition}[{\citet[Theorem~1]{kabanov2001harrison}}]
 The no-arbitrage condition 
 \begin{equation} \label{eq:NA}
  \left\{y_T:y\in\Phi,y_{-1}=0\right\} \cap \left\{z\in\mathcal{L}^2_T:z\ge0\right\} = \{0\}
 \end{equation}
 holds if and only if~$\mathcal{P}\neq\emptyset$.
\end{proposition}

We will assume a stronger condition in this paper, namely \emph{robust no-arbitrage} \citep[Def.~1.9]{schachermayer2004fundamental}, which ensures existence of a solution to the disutility minimisation problem. It is characterised as follows.

\begin{proposition}[{\citet[Theorem~1.7]{schachermayer2004fundamental}}] \label{prop:RNA}
 The robust no-arbitrage condition holds if and only if there exists an equivalent martingale pair~$(\mathbb{Q},S)\in\mathcal{P}$ such that \begin{equation} \label{eq:RNA} S_t\in\relint[S^b_t,S^a_t]\text{ for all }t.\end{equation}
\end{proposition}

We assume throughout the rest of this paper that the model satisfies the robust no-arbitrage condition~\eqref{eq:RNA}. Here $\relint$ denotes \emph{relative interior}, so that
\[
 \relint[S^{b\omega}_t,S^{a\omega}_t] =
 \begin{cases}
  \big\{S^{b\omega}_t\big\} & \text{if } S^{b\omega}_t=S^{a\omega}_t, \\
  \big(S^{b\omega}_t,S^{a\omega}_t\big) & \text{if } S^{b\omega}_t<S^{a\omega}_t
 \end{cases}
\]
for all $t$ and $\omega\in\Omega$. 

The following notation will be useful when working with martingale pairs. For every $\mathbb{Q}\ll\mathbb{P}$, we write
\begin{equation}
\Lambda_t^{\mathbb{Q}}\coloneqq\mathbb{E}\left[\tfrac{d\mathbb{Q}}{d\mathbb{P}}\middle|\mathcal{F}_t\right]\text{ for all }t=0,\ldots,T,\label{eq:Lambda^Q}
\end{equation}
where $\frac{d\mathbb{Q}}{d\mathbb{P}}$ is the Radon-Nikodym density
of $\mathbb{Q}$ with respect to $\mathbb{P}$. As $\Omega$ is finite it follows that
 \begin{equation}\label{eq:Lambda^Q-as-fraction}
  \Lambda^{\mathbb{Q}\nu}_t = \tfrac{\mathbb{Q}(\nu)}{\mathbb{P}(\nu)} \text{ for all }t\text{ and }\nu\in\Omega_t.
 \end{equation}
Define also for all $t$
\[\Omega_t^{\mathbb{Q}}\coloneqq\{\nu\in\Omega_t:\mathbb{Q}(\nu)>0\}\]
as the collection of nodes in $\Omega_t$ with positive probability under $\mathbb{Q}$. Moreover, for every $t<T$ and $\mu\in\Omega_t^{\mathbb{Q}}$, denote the transition probability from $\mu$ to any successor node $\nu\in\mu^+$ by $q_{t+1}^\nu\coloneqq\frac{\mathbb{Q}(\nu)}{\mathbb{Q}(\mu)}$. Simple rearrangement of~\eqref{eq:Lambda^Q-as-fraction} then gives
 \begin{equation}\label{eq:Lambda^Q-and-transition-probabilities}
  \Lambda^{\mathbb{Q}\nu}_{t+1} = \tfrac{\mathbb{Q}(\mu)q^\nu_{t+1}}{\mathbb{P}(\mu)p^\nu_{t+1}} = \Lambda^{\mathbb{Q}\mu}_t\tfrac{q^\nu_{t+1}}{p^\nu_{t+1}} \text{ for all }t<T,\ \mu\in\Omega_t\text{ and }\nu\in\mu^+.
 \end{equation}

\subsection{Superhedging}
\label{sec:superhedging}

If the seller of a claim is completely risk-averse, then he would charge (at least) the \emph{superhedging price}, which is the lowest amount that the seller of a claim can charge that will allow him to sell the claim without taking any risk. Such prices are usually lower than the cost of replication \citep[see, for example][]{bensaid1992derivative}, and have been well studied for European options offering a payoff at a single expiration date; for a selection of contributions at a similar technical level to the current paper, see work by \citet{delbaen2002hedging,dempster2006asset,edirisinghe1993optimal,jouini_kallal1995a,kabanov2001harrison,lohne2014algorithm,perrakis1997derivative,roux2008options,roux2016american}.

In this subsection we generalise the theory slightly to the case of payment streams of the form $c\in\mathcal{N}^{2}$, consisting of sequences of (portfolio) payments $c_t=(c^b_t,c^s_t)$ to be made at all trading dates $t$. A trading strategy $y\in\mathcal{N}^{2\prime}$ is said to \emph{superhedge} such a payment stream $c$ if it allows a trader to deliver $c$ without risk, in other words, $y_T = 0$ and $-\Delta y_t - c_t \in\mathcal{K}_t$ for all $t$.

The \emph{seller's superhedging price} of the payment stream $c$ is defined as the smallest cash endowment that is sufficient to superhedge $c$, in other words,
\begin{align}
 \pi^a(c) &\coloneqq \inf \left\{ x\in\mathbb{R}:\exists y\in\mathcal{N}^{2\prime}\text{ superhedging }c\text{ with }y_0=(x,0)\right\}. \nonumber
\intertext{The \emph{buyer's superhedging price} of $c$ is defined as}
\pi^b(c) &\coloneqq \sup \left\{ x\in\mathbb{R}:\exists y\in\mathcal{N}^{2\prime}\text{ superhedging } {-c} \text{ with }y_0=(-x,0)\right\} \nonumber\\
&=-\pi^a(-c).\label{eq:buyer's-superhedging-price}
\end{align}
It is the largest cash amount that can be raised without risk by using the payoff of $c$ as collateral. The superhedging prices admit the following dual representation.

\begin{proposition} \label{prop:ask-price-rep}
 For every $c\in\mathcal{N}^{2}$ we have
 \begin{align}
  \pi^a(c) &= \sup_{(\mathbb{Q},S)\in\mathcal{P}} \total{t=0}{T}\mathbb{E}_{\mathbb{Q}} \big[c^b_t + c^s_tS_T\big] = \max_{(\mathbb{Q},S)\in\bar{\mathcal{P}}} \total{t=0}{T}\mathbb{E}_{\mathbb{Q}} \big[c^b_t + c^s_tS_T\big], \label{eq:pi-a-dual}\\
  \pi^b(c) &= \inf_{(\mathbb{Q},S)\in\mathcal{P}} \total{t=0}{T}\mathbb{E}_{\mathbb{Q}} \big[c^b_t + c^s_tS_T\big] = \min_{(\mathbb{Q},S)\in\bar{\mathcal{P}}} \total{t=0}{T}\mathbb{E}_{\mathbb{Q}} \big[c^b_t + c^s_tS_T\big].
 \end{align}
\end{proposition}

The collection of payment streams that can be superhedged from zero will play an important role in the next section. Proposition~\ref{prop:ask-price-rep} gives that
\begin{align}
 \mathcal{Z}
 &\coloneqq \left\{ c\in\mathcal{N}^2:\exists y\in\Psi\text{ superhedging }c\right\}  \label{eq:def-of-Z}\\
 &= \left\{ c\in\mathcal{N}^2:\pi^a(c)\le0\right\}\nonumber \\
 &= \left\{c\in\mathcal{N}^2:\ttotal{t=0}{T}\mathbb{E}_{\mathbb{Q}} \big[c^b_t + c^s_tS_T\big]\le0\ \forall(\mathbb{Q},S)\in\bar{\mathcal{P}}\right\}. \label{eq:def-of-Z-extra}
\end{align}
It is self-evident from the representation~\eqref{eq:def-of-Z-extra} that $\mathcal{Z}$ is a convex cone.

%
%


\section{Disutility minimisation problem}
\label{sec:Disutility-minimisation}

The ability to manage investments in such a way that their proceeds cover an investor's liabilities as well as possible, is of fundamental importance in financial economics, and has therefore been well studied in the literature; see, for example, the work of \citet{Pricing1997Davis,delbaen2002exponential,guasoni2002,hugonnier2005utility}. The purpose of this section is to formulate an optimal investment problem in the model with proportional transaction costs, which will form the basis of the indifference prices that will be studied in Section~\ref{sec:Indifference-pricing}.

Consider an investor who faces the liability of a (given) payment stream \mbox{$u\in\mathcal{N}^{2}$}. The investor can create a trading strategy 
$
 y \in \Psi
$
in cash and stock, and is additionally allowed to inject (invest) cash on every trading date in a given set $\mathcal{I}\subseteq\{0,\ldots,T\}$. At each trading date $t\in\mathcal{I}$, in order to manage his position, the investor needs to inject $\phi_t(\Delta y_t+u_t)$ in cash in order to manage his position. At trading dates $t\notin\mathcal{I}$, the investor is required to manage his position in a self-financing manner, in other words,
$\phi_t(\Delta y_t+u_t) \le 0$. Denote the number of elements of $\mathcal{I}$ by $\lvert\mathcal{I}\rvert$ and assume that $\lvert\mathcal{I}\rvert>0$, in other words, injection is allowed at least once. It is \emph{not} assumed that $T\in\mathcal{I}$.

The objective of the investor is to choose $y$ in such a way as to minimise the sum of expected disutility of the cash injections over all the trading dates in $\mathcal{I}$, using for each time step $t\in\mathcal{I}$ the risk-averse exponential disutility (regret) function 
\begin{align*}
 v_t(x) &\coloneqq e^{\alpha_t x} - 1 \text{ for all }x\in\mathbb{R}
\intertext{with deterministic risk aversion parameter $\alpha_t \in (0,\infty)$. Define for every $t\notin\mathcal{I}$}
 v_t(x) &\coloneqq \begin{cases}
                                      0 &\text{if } x\le 0,\\
                                      \infty &\text{if } x>0.
                                     \end{cases}
\end{align*}
The investor's objective can then be written as the unconstrained optimisation problem
\begin{equation}
\mbox{minimise }\total{t=0}{T}\mathbb{E}[v_t(\phi_t(\Delta y_t+u_t))]\mbox{ over }y\in\Psi.\label{eq:Problem 1}
\end{equation}
The value function $V$ of~\eqref{eq:Problem 1} is defined as
\begin{equation}
V(u)\coloneqq\inf_{y\in\Psi}\total{t=0}{T}\mathbb{E}[v_t(\phi_t(\Delta y_t+u_t))].\label{eq:def_mindisutility}
\end{equation}
The value of $V(u)$ is finite because $v_t$ is bounded from below for all $t$. 

\begin{remark}
In the special case where $\mathcal{I}=\{T\}$ and $u_t=0$ for all $t<T$, the problem~\eqref{eq:Problem 1} becomes
 \begin{equation}
\mbox{maximise }\mathbb{E}\big[1-e^{-\alpha_T(-\phi_T(-y_{T-1}+u_T))}\big]\mbox{ over }y\in\Psi, -\Delta y_t\in\mathcal{K}_t\ \forall t<T.
\end{equation}
Noting that $-\phi_T(-y_{T-1}+u_T)$ is the liquidation value of the portfolio $y_{T-1}-u_T$, this is the classical utility maximisation problem of an investor facing a liability of $u_T$ at time $T$.
\end{remark}

It is possible to rewrite~\eqref{eq:Problem 1} directly
in terms of the cash injections. This reduces the
dimensionality of the controlled process from two to one, and will aid in the study of the dual problem
in the next section. Combining the fact that $v_t$ is nondecreasing for all $t$ with~\eqref{eq:def-of-Z}, we obtain
\begin{align}
V(u)
&=\inf\left\{\ttotal{t=0}{T}\mathbb{E}[v_t(x_t)]:(x,y)\in\mathcal{N}\times\Psi, x_t\ge\phi_t(\Delta y_t+u_t)\ \forall t\right\} \nonumber \\
&=\inf\left\{\ttotal{t=0}{T}\mathbb{E}[v_t(x_t)]:(x,y)\in\mathcal{N}\times\Psi, -\Delta y_t-u_t+(x_t,0)\in\mathcal{K}_t\ \forall t\right\} \label{eq:V-presentation01-intermediate}  \\ 
&=\inf\left\{\ttotal{t=0}{T}\mathbb{E}[v_t(x_t)]:(x,y)\in\mathcal{N}\times\Psi, y\text{ superhedges }u-(x,0)\right\} \nonumber \\
&=\inf\left\{\ttotal{t=0}{T}\mathbb{E}[v_t(x_t)]:x\in\mathcal{N}, u-(x,0)\in\mathcal{Z}\right\} \nonumber \\
&= \inf_{x\in \mathcal{A}_{u}}\total{t=0}{T}\mathbb{E}[v_t(x_t)], \label{eq:V-presentation01}
\end{align}
where
\begin{equation} \label{eq:def:Au}
\mathcal{A}_{u}\coloneqq\{ x\in\mathcal{N}:u-(x,0)\in\mathcal{Z}\}.
\end{equation}
In conclusion, the problem~\eqref{eq:Problem 1} has the same value function as the optimisation problem
\begin{equation}
\mbox{minimise }\total{t=0}{T}\mathbb{E}[v_t(x_t)]\mbox{ over }x\in \mathcal{A}_{u}.\label{eq:Problem 1'}
\end{equation}

The following result summarises a few key properties of $V$.

\begin{theorem} \label{th:solution-exists}
 The function $V$ is convex and lower semicontinuous on $\mathcal{N}^2$, and the infima in~\eqref{eq:def_mindisutility} and~\eqref{eq:V-presentation01} are attained for every $u\in\mathcal{N}^2$.
\end{theorem}

This result means that the optimisation problems \eqref{eq:Problem 1} and \eqref{eq:Problem 1'} can be solved. Whilst the optimal trading strategy in \eqref{eq:Problem 1} is not unique, it will be shown as part of the construction (in Proposition \ref{prop:Optimal_Investment}) that the optimal cash injection strategy in~\eqref{eq:Problem 1'} is unique.

Shadow price processes are often considered in the utility optimisation problem under transaction costs; see \citet*{Czichowsky_Muhle-Karbe_Schachermayer2014} and
\citet{Kallsen_Muhle-Karbe2011}, for example. In the present work, they will play a role in the construction of optimisers for \eqref{eq:Problem 1}. A process $\hat{S}\in\mathcal{N}$
is called a \emph{shadow price process} for a given liability $u\in\mathcal{N}^2$ if $S^b_t\le \hat{S}_t\le S^a_t$ for all $t$, and the optimal disutility in the model with bid-ask spread $[S^{b},S^{a}]$ and in the friction-free model with price process $\hat{S}$ coincide, in other words, 
\begin{equation}
V(u)=\inf_{y\in\Psi}\total{t=0}{T}\mathbb{E}\big[v_{t}\big(\Delta y^b_{t}+u^b_{t} + (\Delta y^s_{t}+u^s_{t})\hat{S}_t\big)\big].\label{eq:EqualDisutilitty_ShadowPrice}
\end{equation}
The shadow price process depends on the given liability, bid-ask spread of the stock, and investor's risk preference. It will be shown in Sections \ref{sec:Solution-construction} and \ref{sec:optimal-injection-investment}, by explicit construction, that a shadow price $\hat{S}$ exists for any given liability, and that it corresponds to a friction-free model that is free of arbitrage, in other words, there exists a probability measure $\hat{\mathbb{Q}}\sim\mathbb{P}$ such that $(\hat{\mathbb{Q}},\hat{S})\in\mathcal{P}$.

The following result concludes this section. The formulation below matches its usage in the constructions in Section \ref{sec:optimal-injection-investment}; however item \ref{prop:shadow-price-properties:2} holds under more general conditions than stated here.

\begin{proposition} \label{prop:shadow-price-properties}
The following holds true for any liability $u\in\mathcal{N}^2$ and associated shadow price process $\hat{S}$:
\begin{enumerate}[(1)]
\item \label{prop:shadow-price-properties:1} Any trading strategy $\hat{y}\in\Psi$ that solves~\eqref{eq:Problem 1} in the market model with price process $\hat{S}$ and satisfies
\begin{equation}
(\Delta\hat{y}^s_{t}+u^s_{t})_+S^a_t - (\Delta\hat{y}^s_{t}+u^s_{t})_-S^b_t = (\Delta\hat{y}^s_{t}+u^s_{t})\hat{S}_t \text{ for all }t,\label{eq:TradeAtSpread}
\end{equation}
also solves~\eqref{eq:Problem 1} in the model with bid-ask spread $[S^b,S^a]$.

\item \label{prop:shadow-price-properties:2} Any trading strategy $\hat{y}\in\Psi$ solving~\eqref{eq:Problem 1} in the market model with price process $[S^b,S^a]$ solves~\eqref{eq:Problem 1} in the friction-free model with stock price process $\hat{S}$. Furthermore, if the friction-free model with stock price process $\hat{S}$ is free of arbitrage, then~$\hat{y}$ also satisfies \eqref{eq:TradeAtSpread}.
\end{enumerate}
\end{proposition}

Condition \eqref{eq:TradeAtSpread} can be formulated equivalently as
\[
 \big\{\Delta\hat{y}_{t}^{s}+u_{t}^{s}>0\big\}\subseteq\big\{\hat{S}_{t}=S_{t}^{a}\big\}\text{ and }\big\{\Delta\hat{y}_{t}^{s}+u_{t}^{s}<0\big\}\subseteq\big\{\hat{S}_{t}=S_{t}^{b}\big\}\text{ for all }t,
\]
in other words, a strategy $\hat{y}$ satisfying \eqref{eq:TradeAtSpread} trades only when $\hat{S}$ coincides with the bid and ask prices in the model with transaction costs. The proof of item \ref{prop:shadow-price-properties:2} depends on Proposition \ref{prop:Optimal_Investment} and, accordingly, appears in logical order after its proof in Appendix \ref{app:proofs}. 

\section{Dual formulation}
\label{sec:The-dual-problem}

It is possible to obtain a Lagrangian dual formulation for the optimisation problem~\eqref{eq:Problem 1'}. For every $u\in\mathcal{N}^2$, define the
Lagrangian $L_{u}:\mathcal{N}\times[0,\infty)\times\bar{\mathcal{P}}\rightarrow\mathbb{R}\cup\{\infty\} $
as 
\begin{equation}
L_{u}(x,\lambda,(\mathbb{Q},S))\coloneqq\total{t=0}{T}\big(\mathbb{E}[v_t(x_t)]+\lambda\mathbb{E}_{\mathbb{Q}}\big[u^b_t + u^s_tS_T - x_t\big] \big). \label{eq:Lagrangian (scalar)}
\end{equation}
The formulation of $L_u$ is motivated by an argument of \citet[(74)]{Schachermayer2002} (in the context of utility maximisation in incomplete market models without transaction costs). The coefficient of $\lambda$ encapsulates the constraints in~\eqref{eq:Problem 1'}; see~\eqref{eq:def-of-Z-extra}.

The following strong duality result holds.

\begin{theorem}
\label{thm:StrongDuality}For all $u\in\mathcal{N}^{2}$, we have
\begin{equation}
V(u)
= \inf_{x\in\mathcal{N}}\sup_{\lambda\geq0,(\mathbb{Q},S)\in\bar{\mathcal{P}}}L_{u}(x,\lambda,(\mathbb{Q},S)) = \sup_{\lambda\geq0,(\mathbb{Q},S)\in\bar{\mathcal{P}}}\inf_{x\in\mathcal{N}}L_{u}(x,\lambda,(\mathbb{Q},S)).\label{eq:StrongDuality}
\end{equation}
\end{theorem}

The strong duality established in Theorem~\ref{eq:V-ito-hat-lambda} suggests that further study of the \emph{dual problem} 
\begin{equation}
\mbox{maximise }\inf_{x\in\mathcal{N}}L_{u}(x,\lambda,(\mathbb{Q},S))\mbox{ over }(\lambda,(\mathbb{Q},S))\in[0,\infty)\times\bar{\mathcal{P}}\label{eq:problem 1' dual}
\end{equation}
of~\eqref{eq:Problem 1'} would be profitable. It turns out that there is an explicit formula for the value of the inner optimisation problem over $x$. Note that in this paper we adopt the convention $0\ln0=0$.

\begin{proposition}
\label{prop:infL_dual}For any $u\in\mathcal{N}^2$ and $(\lambda,(\mathbb{Q},S))\in[0,\infty)\in\bar{\mathcal{P}}$,
we have 
\begin{multline}
\inf_{x\in\mathcal{N}}L_{u}(x,\lambda,(\mathbb{Q},S))
= - \total{t\in\mathcal{I}}{} \tfrac{\lambda}{\alpha_t}\mathbb{E}_{\mathbb{Q}}\big[\ln\Lambda^\mathbb{Q}_t\big] + \lambda\total{t=0}{T}\mathbb{E}_{\mathbb{Q}}\big[u^b_t + u^s_tS_T\big] \\ - \total{t\in\mathcal{I}}{}\tfrac{\lambda}{\alpha_t}\big(\ln\tfrac{\lambda}{\alpha_t}-1\big) - \lvert\mathcal{I}\rvert. \label{eq:prop:infL_dual:formula}
\end{multline}
\end{proposition}

The joint dependence on $\lambda$ and $(\mathbb{Q},S)$ in~\eqref{eq:prop:infL_dual:formula} is very simple: the two terms on the right hand side that depend on $(\mathbb{Q},S)$, both contain $\lambda$ only as a nonnegative linear coefficient. This suggests that it should be possible to rewrite the outer maximisation in the dual problem~\eqref{eq:problem 1' dual} as a two-step maximisation, in other words, maximising first over $(\mathbb{Q},S)$, and then over $\lambda$. 

The solution to the first step, maximisation over $(\mathbb{Q},S)$, will be the subject of Section~\ref{sec:Solution-construction}. In the remainder of this section, we introduce some notation in order to capture the two-step nature of the maximisation, and then show that the maximisation problem over $\lambda$ has a unique closed form solution. To this end, define
\begin{align}
H((\mathbb{Q},S);X) &\coloneqq \total{t\in\mathcal{I}}{}\tfrac{1}{\alpha_t}\mathbb{E}_{\mathbb{Q}}\big[\ln\Lambda^\mathbb{Q}_t\big]+\mathbb{E}_{\mathbb{Q}}\big[X^b+X^sS_T\big] \text{ for all } (\mathbb{Q},S)\in\bar{\mathcal{P}},\label{eq:def_HI} \\
K(X) &\coloneqq \inf_{(\mathbb{Q},S)\in\bar{\mathcal{P}}} H((\mathbb{Q},S);X).\label{eq:K_I(X)}
\end{align}
for any $X\in\mathcal{L}^2_T$. Notice that $K(X)$ is finite because the
values of the mapping $x\mapsto x\ln x$ are finite and bounded from below on
$[0,\infty)$. Combining this notation with~\eqref{eq:StrongDuality} and~\eqref{eq:prop:infL_dual:formula}, we obtain, for all $u\in\mathcal{N}^2$,
\begin{align}
 V(u) &= \sup_{\lambda\ge0}\left\{-\lambda K\big(-\ttotal{t=0}{T}u_t\big) - \ttotal{t\in\mathcal{I}}{}\tfrac{\lambda}{\alpha_t}\big(\ln\tfrac{\lambda}{\alpha_t}-1\big)-\lvert\mathcal{I}\rvert\right\} \nonumber \\
 &= -\inf_{\lambda\ge0}\left\{\lambda K\big(-\ttotal{t=0}{T}u_t\big) + \ttotal{t\in\mathcal{I}}{}\tfrac{\lambda}{\alpha_t}\big(\ln\tfrac{\lambda}{\alpha_t}-1\big)\right\} - \lvert\mathcal{I}\rvert. \label{eq:V-ito-lambda}
\end{align}

The following result concludes this section. 

\begin{theorem}
\label{thm:min_disutility}
For any $u\in\mathcal{N}^2$, the minimal disutility is
\begin{equation} \label{eq:V-ito-hat-lambda}
V(u)=\hat{\lambda}_u\total{t\in\mathcal{I}}{}\tfrac{1}{\alpha_t}-\lvert\mathcal{I}\rvert,
\end{equation}
where 
\begin{equation}
\hat{\lambda}_u\coloneqq\exp\left\{\left(\ttotal{t\in\mathcal{I}}{}\tfrac{\ln\alpha_t}{\alpha_t}-K\big(-\ttotal{t=0}{T}u_t\big)\right)\middle/\ttotal{t\in\mathcal{I}}{}\tfrac{1}{\alpha_t}\right\}>0\label{eq:def_lambda()}
\end{equation}
is the unique value attaining the infimum in~\eqref{eq:V-ito-lambda}.
\end{theorem}

Note that Theorem~\ref{thm:min_disutility} implies that $\hat{\lambda}_u$, and hence $V(u)$, depend on $u$ only through $\ttotal{t=0}{T}u_t$. This is perhaps surprising in view of the definition~\eqref{eq:def_mindisutility} of $V(u)$. The reason for this comes from the dual formulation and the nature of the dual objects in models with proportional transaction costs: for example, it can be seen in~\eqref{eq:def-of-Z-extra} that whether a payment stream can be superhedged from zero depends only on its total payoff. This is the reason why the Lagrangian $L_u$ depends linearly on $\ttotal{t=0}{T}u_t$, which in turn leads directly into the dual formulation of $V(u)$.

\section{Indifference pricing}
\label{sec:Indifference-pricing}

In this section we consider an investor trading in cash and shares and who is entitled to receive a given portfolio $w_t\in\mathcal{L}_t^{2}$ at each time step $t$. We refer to the payment stream $w\in\mathcal{N}^2$ as the \emph{endowment} of the investor (though it may in fact represent a liability if negative). The minimal disutility of the investor in this situation is~$V(-w)$.

Indifference pricing provides a way for such an investor to determine the value of derivatives, or payment streams. We will introduce disutility indifference prices for the seller and buyer of a payment stream $c\in\mathcal{N}^{2}$. Consider the situation where the investor is selling the payment stream $c$. He receives a single payment of $\delta\in\mathbb{R}$ in cash at time $0$, and then delivers the portfolio $c_t$ at each time step $t$. After selling $c$, the investor's minimum disutility becomes $V(c-\delta\mathbbm{1}-w)$, where the process $\mathbbm{1}=(\mathbbm{1}_t)_{t=0}^T$ is defined as
\[
\mathbbm{1}_t\coloneqq\begin{cases}
(1,0) & \text{if }t=0,\\
(0,0) & \text{if }t>0.
\end{cases}
\]
The \emph{seller's disutility indifference price} $\pi^{ai}(c;w)$ of $c$ is defined as the lowest price for which he could sell $c$ without increasing his minimal disutility, in other words,
\begin{equation}
\pi^{ai}(c;w)\coloneqq\inf\{\delta\in\mathbb{R}: V(c-\delta\mathbbm{1}-w)\leq V(-w)\} .\label{eq:seller_indifferenceprice}
\end{equation}

The \emph{buyer's disutility indifference price} $\pi^{bi}(c;w)$ is similarly defined as the highest price at which the investor could buy the payment stream (and receive $c_t$ at each time step
$t$) without increasing his minimal disutility, in other words,
\begin{align}
\pi^{bi}(c;w) & \coloneqq\sup\{ \delta\in\mathbb{R}: V(-c+\delta\mathbbm{1}-w)\leq V(-w)\} \nonumber \\
 & =-\inf\{ \delta\in\mathbb{R}: V(-c-\delta\mathbbm{1}-w)\leq V(-w)\} =-\pi^{ai}(-c;w).\label{eq:pibF=-piaF}
\end{align}

The following theorem gives formulae for computing the buyer's and seller's indifference prices. These pricing formulae resemble existing formulae for utility indifference prices in friction-free models under exponential utility, in particular those obtained by \citet{delbaen2002exponential} and \citet{rouge2000pricing} in general continuous-time market models without transaction costs, and \citet{musiela2004valuation} in a discrete time friction-free model with a non-traded asset.

Observe that, to determine the buyer's and seller's indifference prices of a payment stream, it is sufficient to be able to determine the value of $K$ for three different random variables.

\begin{theorem}
\label{thm:indifferenceprices_expdisutility}For any $c,w\in\mathcal{N}^{2}$,
we have 
\begin{align}
\pi^{ai}(c;w) & =K\left(\ttotal{t=0}{T}w_t\right)-K\left(\ttotal{t=0}{T}(w_t-c_t)\right), \label{eq:formula:piFai}\\
\pi^{bi}(c;w) & =K\left(\ttotal{t=0}{T}(w_t+c_t)\right)-K\left(\ttotal{t=0}{T}w_t\right). \label{eq:formula:piFbi}
\end{align}
\end{theorem}

The following one-step toy model demonstrates the calculation of the indifference prices using~\eqref{eq:formula:piFai} and~\eqref{eq:formula:piFbi}.

\begin{example} \label{ex:one-step-model}
Let $T=1$ and $\Omega=\{\mathrm{u},\mathrm{d}\}$, and take any probability measure $\mathbb{P}$ with $p\coloneqq\mathbb{P}(\mathrm{u})\in(0,1)$. Suppose furthermore that the bid and ask prices in this model satisfy
\begin{equation} \label{eq:example:1}
 S^{b\mathrm{d}}_1\le S^{a\mathrm{d}}_1 < S^b_0 = \bar{S}_0 = S^a_0 < S^{b\mathrm{u}}_1\le S^{a\mathrm{u}}_1.
\end{equation}
The mid-price process $\bar{S}=(\bar{S}_0,\bar{S}_1)\in\mathcal{N}$ with $\bar{S}_1\coloneqq\tfrac{1}{2}(S^a_1 + S^b_1)$ together with the unique probability measure $\mathbb{Q}$ with $\mathbb{Q}(\mathrm{u})=\tfrac{\bar{S}_0-\bar{S}_1^{\mathrm{d}}}{\bar{S}_1^{\mathrm{u}} - \bar{S}_1^{\mathrm{d}}}$ satisfies the robust no-arbitrage condition in Proposition~\ref{prop:RNA}. 

Every probability measure $\mathbb{Q}$ in this model can be characterised uniquely by $\mathbb{Q}(\mathrm{u})$. It follows from~\eqref{eq:example:1} that
\begin{align*}
\mathcal{Q}&\coloneqq  \left\{ \mathbb{Q}(\mathrm{u}):(\mathbb{Q},S)\in\bar{\mathcal{P}}\right\} \\
&=  \big\{ q\in[0,1]:qx^{\mathrm{u}}+(1-q)x^{\mathrm{d}}=\bar{S}_0\text{ for some }x^{\mathrm{u}}\in\big[S^{b\mathrm{u}}_1,S^{a\mathrm{u}}_1\big],x^{\mathrm{d}}\in\big[S^{b\mathrm{d}}_1,S^{a\mathrm{d}}_1\big]\big\} \\
&=\left\{ \tfrac{\bar{S}_0-x^{\mathrm{d}}}{x^{\mathrm{u}}-x^{\mathrm{d}}}:x^{\mathrm{u}}\in\big[S^{b\mathrm{u}}_1,S^{a\mathrm{u}}_1\big],x^{\mathrm{d}}\in\big[S^{b\mathrm{d}}_1,S^{a\mathrm{d}}_1\big]\right\} \\
&=\left[\tfrac{\bar{S}_0-S_1^{a\mathrm{d}}}{S_1^{a\mathrm{u}}-S_1^{a\mathrm{d}}},\tfrac{\bar{S}_0-S_1^{b\mathrm{d}}}{S_1^{b\mathrm{u}}-S_1^{b\mathrm{d}}}\right] \eqqcolon [q_{\min},q_{\max}].
\end{align*}
Observe in particular that $\mathcal{Q}\subset(0,1)$.

Let $\mathcal{I}\coloneqq\{0,1\}$ and $\alpha_0=\alpha_1=\alpha>0$, and set the investor's endowment $w=(w_0,w_1)\in\mathcal{N}^2$ to be zero, in other words, $w_0=w_1=(0,0)$. It is possible to derive explicit formulae for the buyer's and seller's disutility indifference prices of a derivative security with cash payoff $D\in\mathcal{L}_1$ at time~$1$. This corresponds to the payment stream $c=(c_0,c_1)\in\mathcal{N}^2$ satisfying $c_0=(0,0)$ and $c_1 = (D,0)$. From~\eqref{eq:formula:piFai} and~\eqref{eq:formula:piFbi}, these prices involve terms of the form $K((Y,0))$ where $Y\in\mathcal{L}_1$. For any such $Y$ and any $(\mathbb{Q},S)\in\bar{\mathcal{P}}$, combining~\eqref{eq:def_HI} and~\eqref{eq:Lambda^Q-as-fraction} gives
\[
H((\mathbb{Q},S);(Y,0)) 
= \tfrac{1}{\alpha}\mathbb{E}_{\mathbb{Q}}\big[\ln\Lambda_1^{\mathbb{Q}}\big] + \mathbb{E}_{\mathbb{Q}}[Y] = f_Y(\mathbb{Q}(\mathrm{u})),
\]
where
\[
f_Y(q)\coloneqq\tfrac{1}{\alpha}\left(q\ln\tfrac{q}{p}+(1-q)\ln\tfrac{1-q}{1-p}\right) + qY^{\mathrm{u}}+(1-q)Y^{\mathrm{d}}\text{ for all }q\in[0,1].
\]
The function $f_Y$ is continuous and convex on $[0,1]$, and that it reaches its minimum at
\[
 \hat{q}_Y \coloneqq \left.pe^{-\alpha Y^{\mathrm{u}}}\middle/\left(pe^{-\alpha Y^{\mathrm{u}}}+(1-p)e^{-\alpha Y^{\mathrm{d}}}\right)\right. \in (0,1).
\]
It then follows from~\eqref{eq:K_I(X)} that
\begin{equation} \label{eq:example:2}
K((Y,0)) = \inf_{(\mathbb{Q},S)\in\bar{\mathcal{P}}} H((\mathbb{Q},S);(Y,0)) = \inf_{q\in[q_{\min},q_{\max}]}f_Y(q) = f_Y(q_Y),
\end{equation}
where
$
 q_Y \coloneqq \min\{\max\{\hat{q}_Y,q_{\min}\},q_{\max}\}.
$
After substituting~\eqref{eq:example:2} into~\eqref{eq:formula:piFai} and~\eqref{eq:formula:piFbi}, the buyer's and seller's disutility indifference prices of $c$ become
\begin{align*}
\pi^{ai}(c;0) & = K((0,0))-K((-D,0)) = f_0(q_0) - f_{-D}(q_{-D}), \\
\pi^{bi}(c;0) & = K((D,0))-K((0,0)) = f_D(q_D) - f_0(q_0).
\end{align*}
\end{example}

We conclude this section by presenting a key property of disutility indifference prices, namely that they produce smaller bid-ask intervals than superhedging prices.

\begin{theorem}\label{th:bidaskspread}
 We have for any $c,w\in\mathcal{N}^2$ that
 \[
  \pi^b(c) \le \pi^{bi}(c;w) \le \pi^{ai}(c;w) \le \pi^a(c).
 \]
 Moreover, the mapping $u\mapsto\pi^{ai}(u;w)$ is convex, and $u\mapsto\pi^{bi}(u;w)$ is concave.
\end{theorem}

\section{Solving the dual problem}
\label{sec:Solution-construction}

It was shown in Section~\ref{sec:The-dual-problem} that solving the disutility minimisation problem~\eqref{eq:Problem 1} amounts to computing the value of $K(X)$, defined in~\eqref{eq:K_I(X)}, for suitably chosen $X$ (see Theorem~\ref{thm:min_disutility}). The same holds true for determining the buyer's and seller's indifference prices in Section~\ref{sec:Indifference-pricing} (see Theorem~\ref{thm:indifferenceprices_expdisutility}). In this section, we propose a dynamic procedure for determining $K(X)$ for any $X\in\mathcal{L}^2_T$. We also present a dynamic procedure for constructing a pair $(\hat{\mathbb{Q}},\hat{S})\in\mathcal{P}$ such that 
\begin{equation} \label{eq:Solution-construction}
K(X) = H((\hat{\mathbb{Q}},\hat{S});X) = \total{t\in\mathcal{I}}{}\tfrac{1}{\alpha_t}\mathbb{E}_{\hat{\mathbb{Q}}}\big[\ln\Lambda_t^{\hat{\mathbb{Q}}}\big]+\mathbb{E}_{\hat{\mathbb{Q}}}\big[X^b+X^s\hat{S}_T\big].
\end{equation} 

\begin{remark}
  The dynamic procedure can also be used to find the \emph{minimal entropy martingale measure} \citep{frittelli2000introduction,frittelli2000minimal}. This is the measure $\hat{\mathbb{Q}}$ satisfying
 \[
  K(0) = \mathbb{E}_{\hat{\mathbb{Q}}}\big[\ln\Lambda_T^{\hat{\mathbb{Q}}}\big] = \mathbb{E}\left[\tfrac{d\hat{\mathbb{Q}}}{d\mathbb{P}}\ln\tfrac{d\hat{\mathbb{Q}}}{d\mathbb{P}}\right],
 \]
 in the special case when $\mathcal{I}=\{T\}$ and there are no transaction costs (in other words, $\hat{S}=S^b=S^a$).
\end{remark}

The ability to construct a solution by dynamic programming follows from the following representation for $H$ in terms of transition probabilities. The notation
\begin{equation}\label{eq:def:a}
a_t\coloneqq\total{k\in\mathcal{I},k\geq t}{}\tfrac{1}{\alpha_k} \text{ for all }t
\end{equation}
will be used throughout the remainder of this paper for brevity.

\begin{proposition}
\label{prop:H_formula} For all $X\in\mathcal{L}_T^{2}$ and $(\mathbb{Q},S)\in\bar{\mathcal{P}}$, we have
\begin{multline} \label{eq:H_formula}
 H((\mathbb{Q},S);X)
 = \total{t=0}{T-1}a_{t+1}\total{\mu\in\Omega_t^{\mathbb{Q}}}{}\mathbb{Q}(\mu)\total{\nu\in\mu^{+}}{}q_{t+1}^\nu\ln\tfrac{q_{t+1}^\nu}{p_{t+1}^\nu} \\
 + \total{\mu\in\Omega_{T-1}^{\mathbb{Q}}}{}\mathbb{Q}(\mu)\total{\nu\in\mu^{+}}{}q_T^\nu\big(X^{b\nu} + X^{s\nu}S_T^\nu\big). 
\end{multline}
\end{proposition}

The representation in Proposition~\ref{prop:H_formula} suggests that it is possible to construct a sequence $(\hat{q}_t)_{t=1}^T$ of transition probabilities, from which then to assemble the probability measure $\hat{\mathbb{Q}}$. The following construction provides a sequence of auxiliary functions to achieve this aim.

\begin{construction} \label{constr:Jt}
For given $X\in\mathcal{L}_T^{2}$, construct two adapted sequences of random functions $\process{f}{t}{0}{T-1}$ and $\process{J}{t}{0}{T}$ by backward induction. Define $J_T:\Omega\times\mathbb{R}\rightarrow\mathbb{R}\cup\{\infty\}$ as
\begin{equation} \label{eq:def_JT}
J_T^\nu(x)\coloneqq\begin{cases}
X^{b\nu}+xX^{s\nu} & \text{if }x\in\big[S_T^{b\nu},S_T^{a\nu}\big],\\
\infty & \text{otherwise}.
\end{cases}
\end{equation}
for all $\nu\in\Omega_T$. For every $t<T$, assume that $J_{t+1}$ has already been constructed, and define
\begin{align}
  f_t^\mu(x) &\coloneqq \inf \Big\{\ttotal{\nu\in\mu^{+}}{}q^\nu\left(J_{t+1}^\nu(x^\nu) + a_{t+1} \ln\tfrac{q^\nu}{p_{t+1}^\nu}\right) \nonumber\\
 &~:q^\nu\in[0,1],x^\nu\in\dom J_{t+1}^\nu
 \thinspace\forall \nu\in\mu^{+},
  \ttotal{k=1}{m} q^\nu = 1, \ttotal{k=1}{m} q^\nu x^\nu = x
 \Big\},  \label{eq:def_ft}\\
J_t^\mu (x) &\coloneqq
\begin{cases}
 f_t^\mu(x) & \text{if }x\in\big[S_t^{b\nu},S_t^{a\nu}\big],\\
 \infty & \text{otherwise}.
\end{cases} \label{eq:def_Jt}
\end{align}
for all $\mu\in\Omega_t$ and $x\in\mathbb{R}$.
\end{construction}

The definition~\eqref{eq:def_ft} of $f^\nu_t$ is reminiscent of that of the convex hull of the collection $\{J^\nu_{t+1}\}_{\nu\in\mu^+}$ of convex functions, if the term involving the logarithm is disregarded \citep[cf.][Theorem~5.6]{rockafellar1997convex}. The following result summarises the main properties of $\process{J}{t}{0}{T}$, with some of the technical arguments of the generalised convex hull deferred to \ref{sec:conv-hull}. Recall that the $\sigma$-field $\mathcal{F}_0$ is trivial, and therefore $J_0$ is a deterministic function.

\begin{proposition} \label{prop:J-minimizes-H}
 Fix any $X\in\mathcal{L}_T^{2}$ and let $\process{J}{t}{0}{T}$ be the sequence of functions from Construction~\ref{constr:Jt}. Then for each $t$ and $\nu\in\Omega_t$, the function $J^\nu_t$ is convex, bounded from below, continuous on its closed effective domain $\dom J^\nu_t\subseteq[S_t^{b\nu},S_t^{a\nu}]$ and the infimum in~\eqref{eq:def_ft} is attained whenever it is finite. Moreover,
 \begin{equation} \label{eq:J0S0-formula}
  J_0(S_0) = \inf_{(\bar{\mathbb{Q}},\bar{S})\in\bar{\mathcal{P}},\bar{S}_0=S_0}H((\bar{\mathbb{Q}},\bar{S});X) \text{ for all }(\mathbb{Q},S)\in\bar{\mathcal{P}}.
\end{equation}
\end{proposition}

The following construction uses the sequence $\process{J}{t}{0}{T}$ of Construction~\ref{constr:Jt} to produce a pair $(\hat{\mathbb{Q}},\hat{S})$ satisfying~\eqref{eq:Solution-construction}. It will be shown in Theorem~\ref{thm:opt_(QS)=opt(qx)} below that this does indeed produe a solution to~\eqref{eq:K_I(X)}.

\begin{construction} \label{constr:optQS}
 For given $X\in\mathcal{L}_T^{2}$ and associated sequence $\process{J}{t}{0}{T}$ from Construction~\ref{constr:Jt}, construct a process $\hat{S}\in\mathcal{N}$ and a predictable process $\process{\hat{q}}{t}{1}{T}$ by induction, as follows. First, choose any $\hat{S}_0$ satisfying
 \begin{equation}
  J_0(\hat{S}_0) = \min_{x\in[S^b_0,S^a_0]} J_0(x). \label{eq:hatS-minimizes-J0}
 \end{equation}
 For each $t<T$ and $\mu\in\Omega_t$, assume that $\hat{S}_t^\mu\in[S^{b\mu}_t,S^{a\mu}_t]$ has already been defined, and choose $\hat{q}_{t+1}^\nu\in[0,1]$, $\hat{S}_{t+1}^\nu\in\big[S_{t+1}^{b\nu},S_{t+1}^{a\nu}\big]$ for all $\nu\in\mu^+$ such that
\begin{align}
    J_t^\mu(\hat{S}_t^\mu) &= \total{\nu\in\mu^{+}}{}\hat{q}_{t+1}^\nu\left(J_{t+1}^\nu(\hat{S}_{t+1}^\nu) + a_{t+1}\ln\tfrac{\hat{q}_{t+1}^\nu}{p_{t+1}^\nu}\right), \label{eq:QShat-and-Js}\\
    \hat{S}_t^\mu &= \total{\nu\in\mu^{+}}{}\hat{q}_{t+1}^\nu\hat{S}_{t+1}^\nu, \label{eq:Shat-martingale}\\
    1 &= \total{\nu\in\mu^{+}}{}\hat{q}_{t+1}^\nu. \label{eq:Shat-martingale:prob}
\end{align}
 Finally, define $\hat{\mathbb{Q}}:\mathcal{F}\rightarrow\mathbb{R}$ as
 $
  \hat{\mathbb{Q}}(A) \coloneqq \ttotal{\omega\in A}{} \tproduct{t=1}{T} \hat{q}^{\omega_t}_t$ for all $A\in\mathcal{F}$,
 where the value of the sum over an empty set is taken to be $0$.
\end{construction}

Construction~\ref{constr:optQS} produces a well-defined pair $(\hat{\mathbb{Q}},\hat{S})$. This is because the existence of $\hat{S}_0$ is assured by the continuity of $J_0$, and the infimum in~\eqref{eq:def_ft} is attained whenever finite. It also produces a solution to the optimization problem~\eqref{eq:K_I(X)}, as claimed at the start of the section.


\begin{theorem}
\label{thm:opt_(QS)=opt(qx)}
For $X\in\mathcal{L}_T^{2}$ given, let $\process{J}{t}{0}{T}$ and $(\hat{\mathbb{Q}},\hat{S})=(\hat{\mathbb{Q}},\hat{S})$ be given by Constructions~\ref{constr:Jt} and~\ref{constr:optQS}. Then $(\hat{\mathbb{Q}},\hat{S})\in\mathcal{P}$ is a minimiser in \eqref{eq:K_I(X)} and
\begin{align*}
 K(X) 
 &= J_0(\hat{S}_0) = \min_{x\in[S^b_0,S^a_0]} J_0(x) \\
 &= H((\hat{\mathbb{Q}},\hat{S});X) = \min_{(\mathbb{Q},S)\in\bar{\mathcal{P}}} H((\mathbb{Q},S);X) \\
 &= \min_{(\mathbb{Q},S)\in\bar{\mathcal{P}}}\left(\ttotal{t\in\mathcal{I}}{}\tfrac{1}{\alpha_t}\mathbb{E}_{\mathbb{Q}}\big[\ln\Lambda^\mathbb{Q}_t\big]+\mathbb{E}_{\mathbb{Q}}\big[X^b+X^sS_T\big]\right).
\end{align*}
Moreover, the probability measure $\hat{\mathbb{Q}}$ is unique on nodes at times in $\mathcal{I}$, in the sense that if $(\mathbb{Q},S)\in\mathcal{P}$ is any other pair produced by Construction~\ref{constr:optQS}, then 
\begin{equation} \label{eq:Qhat-uniqueness}
\hat{\mathbb{Q}}(\nu)=\mathbb{Q}(\nu) \text{ for all }t\in\mathcal{I}\text{ and } \nu\in\Omega_t. 
\end{equation}
\end{theorem}

The property \eqref{eq:Qhat-uniqueness} ensures that $\hat{\mathbb{Q}}$ is unique as long as the $\sigma$-field generated by $\{\nu\in\Omega_t:t\in\mathcal{I}\}$ is $2^\Omega$. This holds true, for example, if $T\in\mathcal{I}.$ However, the pair $(\hat{\mathbb{Q}},\hat{S})$ is not unique in general, because the solutions to~\eqref{eq:hatS-minimizes-J0} and~\eqref{eq:QShat-and-Js}--\eqref{eq:Shat-martingale:prob} might not be unique. Nevertheless, the property \eqref{eq:Qhat-uniqueness} is sufficient to ensure the uniqueness of the optimal injection strategy, which will be considered in the next section.

\section{Optimal injection and investment}
\label{sec:optimal-injection-investment}

The optimal injection in \eqref{eq:Problem 1'}, and hence the optimal trading strategy in \eqref{eq:Problem 1}, can be obtained by means of the dual formulation of Section \ref{sec:The-dual-problem} and the constructions in Section \ref{sec:Solution-construction}.

The following result gives an explicit formula for the optimal injection strategy. It is consistent with Corollary~3.4 of \citet*{Kallsen_Muhle-Karbe2011} (obtained in a slightly different setting).

\begin{proposition}
\label{prop:Optimal_Investment} For any $u\in\mathcal{N}^2$, let $(\hat{\mathbb{Q}},\hat{S})$ be as in Theorem~\ref{thm:opt_(QS)=opt(qx)} for $X=-\ttotal{t=0}{T}u_{t}$. Then the process $\hat{x}\in\mathcal{N}$ defined by
\begin{equation} \label{eq:xhat}
\hat{x}_{t}=\begin{cases}
\frac{1}{\alpha_{t}}\ln\frac{\hat{\lambda}_{u}\Lambda_{t}^{\hat{\mathbb{Q}}}}{\alpha_{t}} & \text{if }t\in\mathcal{I},\\
0 & \text{if }t\notin\mathcal{I},
\end{cases}
\end{equation}
where $\hat{\lambda}_u$ is given by \eqref{eq:def_lambda()}, is the unique minimiser in  \eqref{eq:Problem 1'}.
\end{proposition}

Observe from~\eqref{eq:K_I(X)} and \eqref{eq:def_lambda()} that
\begin{equation} \label{eq:hat-lambda-H}
 \hat{\lambda}_u = \exp\left\{\tfrac{1}{a_0}\left(\ttotal{t\in\mathcal{I}}{}\tfrac{\ln \alpha_s}{\alpha_s}-H((\hat{\mathbb{Q}},\hat{S});X)\right)\right\},
\end{equation}
where $H$ is defined in \eqref{eq:def_HI} and $a_0$ in \eqref{eq:def:a}. This leads to two important observations.

\begin{remark} \label{remark:PL}
Substituting \eqref{eq:def_HI} into \eqref{eq:hat-lambda-H}, the optimal P\&L (cash gain, negative injection) associated with an optimal trading strategy is
\[
-\total{t\in\mathcal{I}}{}\hat{x}_{t} = \total{t\in\mathcal{I}}{}\tfrac{1}{\alpha_{t}}\left(\mathbb{E}_{\hat{\mathbb{Q}}}\big[\ln\Lambda^{\hat{\mathbb{Q}}}_t\big]-\ln\Lambda_{t}^{\hat{\mathbb{Q}}}\right) - \total{t=0}{T}\mathbb{E}_{\hat{\mathbb{Q}}}[u_{t}^{b}+u_{t}^{s}\hat{S}_{T}].
\]
The second term on the right hand side arises naturally in the no-arbitrage pricing of the liability $u$; see Section~\ref{sec:superhedging}. The first term in this expression is effectively a profit that can be achieved from following this particular injection strategy (rather than any other). Taking the expected value of this term under the real-world probability~$\mathbb{P}$ gives that
\[
\total{t\in\mathcal{I}}{}\tfrac{1}{\alpha_{t}}\left(\mathbb{E}_{\hat{\mathbb{Q}}}\big[\ln\Lambda^{\hat{\mathbb{Q}}}_t\big]-\mathbb{E}\big[\ln\Lambda^{\hat{\mathbb{Q}}}_t\big]\right)=\total{t\in\mathcal{I}}{}\tfrac{1}{\alpha_{t}}\total{\omega\in\Omega}{}\left(\hat{\mathbb{Q}}\left(\omega\right)-\mathbb{P}\left(\omega\right)\right)\ln\tfrac{\hat{\mathbb{Q}}\left(\omega\right)}{\mathbb{P}\left(\omega\right)}\geq0.
\]
When $\hat{\mathbb{Q}}=\mathbb{P}$, then this term is zero, but whenever $\hat{\mathbb{Q}}$ is distinct from $\mathbb{P}$, there is some room for profit. The numerical
results in Example~\ref{exa:OptSol_Prob} supports this finding.
\end{remark}

\begin{remark} \label{remark:xhat}
 The optimal injection strategy can be constructed inductively by decomposing \eqref{eq:xhat} into transition probabilities and using Theorem \ref{thm:opt_(QS)=opt(qx)}. Taking the sequence $\process{J}{t}{0}{T}$ from Construction \ref{constr:Jt} with $X=-\ttotal{t=0}{T}u_t$ and pair $(\hat{\mathbb{Q}},\hat{S})$ from Construction \ref{constr:optQS}, one obtains
 \[
  \hat{\lambda}_u = \exp\left\{\tfrac{1}{a_0}\left[\ttotal{t\in\mathcal{I}}{}\tfrac{\ln \alpha_s}{\alpha_s}-J_0(\hat{S}_0)\right]\right\},
 \] 
 and then
   \[
\hat{x}_{t}=\begin{cases}
\frac{1}{\alpha_{t}}\ln\frac{\hat{\lambda}_{u}}{\alpha_{t}} & \text{if }t\in\mathcal{I}\cap\{0\},\\
\frac{1}{\alpha_{t}}\ln\frac{\hat{\lambda}_{u}}{\alpha_{t}} + \frac{1}{\alpha_{t}}\ttotal{s=0}{t-1}\ln\frac{\hat{q}_{s+1}}{p_{s+1}} & \text{if }t\in\mathcal{I}\setminus\{0\},\\
0 & \text{if }t\notin\mathcal{I}.
\end{cases}
  \]
\end{remark}

For any $u\in\mathcal{N}^2$, observe that the process $\hat{S}$ in the martingale pair $(\hat{\mathbb{Q}},\hat{S})$ of Theorem \ref{thm:opt_(QS)=opt(qx)} with $X=-\ttotal{t=0}{T}u_{t}$ is a shadow price process. It satisfies $S^b_t\le \hat{S}_t\le S^a_t$ for all $t$ by construction. Furthermore, the minimal disutility $V(u)$ in the friction-free model with stock price process $\hat{S}$ is exactly the same as in the model with bid-ask spread $[S^{b},S^{a}]$ by Theorem \ref{thm:min_disutility}, leading to \eqref{eq:EqualDisutilitty_ShadowPrice}. The pair $(\hat{\mathbb{Q}},\hat{S})$ also satisfies the claims in Theorem \ref{thm:opt_(QS)=opt(qx)} in the friction-free model, and hence the optimal injection strategy $\hat{x}$ in \eqref{eq:xhat} is the same as in the model with bid-ask spread. From Proposition \ref{prop:shadow-price-properties}, any optimal trading strategy $\hat{y}$ in the friction-free model is also optimal in the model with bid-ask spreads, provided that it satisfies \eqref{eq:TradeAtSpread} and injection is allowed at time $T$. This final observation leads to the following construction of the set of all optimal trading strategies.

\begin{construction}
\label{alg:OptStockPosition}
Assume that $u\in\mathcal{N}^2$ is given. For 
the sequence~$\process{J}{t}{0}{T}$ from Construction~\ref{constr:Jt} with $X=-\ttotal{t=0}{T}u_t$ and the pair~$(\hat{\mathbb{Q}},\hat{S})$ from Construction~\ref{constr:optQS}, construct a sequence of auxiliary sets $\process{\mathcal{W}}{t}{-1}{T}$ by induction, where 
\[\mathcal{W}_t\subset\mathcal{N}^{2\prime}_t\coloneqq\big\{\process{w}{k}{-1}{t}:w\in\mathcal{N}^{2\prime}\big\} \text{ for all }t,\]
and a set $\mathcal{Y}\subset\mathcal{N}^{2\prime}$.

Define $\mathcal{W}_{-1}\coloneqq\{0\}$. For each $t=0,\ldots,T-1$, let $\mathcal{W}_t$ be the collection of all processes $\process{w}{k}{-1}{t}\in\mathcal{N}^{2\prime}_t$ such that $\process{w}{k}{-1}{t-1}\in\mathcal{W}_{t-1}$ and the random variable $w_t\in\mathcal{L}^2_t$ solves on each node $\mu\in\Omega_t$ the system of equations
\begin{align}
 \Delta w^{s\mu}_t\hat{S}^\mu_t &= (\Delta w^s_t)_+S^{a\mu}_t - (\Delta w^s_t)_-S^{b\mu}_t,  \label{eq:OpStockConditionZ}\\
w^{b\mu}_t + w^{s\mu}_t\hat{S}^\nu_{t+1} &= -J^\nu_{t+1}(\hat{S}^\nu_{t+1}) - a_{t+1}\ln\tfrac{\hat{q}^{\nu}_{t+1}}{p^{\nu}_{t+1}} \text{ for all }\nu\in\mu^+, \label{eq:OpStockEquation}
\end{align}
where $a_{t+1}$ is given by \eqref{eq:def:a}. Finally, let $\mathcal{W}_T$ be the collection of all processes $\processdef{w}{t}{-1}{T}\in\mathcal{N}^{2\prime}_T=\mathcal{N}^{2\prime}$ such that $\process{w}{t}{-1}{T-1}\in\mathcal{W}_{T-1}$ and the random variable $w_T\in\mathcal{L}^2_T$ satisfies
\begin{align} \label{eq:OpStockConditions-T}
 w_T &= \total{t=0}{T} u_t, & \Delta w^s_T\hat{S}_T &= (\Delta w^s_T)_+S^a_T - (\Delta w^s_T)_-S^b_T.
\end{align}

Define $\mathcal{Y}$ to be the collection of all trading strategies $\hat{y}\in\mathcal{N}^{2\prime}$ constructed by induction from some $w\in\mathcal{W}_T$ as $\hat{y}_{-1}\coloneqq0$ and
 \begin{align}
\hat{y}^b_t &\coloneqq 
\begin{cases}
 \Delta w^b_0 + \hat{x}_0 - u^b_0 + J_0(\hat{S}_0) & \text{if }t=0,\\
 \hat{y}^b_{t-1}+\Delta w^b_t +\hat{x}_t- u^b_t - a_t\ln\tfrac{\hat{q}_t}{p_t} & \text{if }t>0,
\end{cases}\label{eq:constr:yhatb:Delta}\\
\hat{y}^s_t &\coloneqq 
\hat{y}^s_{t-1}+\Delta w^s_t - u^s_t \text{ for all }t\ge0.\label{eq:constr:yhats:Delta}
\end{align}
Here $\hat{x}\in\mathcal{N}$ is determined as in Remark \ref{remark:xhat}.
\end{construction}

Construction~\ref{alg:OptStockPosition} requires the system of equations \eqref{eq:OpStockConditionZ}--\eqref{eq:OpStockEquation} to be solved at every non-terminal node, and \eqref{eq:OpStockConditions-T} at each terminal node, in each case for two variables. Despite being the stock price process of an arbitrage-free model, the shadow price process~$\hat{S}$ can be degenerate (for example, under large proportional transaction costs it could be constant), which can lead to these systems of equations being underdetermined, and hence having many solutions. This is the reason why the construction produces a collection of processes, rather than a single strategy. In most practical applications (involving models with two or more successors at each non-terminal node and small to moderate transaction costs), the systems involve two or more equations, and hence the collections produced by this construction are very small. That they are not empty (and hence that the systems are well-determined) comes from the following result.

\begin{proposition} \label{prop:Constr-optimal-hedging}
 For given $u\in\mathcal{N}^2$, let $\mathcal{Y}$ be the collection of trading strategies from Construction~\ref{alg:OptStockPosition}. Then $\mathcal{Y}\neq\emptyset$ and every $\hat{y}\in\mathcal{Y}$ is a minimiser in \eqref{eq:Problem 1}.
\end{proposition}

In practice, the computational cost of constructing an optimal trading strategy $\hat{y}$ grows exponentially in the number of time steps, even in recombinant binary trees. The reason for this is that neither $\hat{x}$ nor $\hat{S}$ are generally recombinant processes, even when $\ttotal{t=0}{T}u_t$ is path-independent and the bid-ask spread $[S^b,S^a]$ is a recombinant process. However, it is very efficient for determining the trading strategy in particular scenarios of interest.

\section{Numerical examples}
\label{sec:numerical-examples}

Consider a friction-free binomial tree model with $T=52$ steps representing one year in real time with weekly rehedging, where the stock price $\processdef{S}{t}{0}{52}$ satisfies $S_0=100$ and
\[
 S_{t+1} =
 \begin{cases}
 e^{\sigma\sqrt{1/52}} S_t & \text{with probability }p,\\
 e^{-\sigma\sqrt{1/52}} S_t & \text{with probability }1-p
 \end{cases} 
\]
for all $t<52$. Here $\sigma=0.2$ is the annual volatility of the return on stock, and the model is assumed to have an annual effective interest rate of $r_e=0.02$. Define the bid and ask prices of the stock as
\begin{align*}
S_t^a & \coloneqq(1+k)S_t, &
S_t^b & \coloneqq(1-k)S_t
\end{align*}
for all $t>0$, where $k$ is the proportional transaction cost parameter. We assume that there are no transaction costs at time $0$, in other words $S_0^a\coloneqq S_0^b\coloneqq S_0=100$.

The numerical results in this section have been obtained by applying the approximation methods introduced in Section~\ref{subsec:numerical-approximation} for the generalised convex hull. Each of these methods allow us to construct a sequence of random piecewise linear functions approximating the sequence $\process{J}{t}{0}{52}$ of Construction~\ref{constr:Jt}, starting from the final value $J_{52}$. This leads naturally to an approximation for $K$ via Theorem~\ref{thm:opt_(QS)=opt(qx)}, and $\pi^{ai}(c;w)$ and $\pi^{bi}(c;w)$ via Theorem~\ref{thm:indifferenceprices_expdisutility}. Superhedging bid and ask prices are also provided for the purposes of comparison, calculated using methods previously reported by \citet{roux2008options}.

It is assumed throughout this section that the investor's endowment is $w=0$, and that the risk aversion coefficient is constant, in other words, $\alpha_t=\alpha$ for all $t\in\mathcal{I}$. Consider a call option with expiry one year, strike $100$ and physical delivery (based on the underlying). This corresponds to the payment stream $\processdef{C}{t}{0}{52}$ where $C_t=0$ for all $t<52$ and
  \begin{align*}
   C_{52} &= (-100,1)\mathbbm{1}_{\{S_{52}>100\}}.
  \end{align*}
  
\begin{example} \label{ex:accuracy-numerical-approximation}
  Table~\ref{tab:accuracy-numerical-approximation} contains approximate indifference prices for the seller and buyer of the call option in the case where $p=0.5$, $k=0.005$, $\mathcal{I}=\{0,\ldots,52\}$ and $\alpha=0.1$, as computed by both the upper and lower approximation methods described in Section~\ref{subsec:numerical-approximation}. In each case, the approximation is obtained by dividing each (discounted) bid-ask interval into $n$ subintervals of equal length.
  \begin{table}
  \caption{Indifference prices by approximation method (Example~\ref{ex:accuracy-numerical-approximation})}
  \label{tab:accuracy-numerical-approximation}
  \begin{center}
  \begin{tabular}{|r|*{6}{d{4}}|}
   \hline
   $n$                                 & \multicolumn{1}{r}{20}      & \multicolumn{1}{r}{50}      & \multicolumn{1}{r}{100}      & \multicolumn{1}{r}{150}    & \multicolumn{1}{r}{200} & \multicolumn{1}{r|}{300} \\
   \hline
   \multicolumn{7}{|c|}{Upper approximation method} \\
   \hline
   $\pi^{bi}(C;0)$   & 8.5759& 8.5673& 8.5658& 8.5655& 8.5654& 8.5654 \\
   $\pi^{ai}(C;0)$   & 9.1596& 9.1672& 9.1684& 9.1687& 9.1687& 9.1688 \\
   \hline
   \multicolumn{7}{|c|}{Lower approximation method} \\
   \hline
   $\pi^{bi}(C;0)$   & 8.4974& 8.5533& 8.5633& 8.5647& 8.5652& 8.5653 \\
   $\pi^{ai}(C;0)$   & 9.2357& 9.1797& 9.171& 9.1692& 9.1690& 9.1690 \\
   \hline   
  \end{tabular}
  \end{center}
  \end{table}
  
  The results from the two approximation methods are consistent converge to the same limit, but the upper approximation converges much faster than the lower approximation. The results suggest that taking $n=150$ results in accuracy up to 3 decimal places, which is perfectly adequate for graphical representation.
  
  The indifference pricing spread (between the seller's and buyer's indifference prices) is considerably smaller than the (superhedging) bid-ask spread; note that the ask and bid prices in this case are $\pi^a(C)=10.4788$ and $\pi^b(C)=6.9694$.
\end{example}

  Different possibilities for the set $\mathcal{I}$ of dates on which injection is allowed will be considered below. The case $\mathcal{I}=\{52\}$, in particular, corresponds to the classical utility indifference pricing framework, where the cash injection at time $52$ reflects the hedging shortfall at the expiration date of the option under exponential utility.

\begin{example} \label{ex:alpha-k}
  Figure \ref{fig:ex:alpha-k} illustrates seller's and buyer's indifference prices for a range of values of the risk aversion coefficient $\alpha$ and transaction cost parameter $k$ in the case where $p=0.5$. Observe that the indifference pricing spread (between the seller's and buyer's indifference prices) is smaller for $\mathcal{I}=\{0,\ldots,52\}$ than $\mathcal{I}=\{52\}$. This is because being able to inject cash at different time steps introduces considerable flexibility, which in turn results in decreased hedging costs. 
 
 As seen in part (a), indifference pricing spreads increase as $\alpha$ increases. The indifference pricing spread remains well within the superhedging bid-ask spread for a large range of values of $\alpha$.

   Indifference pricing spreads increase with $k$, the intuitive reason being that increased transaction costs results in increased trading costs. This is illustrated in part (b). Observe finally that the indifference pricing spreads remain well within the superhedging bid-ask spread for all values of $k$, and also expand slower as $k$ increases.

\begin{figure}
\begin{center}
\begin{tikzpicture}

    \pgfplotstableread{alpha-full.dat} {\fullprices}
   \pgfplotstableread{alpha-util.dat} {\utilprices}
   \pgfplotstableread{bid-ask-0.005.dat} {\superprices}
   {\scriptsize
	\begin{axis} [name=plot1, anchor=north, height=\plotHeight,width=\plotWidth, xmin=0, xmax=2, ymin=0, ymax=15, xlabel near ticks, xlabel = Risk aversion $\alpha$, grid = both]
        \addplot[dashed, color1] 
            table [x={riskpref}, y={ask}] {\utilprices};
        \addplot[color1] 
            table [x={riskpref}, y={bid}] {\utilprices};
 
        \addplot[dashed, color2] 
            table [x={riskpref}, y={ask}] {\fullprices};
        \addplot[color2] 
            table [x={riskpref}, y={bid}] {\fullprices};

        \addplot[dashed, color4] 
            table [x={riskpref}, y={ask}] {\superprices};
        \addplot[color4] 
            table [x={riskpref}, y={bid}] {\superprices};
    \end{axis}
	}
	\node[align=center,anchor=north] at ($(plot1.north)+(0,0.5\vSpace)$) {(a) Indifference prices, $k=0.5\%$};

    \pgfplotstableread{k-full.dat} {\fullprices}
   \pgfplotstableread{k-util.dat} {\utilprices}
   \pgfplotstableread{k-bid-ask.dat} {\superprices}
	{\scriptsize
	\begin{axis} [name=plot2, anchor=west, height=\plotHeight,width=\plotWidth, at={($(plot1.east)+(\hSpace,0)$),}, grid=both, legend style={at={(-0.5\hSpace,-\vSpace)}, anchor=north, legend columns=4}, xmin=0, xmax=0.02, ymin=0, ymax=15, xlabel near ticks, xlabel = Transaction cost $k$]
	
	        \addlegendimage{empty legend}
        \addlegendentry{$\mathcal{I}=\{52\}$:}
        \addplot[dashed, color1] 
            table [x={k}, y={ask}] {\utilprices};
            \addlegendentry{$\pi^{ai}(C;0)$}
        \addplot[color1] 
            table [x={k}, y={bid}] {\utilprices};
            \addlegendentry{$\pi^{bi}(C;0)$}

        \addlegendimage{empty legend}
        \addlegendentry{}
        \addlegendimage{empty legend}
        \addlegendentry{$\mathcal{I}=\{0,\ldots,52\}$:}
        \addplot[dashed, color2] 
            table [x={k}, y={ask}] {\fullprices};
            \addlegendentry{$\pi^{ai}(C;0)$}
        \addplot[color2] 
            table [x={k}, y={bid}] {\fullprices};
            \addlegendentry{$\pi^{bi}(C;0)$}
        
        \addlegendimage{empty legend}
        \addlegendentry{}
        \addlegendimage{empty legend}
        \addlegendentry{Superhedging:}
        \addplot[dashed, color4] 
            table [x={k}, y={ask}] {\superprices};
            \addlegendentry{$\pi^a(C)$}
        \addplot[color4] 
            table [x={k}, y={bid}] {\superprices};
            \addlegendentry{$\pi^b(C)$}

	\end{axis}
	}
	\node[align=center,anchor=north] at ($(plot2.north)+(0,0.5\vSpace)$) {(b) Indifference prices, $\alpha=0.1$};
\end{tikzpicture}
\end{center}
\caption{Indifference prices, transaction costs and risk aversion (Example \ref{ex:alpha-k})}
\label{fig:ex:alpha-k}
\end{figure}
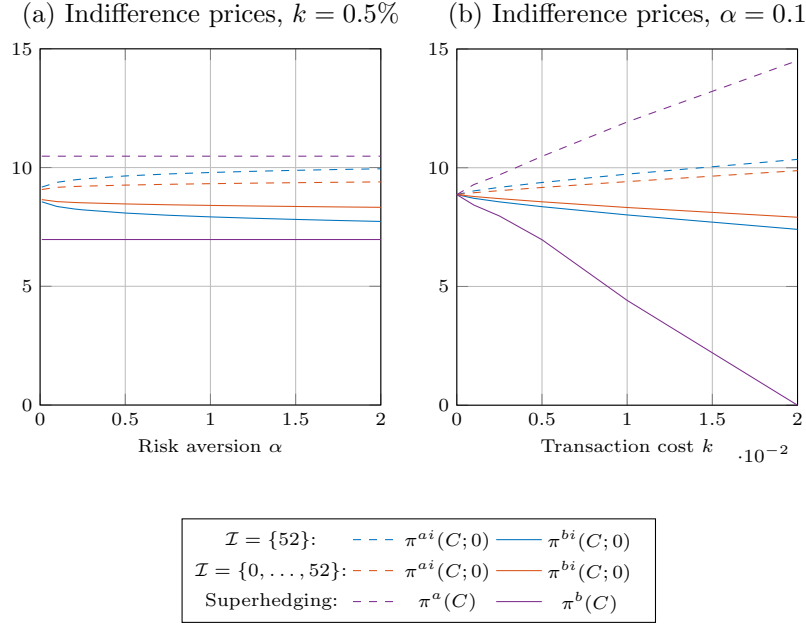
\end{example}

\begin{example} \label{ex:p}
  Buyer's and seller's indifference prices for a range of values of the market probability parameter $p$ in the case where $k=0.005$ and $\alpha=0.1$, are illustrated in Figure \ref{fig:ex:p}. It can be seen in part (a) that indifference pricing spreads tend to be at their largest when $p$ is close to the value of the friction-free risk-neutral probability in this model, which is
  \[q=\frac{(1+r_e)^{1/52}-e^{-\sigma\sqrt{1/52}}}{e^{\sigma\sqrt{1/52}}-e^{-\sigma\sqrt{1/52}}}\approx0.4999.\]
  The effect is more pronounced when injection is allowed at more trading dates. It can be explained by examining the behaviour of $K(0)$, $K(C_{52})$ and $K(-C_{52})$ for different values of $p$, illustrated in part (b). Whilst the dependence of these values on $p$ appear to be convex, they vary in steepness, both within groups associated with the same choice and $\mathcal{I}$, and between groups associated with different choices of~$\mathcal{I}$. This then has consequences for the vertical differences $\pi^{bi}(C;0)=K(C_{52})-K(0)$ and $\pi^{ai}(C;0)=K(0)-K(-C_{52})$.  
\begin{figure}
\begin{center}
\begin{tikzpicture}

   \pgfplotstableread{p-full.dat} {\fullprices}
   \pgfplotstableread{p-util.dat} {\utilprices}
   \pgfplotstableread{p-inter.dat} {\interprices}
   {\scriptsize
	\begin{axis} [name=plot1, anchor=north, height=\plotHeight,width=\plotWidth, xmin=0.05, xmax=0.95, xlabel near ticks, xlabel = Probability $p$, grid = both]
        \addplot[dashed, color1] 
            table [x={pu}, y={ask}] {\utilprices};\label{graphA}
            
        \addplot[color1] 
            table [x={pu}, y={bid}] {\utilprices};\label{graphB}
            
        \addplot[dashed, color4] 
            table [x={pu}, y={ask}] {\interprices};\label{graphC}
        \addplot[color4] 
            table [x={pu}, y={bid}] {\interprices};\label{graphD}

        \addplot[dashed, color2] 
            table [x={pu}, y={ask}] {\fullprices};\label{graphE}
        \addplot[color2] 
            table [x={pu}, y={bid}] {\fullprices};\label{graphF}
    \end{axis}
	}
	\node[align=center,anchor=north] at ($(plot1.north)+(0,0.5\vSpace)$) {(a) Indifference prices};
	
	{\scriptsize
	\begin{axis} [name=plot2, anchor=west, height=\plotHeight,width=\plotWidth, at={($(plot1.east)+(\hSpace,0)$),}, grid=both, legend style={at={(-0.5\hSpace,-\vSpace)}, anchor=north, legend columns=6}, xmin=0.05, xmax=0.95, ymax=100, xlabel near ticks, xlabel = Probability $p$]
	
        \addlegendimage{empty legend}\addlegendentry{$\mathcal{I}=\{52\}$:}
        \addlegendimage{/pgfplots/refstyle=graphA}\addlegendentry{$\pi^{ai}(C;0)$}\addlegendimage{/pgfplots/refstyle=graphB}\addlegendentry{$\pi^{bi}(C;0)$} 
        \addplot[loosely dashed, color1] 
            table [x={pu}, y={KsumWL}] {\utilprices};
            \addlegendentry{$K(C_{52})$}
        \addplot[dashdotted, color1] 
            table [x={pu}, y={Kwealth}] {\utilprices};
            \addlegendentry{$K(0)$}
        \addplot[dotted, color1] 
            table [x={pu}, y={KdiffWL}] {\utilprices};
            \addlegendentry{$K(-C_{52})$}

        \addlegendimage{empty legend}\addlegendentry{$\mathcal{I}=\{0,4,\ldots,52\}$:}
            \addlegendimage{/pgfplots/refstyle=graphC}\addlegendentry{$\pi^{ai}(C;0)$}
            \addlegendimage{/pgfplots/refstyle=graphD}\addlegendentry{$\pi^{bi}(C;0)$}
        \addplot[loosely dashed, color4] 
            table [x={pu}, y={KsumWL}] {\interprices};
            \addlegendentry{$K(C_{52})$}
        \addplot[dashdotted, color4] 
            table [x={pu}, y={Kwealth}] {\interprices};
            \addlegendentry{$K(0)$}
        \addplot[dotted, color4] 
            table [x={pu}, y={KdiffWL}] {\interprices};
            \addlegendentry{$K(-C_{52})$}

        \addlegendimage{empty legend}\addlegendentry{$\mathcal{I}=\{0,1,\ldots,52\}$:}
            \addlegendimage{/pgfplots/refstyle=graphE}\addlegendentry{$\pi^{ai}(C;0)$}
            \addlegendimage{/pgfplots/refstyle=graphF}\addlegendentry{$\pi^{bi}(C;0)$}
        \addplot[loosely dashed, color2] 
            table [x={pu}, y={KsumWL}] {\fullprices};
            \addlegendentry{$K(C_{52})$}
        \addplot[dashdotted, color2] 
            table [x={pu}, y={Kwealth}] {\fullprices};
            \addlegendentry{$K(0)$}
        \addplot[dotted, color2] 
            table [x={pu}, y={KdiffWL}] {\fullprices};
            \addlegendentry{$K(-C_{52})$}

	\end{axis}
	}
	\node[align=center,anchor=north] at ($(plot2.north)+(0,0.5\vSpace)$) {(b) Values of $K$};
\end{tikzpicture}
\end{center}
\caption{Indifference prices and market probability (Example \ref{ex:p})}
\label{fig:ex:p}
\end{figure}
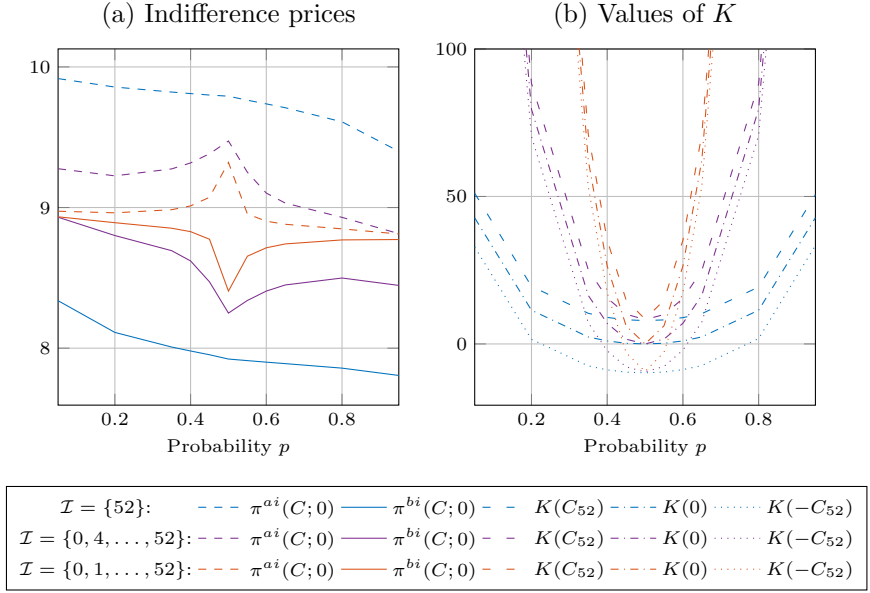
\end{example}

\begin{example} \label{exa:OptSol_Prob}
Figure~\ref{fig:OptSol_Prob} illustrates a number of numerical results related to  optimal injection and hedging strategies for $\mathcal{I}=\{52\}$ and $\mathcal{I}=\{0,13,\ldots,52\}$ and for different values of the probability $p$. The risk-aversion parameter is $\alpha=0.2$ throughout.

Parts (a) and (b) contain histograms of the optimal P\&L $-\ttotal{t\in\mathcal{I}}{}\hat{x}_{t}$ for 100000 randomly generated scenarios in the case where $k=0.005$. It is clear that the P\&L tends to be larger if the real-world probability is further away from the risk-neutral probability (calculated in Example~\ref{ex:p}), thus confirming the analysis in Remark \ref{remark:PL}. The distribution of P\&L depends on $\mathcal{I}$, too, with distributions being much wider in the case where $\mathcal{I}=\{0,13,\ldots,52\}$. Making injections quarterly, instead of at the terminal time step, allows an investor to reduce their regret by taking advantage of the convexity of the disutility function.

Due to the smallness of the transaction costs, Construction \ref{alg:OptStockPosition} produces a unique optimal trading strategy $\processdef{\hat{y}}{t}{-1}{52}$ in this model. Parts (c)--(f) illustrate the optimal stock positions $\process{\hat{y}^s}{t}{0}{52}$ associated with this strategy in two scenarios. The stock positions should be compared to the stock positions associated with the replicating strategy in the binary model without transaction costs (pictured).

Parts (c) and (e) focus on the stock positions when $\mathcal{I}=\{52\}$ in the case of no transaction costs ($k=0$) and $k=0.005$. The presence of transaction costs lead to smoother stock positions due to a reduction in trading. Stock positions tend to be higher for higher values of $p$; this indicates that the investor is taking advantage of market information. 

The corresponding results for the case $\mathcal{I}=\{0,13,\ldots,52\}$ are provided in (d) and (f). In this case the tendency is for stock holdings to be larger (in absolute value) initially, but with larger adjustments each quarter, and tending to similar values in the final quarter as in the case $\mathcal{I}=\{52\}$.
\begin{figure}
\begin{center}

\begin{tikzpicture}
	{\scriptsize
	\begin{axis} [name=plot1,height=\plotHeight,width=\plotWidth, area style, ylabel near ticks, ylabel = Frequency density, xlabel near ticks, grid = both]
		\addplot+[color1, ybar interval, mark=no, opacity=0.5] table [x=Edges1, y = ProbRW1] {Coor_prob_ItimeUtility.dat};
		\label{distA}\addlegendentry{$p=0.3$}
		\addplot+[color2, ybar interval, mark=no, opacity=0.5] table [x=Edges2, y = ProbRW2] {Coor_prob_ItimeUtility.dat};	
		\label{distB}\addlegendentry{$p=0.6$}
		\addplot+[color3, ybar interval, mark=no, opacity=0.5] table [x=Edges3, y = ProbRW3] {Coor_prob_ItimeUtility.dat};	
		\label{distC}\addlegendentry{$p=0.7$}	
	\end{axis}}
	\node[align=center,anchor=north] at ($(plot1.north)+(0,0.5\vSpace)$) {(a) P\&L, $\mathcal{I}=\{52\}$};
	
	{\scriptsize
	\begin{axis} [name=plot2, anchor=west,height=\plotHeight,width=\plotWidth, at={($(plot1.east)+(\hSpace,0)$)}, area style, xlabel near ticks, grid = both]
		\addplot+[color1, ybar interval, mark=no, opacity=0.5] table [x=Edges1, y = ProbRW1] {Coor_prob_ItimeRegret.dat};
		\addlegendentry{$p=0.3$}	
		\addplot+[color2, ybar interval, mark=no, opacity=0.5] table [x=Edges2, y = ProbRW2] {Coor_prob_ItimeRegret.dat};	
		\addlegendentry{$p=0.6$}
		\addplot+[color3, ybar interval, mark=no, opacity=0.5] table [x=Edges3, y = ProbRW3] {Coor_prob_ItimeRegret.dat};	
		\addlegendentry{$p=0.7$}
	\end{axis}
	}
	\node[align=center,anchor=north] at ($(plot2.north)+(0,0.5\vSpace)$) {(b) P\&L, $\mathcal{I}=\{0,13,\ldots,52\}$};
	
	{\scriptsize
	\begin{axis} [name=plot3, anchor=north, height=\plotHeight,width=\plotWidth, at={($(plot1.south)+(0,-\vSpace)$)}, axis y line*=left, xmin=0, xmax=52, xlabel near ticks, xlabel = Time step $t$, xtick={0, 13,26, 39, 52}, ylabel near ticks, ylabel = Stock position $\hat{y}^s_t$, grid = both]
		\addplot[color1] table [x=TimeStep, y = ProbRW1] {Stra_prob_ItimeUtility_Scen_1.dat};	
		\addplot[color2] table [x=TimeStep, y = ProbRW2] {Stra_prob_ItimeUtility_Scen_1.dat};	
		\addplot[color3] table [x=TimeStep, y = ProbRW3] {Stra_prob_ItimeUtility_Scen_1.dat};	
		\addplot[color1, dashed] table [x=TimeStep, y = ProbRW1_TCZero] {Stra_prob_ItimeUtility_Scen_1.dat};	
		\addplot[color2, dashed] table [x=TimeStep, y = ProbRW2_TCZero] {Stra_prob_ItimeUtility_Scen_1.dat};	
		\addplot[color3, dashed] table [x=TimeStep, y = ProbRW3_TCZero] {Stra_prob_ItimeUtility_Scen_1.dat};	
		\addplot[color4] table [x=TimeStep, y = Replication] {Stra_prob_ItimeUtility_Scen_1.dat};	
	\end{axis}
	\begin{axis} [anchor=north, height=\plotHeight,width=\plotWidth, at={($(plot1.south)+(0,-\vSpace)$),}, axis y line*=right, xmin=0, xmax=52, xtick=\empty]
		\addplot[color5, opacity=0.3] table [x=TimeStep, y = Stock] {Stra_prob_ItimeUtility_Scen_1.dat};	
	\end{axis}
	}
	\node[align=center,anchor=north] at ($(plot3.north)+(0,0.5\vSpace)$) {(c) Strategy, $\mathcal{I}=\{52\}$};
	
	{\scriptsize
	\begin{axis} [name=plot4, anchor=north, height=\plotHeight,width=\plotWidth, at={($(plot2.south)+(0,-\vSpace)$)}, axis y line*=left, xmin=0, xmax=52, xlabel near ticks, xlabel = Time step $t$, xtick={0, 13,26, 39, 52}, grid = both]
		\addplot[color1] table [x=TimeStep, y = ProbRW1] {Stra_prob_ItimeRegret_Scen_1.dat};	
		\addplot[color2] table [x=TimeStep, y = ProbRW2] {Stra_prob_ItimeRegret_Scen_1.dat};
		\addplot[color3] table [x=TimeStep, y = ProbRW3] {Stra_prob_ItimeRegret_Scen_1.dat};
		\addplot[color1, dashed] table [x=TimeStep, y = ProbRW1_TCZero] {Stra_prob_ItimeRegret_Scen_1.dat};
		\addplot[color2, dashed] table [x=TimeStep, y = ProbRW2_TCZero] {Stra_prob_ItimeRegret_Scen_1.dat};
		\addplot[color3, dashed] table [x=TimeStep, y = ProbRW3_TCZero] {Stra_prob_ItimeRegret_Scen_1.dat};
		\addplot[color4] table [x=TimeStep, y = Replication] {Stra_prob_ItimeRegret_Scen_1.dat};			
	\end{axis}
	\begin{axis} [anchor=north, height=\plotHeight,width=\plotWidth, at={($(plot2.south)+(0,-\vSpace)$),}, axis y line*=right, xmin=0, xmax=52, xtick=\empty, ylabel near ticks, ylabel = Stock price $S_t$]		
		\addplot[color5, opacity=0.3] table [x=TimeStep, y = Stock] {Stra_prob_ItimeRegret_Scen_1.dat};
	\end{axis}
	}
	\node[align=center,anchor=north] at ($(plot4.north)+(0,0.5\vSpace)$) {(d) Strategy, $\mathcal{I}=\{0,13,\ldots,52\}$};
	
	{\scriptsize
	\begin{axis} [name=plot5, anchor=north, height=\plotHeight,width=\plotWidth, at={($(plot3.south)+(0,-\vSpace)$)}, axis y line*=left, xmin=0, xmax=52, xlabel near ticks, xlabel = Time step $t$, xtick={0, 13,26, 39, 52}, ylabel near ticks, ylabel = Stock position $\hat{y}^s_t$, grid = both]
		\addplot[color1] table [x=TimeStep, y = ProbRW1] {Stra_prob_ItimeUtility_Scen_2.dat};	
		\addplot[color2] table [x=TimeStep, y = ProbRW2] {Stra_prob_ItimeUtility_Scen_2.dat};	
		\addplot[color3] table [x=TimeStep, y = ProbRW3] {Stra_prob_ItimeUtility_Scen_2.dat};	
		\addplot[color1, dashed] table [x=TimeStep, y = ProbRW1_TCZero] {Stra_prob_ItimeUtility_Scen_2.dat};	
		\addplot[color2, dashed] table [x=TimeStep, y = ProbRW2_TCZero] {Stra_prob_ItimeUtility_Scen_2.dat};	
		\addplot[color3, dashed] table [x=TimeStep, y = ProbRW3_TCZero] {Stra_prob_ItimeUtility_Scen_2.dat};	
		\addplot[color4] table [x=TimeStep, y = Replication] {Stra_prob_ItimeUtility_Scen_2.dat};	
	\end{axis}
	\begin{axis} [anchor=north, height=\plotHeight,width=\plotWidth, at={($(plot3.south)+(0,-\vSpace)$),}, axis y line*=right, xmin=0, xmax=52, xtick=\empty]
		\addplot[color5, opacity=0.3] table [x=TimeStep, y = Stock] {Stra_prob_ItimeUtility_Scen_2.dat};	
	\end{axis}
	}
	\node[align=center,anchor=north] at ($(plot5.north)+(0,0.5\vSpace)$) {(e) Strategy, $\mathcal{I}=\{52\}$};	
	
	{\scriptsize
	\begin{axis} [name=plot6, anchor=north, height=\plotHeight,width=\plotWidth, at={($(plot4.south)+(0,-\vSpace)$)}, axis y line*=left, xmin=0, xmax=52, xlabel near ticks, xlabel = Time step $t$, xtick={0, 13,26, 39, 52}, grid = both]
		\addplot[color1] table [x=TimeStep, y = ProbRW1] {Stra_prob_ItimeRegret_Scen_2.dat};\label{curveA}
		\addplot[color2] table [x=TimeStep, y = ProbRW2] {Stra_prob_ItimeRegret_Scen_2.dat};\label{curveB}
		\addplot[color3] table [x=TimeStep, y = ProbRW3] {Stra_prob_ItimeRegret_Scen_2.dat};\label{curveC}

		\addplot[color1, dashed] table [x=TimeStep, y = ProbRW1_TCZero] {Stra_prob_ItimeRegret_Scen_2.dat};\label{curveD}
		\addplot[color2, dashed] table [x=TimeStep, y = ProbRW2_TCZero] {Stra_prob_ItimeRegret_Scen_2.dat};\label{curveE}
		\addplot[color3, dashed] table [x=TimeStep, y = ProbRW3_TCZero] {Stra_prob_ItimeRegret_Scen_2.dat};\label{curveF}
		\addplot[color4] table [x=TimeStep, y = Replication] {Stra_prob_ItimeRegret_Scen_2.dat};\label{curveG}
	\end{axis}
	}
	\node[align=center,anchor=north] at ($(plot6.north)+(0,0.5\vSpace)$) {(f) Strategy, $\mathcal{I}=\{0,13,\ldots,52\}$};
	
	{\scriptsize
	\begin{axis} [anchor=north, height=\plotHeight,width=\plotWidth, at={($(plot4.south)+(0,-\vSpace)$),}, axis y line*=right, xmin=0, xmax=52, xtick=\empty, legend style={at={(-0.5\hSpace,-\vSpace)}, anchor=north, legend columns=4}, ylabel near ticks, ylabel = Stock price $S_t$]
		\addlegendimage{/pgfplots/refstyle=curveA}\addlegendentry{$p=0.3$, $k=0.5\%$}
		\addlegendimage{/pgfplots/refstyle=curveB}\addlegendentry{$p=0.6$, $k=0.5\%$}
		\addlegendimage{/pgfplots/refstyle=curveC}\addlegendentry{$p=0.7$, $k=0.5\%$}
		\addlegendimage{/pgfplots/refstyle=curveG}\addlegendentry{Replication}
		\addlegendimage{/pgfplots/refstyle=curveD}\addlegendentry{$p=0.3$, $k=0\phantom{.5\%}$}
		\addlegendimage{/pgfplots/refstyle=curveE}\addlegendentry{$p=0.6$, $k=0\phantom{.5\%}$}
		\addlegendimage{/pgfplots/refstyle=curveF}\addlegendentry{$p=0.7$, $k=0\phantom{.5\%}$}
		
		\addplot[color5, opacity=0.3] table [x=TimeStep, y = Stock] {Stra_prob_ItimeRegret_Scen_2.dat};
		\addlegendentry{Stock Price}
	\end{axis}
	}
\end{tikzpicture}
\end{center}
\caption{Optimal injection and trading strategies}
\label{fig:OptSol_Prob}
\end{figure}
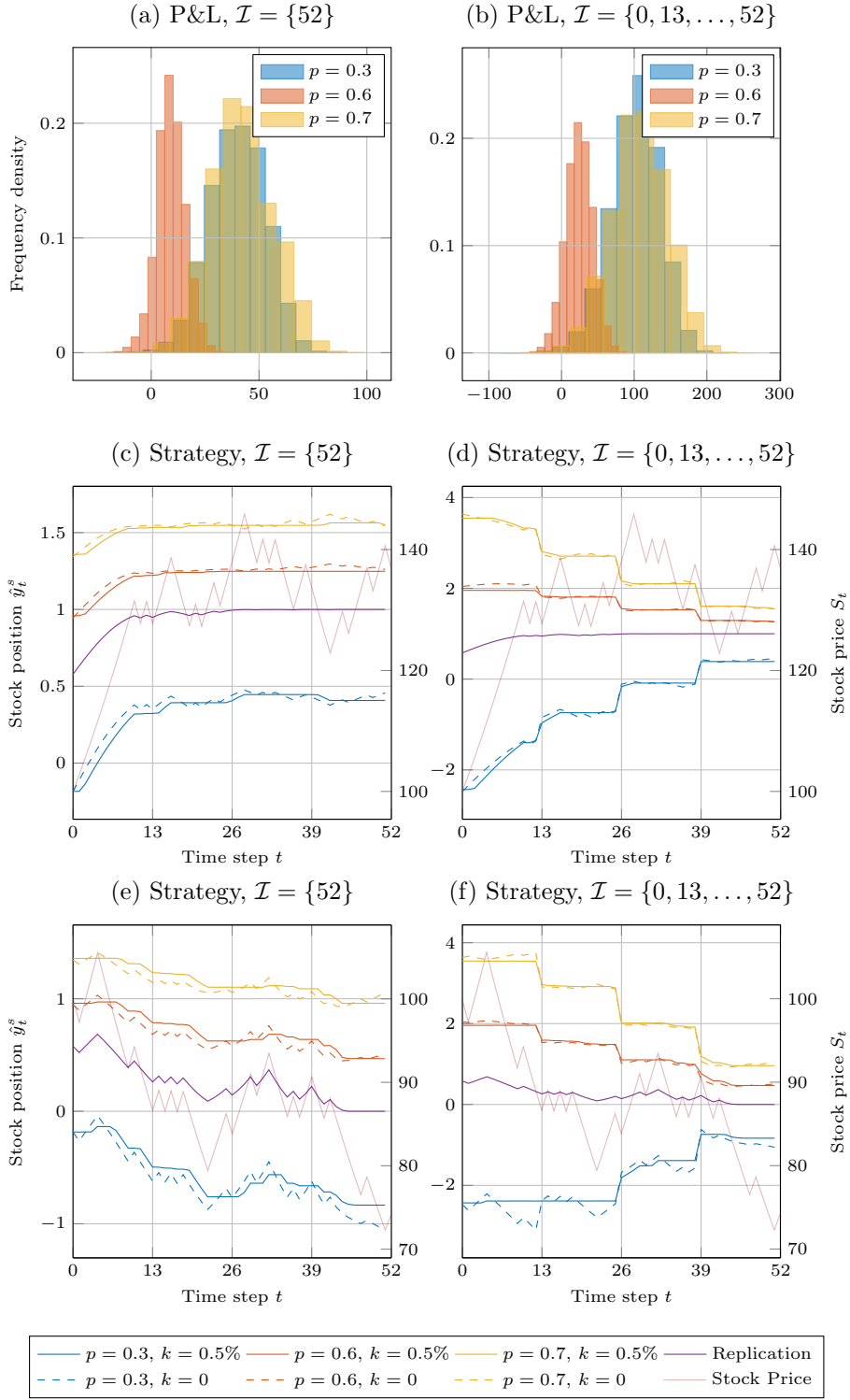
\end{example}

\citet[Section~5.5]{Xu2018Pricng} reported a large number of numerical examples illustrating the methods of this paper, for a selection of options with cash and physical delivery, and for a range of values of $r_e$ and $T$.

\appendix
\section{Generalised convex hull}
\label{sec:conv-hull}

\renewcommand\thesection{\Alph{section}}

The constructions in Section~\ref{sec:Solution-construction} involve a generalisation of the convex hull of convex functions. This section outlines the main properties used in this paper in an abstract setting.

For $k=1,\ldots,m$, let $f_k,g_k:\mathbb{R}\rightarrow\mathbb{R}\cup\{\infty\}$ be proper convex functions that are continuous on their effective domains $\dom f_k=[b_k,a_k]$ for some $b_k,a_k\in\mathbb{R}$ and $\dom g_k=[0,1]$, and 
\begin{equation} \label{eq:g0=0}
 g_k(0)=0.
\end{equation}
Define the \emph{generalised convex hull} $f:\mathbb{R}\rightarrow\mathbb{R}\cup\{\infty\}$ of $f_1,\ldots,f_m$ and $g_1,\ldots g_m$ as
\begin{multline} \label{eq:def:f}
 f(x) \coloneqq \inf \left\{\ttotal{k=1}{m} (q_kf_k(x_k) + g_k(q_k)):q_k\in[0,1],x_k\in[b_k,a_k] \text{ for all } k,\right.\\
 \left.
  \ttotal{k=1}{m} q_k = 1, \ttotal{k=1}{m} q_k x_k = x
 \right\}.
\end{multline}

\subsection{General properties}

The main aim of this section is to establish the key properties needed in Section~\ref{sec:Solution-construction}. Further detail on the arguments below, in a slightly more general setting, were presented by \citet[Chapter 4]{Xu2018Pricng}.

The first result establishes the convexity and boundedness of $f$, as well as the compactness of its effective domain.

\begin{proposition} \label{prop:f-convex}
The function $f$ in~\eqref{eq:def:f} is proper, convex, and
\begin{equation} \label{eq:dom f}
 \dom f = \conv \union{k=1}{m}[b_k,a_k] = \Big[\min_{k}b_k,\max_{k}a_k\Big].
\end{equation}
\end{proposition}

\begin{proof}
Much of the proof is straightforward, hence omitted. The compactness of $\dom f$ comes from \cite[Corollary~9.8.2]{rockafellar1997convex}. The properness of $f$ follows from the fact that continuous proper convex functions with compact domains are bounded from below. To show that $f$ is convex, fix any $y,z\in\dom f$ and $\lambda\in(0,1)$. By~\eqref{eq:dom f} there exists $(q^y_k, y_k)_{k=1}^m$ and $(q^z_k, z_k)_{k=1}^m$ such that $q^y_k,q^z_k\ge0$ and \mbox{$y_k,z_k\in[b_k,a_k]$} for all $k$ and 
 $
 \ttotal{k=1}{m} q^y_k = 1$, $\ttotal{k=1}{m} q^z_k = 1$, $\ttotal{k=1}{m} q^y_k y_k = y$ and $\ttotal{k=1}{m} q^z_k z_k = z$.
Define now
\begin{align*}
 q_k &\coloneqq \lambda q^y_k + (1-\lambda) q^z_k, & x_k \coloneqq
 \begin{cases}
  y_k & \text{if }q_k = 0,\\
  \tfrac{1}{q_k}\left(\lambda q^y_ky_k + (1-\lambda)q^z_kz_k\right) &\text{if } q_k>0
 \end{cases} 
\end{align*}
for all $k$; then
$
 q_k\ge0$ for all $k$, and $\ttotal{k=1}{m} q_k = 1$ and $\ttotal{k=1}{m} q_k x_k = \lambda y + (1-\lambda)z$.
It then follows from~\eqref{eq:def:f} and the convexity of the $f_k$'s and $g_k$'s that 
\begin{multline*}
 f(\lambda y + (1-\lambda)z) \\
 \begin{aligned}
 &\le \ttotal{k=1}{m} \left(q_kf_k(x_k) + g_k(q_k)\right) \\
 &\le \lambda\ttotal{k=1}{m} \left(q^y_kf_k(y_k) + g_k(q^y_k)\right) + (1-\lambda)\ttotal{k=1}{m} \left(q^z_kf_k(z_k) + g_k(q^z_k)\right),
 \end{aligned}
\end{multline*}
and convexity follows from taking the infimum in both terms on the right.
\end{proof}

The remainder of this section is devoted to establishing the closedness of the epigraph of $f$. This then allows us to establish the desired properties; see Proposition~\ref{prop:f-continuous-closed} at the end of the appendix. Define
\begin{align}
 A^g_k &\coloneqq \left\{(q,qx,qy+g_k(q)):q\in[0,1], (x,y)\in\epigraph f_k\right\} \text{ for all }k.\label{eq:Agk}
\end{align}
If $q=0$, then $(q,a,b)\in A^g_k$ if and only if $a=b=0$. This also implies that $A^g_k\neq\emptyset$. Moreover, if $(q,a,b)\in A^g_k$ satisfies $q>0$, then $(q,a,b)+U\subset A^g_k$, where
\[
 U \coloneqq \{(0,0,b)\in\mathbb{R}^3:b\ge0\}.
\]
The properties of $A^g_k$ in the next result will be used in Proposition~\ref{prop:Ef-closed}.

\begin{proposition} \label{prop:collect-Agk}
 The following holds true for the set $A^g_k$ in~\eqref{eq:Agk} for any $k$:
 \begin{enumerate}[(1)]
  \item\label{item:lemma4.6} The set $A^g_k$ is convex.
  \item\label{item:lemma4.7} The closure of $A^g_k$ is $\cl A^g_k = U\cup A^g_k$.
  \item\label{item:prop4.8part1} The recession cone of $\cl A^g_k$ is $0^+(\cl A^g_k) = U$.
 \end{enumerate}
\end{proposition}

\begin{proof}
Item~\ref{item:lemma4.6}: Fix any $\lambda\in(0,1)$, $q_1,q_2\in[0,1]$ and $(x_1,y_1),(x_2,y_2)\in\epigraph f_k$ and define $q \coloneqq \lambda q_1 + (1-\lambda)q_2$ and
\begin{align*}
 z &\coloneqq \lambda(q_1, q_1x_1, q_1y_1 + g_k(q_1)) + (1-\lambda)(q_2,q_2x_2, q_2y_2 + g_k(q_2)). 
\end{align*}
If $q=0$, then $q_1=q_2=0$, after which $x_1=y_1=x_2=y_2=0$ by the observation above, so that $z=0\in A^g_k$. If $q>0$, then define $\varepsilon \coloneqq \lambda g_k(q_1) + (1-\lambda) g_k(q_2) - g_k(q)$ and
$
 (x,y) \coloneqq \tfrac{1}{q}(\lambda q_1(x_1,y_1) + (1-\lambda) q_2(x_2,y_2) + (0,\varepsilon)).
$
Then $\varepsilon\ge0$ because $g_k$ is convex and $(x,y)\in\epigraph f_k$ because $\epigraph f_k$ is convex and unbounded from above. Thus $z=(q,qx,qy+g_k(q))\in A^g_k$, so that $A^g_k$ is convex.

Item~\ref{item:lemma4.7}: Define
$A_k : = \cone(\{1\}\times \epigraph f_k)=\{\lambda(1, z) : \lambda\ge0, z\in \epigraph f_k\}$; then $\cl A_k = U\cup A_k$ due to the compactness of $\dom f_k$ \citep[Theorem~8.2]{rockafellar1997convex}. For every $(0,0,b)\in U\subset\cl A_k$ there exist $(q_n)_{n\ge1}$ in $[0,1]$ and $(x_n,y_n)_{n\ge1}$ in $\epigraph f_k$ such that
\begin{align*}
 (0,0,b) &= \lim_{n\rightarrow\infty} q_n(1,x_n,y_n) = \lim_{n\rightarrow\infty} q_n(1,x_n,y_n + g_k(q_n)),
\end{align*}
with the last equality due to~\eqref{eq:g0=0} and the continuity of $g_n$. Thus $(0,0,b)\in\cl A^g_k$. Combining this with $A^g_k \subseteq \cl A^g_k$ gives that $U\cup A^g_k \subseteq \cl A^g_k$.

To establish the opposite inclusion, suppose that $(q,a,b)\in\cl A^g_k$. Then there exist $(q_n)_{n\ge1}$ in $[0,1]$ and $(x_n,y_n)_{n\ge1}$ in $\epigraph f_k$ such that
\[
 (q,a,b) = \lim_{n\rightarrow\infty}(q_n,q_nx_n,q_ny_n+g_k(q_n)).
\]
Observe that $\lim_{n\rightarrow\infty}g_k(q_n)=g_k(q)$ by the continuity of $g_k$, so that
\[
 b-g_k(q) = \lim_{n\rightarrow\infty} q_ny_n.
\]
Moreover, since $q_n(1,x_n,y_n)\in A_k$ for all $n\in\mathbb{N}$ it follows that
\[
 (q,a,b-g_k(q)) = \lim_{n\rightarrow\infty}q_n(1,x_n,y_n) \in \cl A_k = U\cup A_k.
\]
There are now two possibilities. If $(q,a,b-g_k(q))\in U$, then $q=0$ and therefore \mbox{$(q,a,b)\in U$} by~\eqref{eq:g0=0}. If $(q,a,b-g_k(q))\in A_k$ then there exist $(x,y)\in\epigraph f_k$ such that \mbox{$(q,a,b-g_k(q)) = q(1,x,y)$}, in other words,
$(q,a,b) = (q,qx,qy+g_k(q))\in A^g_k$.

Item~\ref{item:prop4.8part1}: The comments just before this proposition together with item~\ref{item:lemma4.7} gives that $U\subseteq 0^+(\cl A^g_k)$. For the opposite inclusion, take any $(q,a,b)\in0^+(\cl A^g_k)$. Since $0\in\cl A^g_k$, this implies that
\[
 \lambda(q,a,b) = 0 + \lambda(q,a,b) \in \cl A^g_k = U\cup A^g_k \text{ for all } \lambda > 0.
\]
It then follows from~\eqref{eq:Agk} and the comments following it that $q=a=0$, whence $(q,a,b)\in U$.
\end{proof}

\begin{proposition} \label{prop:Ef-closed}
 The set 
\begin{align}
 E_f &\coloneqq \left\{(a,b):(1,a,b)\in\ttotal{k=1}{m} A^g_k\right\}\label{eq:prop4.10-part-1} \\
 &= \left\{ \ttotal{k=1}{m}(q_kx_k,q_ky_k + g_k(q_k)):q_k\in[0,1],(x_k,y_k)\in\epigraph f_k\ \forall k, \ttotal{k=1}{m}q_k=1 \right\}  \label{eq:def-Ef}
 \end{align} 
 is closed.
\end{proposition}

\begin{proof}
We first show that
\begin{equation} \label{eq:prop4.10}
 \{1\}\times E_f = M\cap\ttotal{k=1}{m}\cl A^g_k,
\end{equation}
where
$
 M \coloneqq \{1\} \times \mathbb{R}^2.
$
Equation~\eqref{eq:prop4.10-part-1} gives $\{1\}\times E_f \subseteq M\cap\ttotal{k=1}{m}\cl A^g_k$. To establish the opposite inclusion, fix any $(q,a,b)\in M\cap\ttotal{k=1}{m}\cl A^g_k$; then $q=1$ and by Proposition~\ref{prop:collect-Agk}\ref{item:lemma4.7} there exist $(q_k,a_k,b_k)\in U\cup A^g_k$ for every $k$ such that
\begin{align*}
 (1,a,b) = \ttotal{k=1}{m}(q_k,a_k,b_k).
\end{align*}
Define
$ B \coloneqq \{k:(q_k,a_k,b_k)\in U\}$ and
 $C \coloneqq \{k:(q_k,a_k,b_k)\in A^g_k\setminus U\}.
$
For each $k\in B$, we have $q_k=a_k=0$ and $b_k\ge0$; select any $(x_k,y_k)\in\epigraph f_k$ and observe that
$
 (q_k, q_kx_k, q_ky_k + g_k(q_k)) = 0 = (q_k, a_k, b_k - b_k).
$
Noting that $C\neq\emptyset$ (because $q_k>0$ for at least one $k$), define
$c\coloneqq\tfrac{1}{\lvert C\rvert}\ttotal{k\in B}{}b_k\ge0.$
For each $k\in C$ there exists some $(x_k,y'_k)\in\epigraph f_k$ such that
$
 (q_k,a_k,b_k) = (q_k, q_kx_k, q_ky'_k + g_k(q_k)).
$
Define
\mbox{$
 y_k \coloneqq y'_k + \tfrac{c}{q_k} \ge y'_k;
$}
then $(x_k,y_k)\in\epigraph f_k$ and
\[
 (q_k, q_kx_k, q_ky_k + g_k(q_k)) = (q_k, a_k, b_k + c).
\]
Finally, rearrangement gives that
\begin{align*}
 (1,a,b) 
 &= \ttotal{k\in C}{}(q_k,a_k,b_k+c) = \ttotal{k=1}{m}(q_k, q_kx_k, q_ky_k + g_k(q_k)) \in M\cap\ttotal{k=1}{m}\cl A^g_k,
\end{align*}
which establishes~\eqref{eq:prop4.10}.

Note that $\ttotal{k=1}{m}A^g_k$ is convex \citep[Theorem~3.1]{rockafellar1997convex}. Furthermore, if $z_k\in 0^+(\cl A^g_k)=U$ for all $k$ satisfies $\ttotal{k=1}{m} z_k=0$, then $z_1=\cdots=z_m=0\in U\cap(-U)$; this means that
\begin{equation} \label{eq:prop4.8}
 \cl \ttotal{k=1}{m} A^g_k = \ttotal{k=1}{m}\cl A^g_k
\end{equation}
\citep[Corollary~9.1.1]{rockafellar1997convex}. It remains to show that 
\begin{equation} \label{eq:prop4.9-part-one}
 M \cap \relint \ttotal{k=1}{m} A^g_k \neq \emptyset,
\end{equation}
because then the closedness $E_f$ follows from~\eqref{eq:prop4.8},~\eqref{eq:prop4.10} and
\[
 M\cap\cl\ttotal{k=1}{m} A^g_k = \cl\left(M\cap\ttotal{k=1}{m} A^g_k\right)
\]
\citep[Corollary~6.5.1]{rockafellar1997convex}.

To establish~\eqref{eq:prop4.9-part-one}, observe that $\relint \ttotal{k=1}{m} A^g_k \neq \emptyset$ because $\ttotal{k=1}{m} A^g_k\neq\emptyset$. Thus there exist $q_k\in[0,1]$ and $(x_k,y_k)\in\epigraph f_k$ for all $k$ such that
\[
 (q,a,b) \coloneqq \ttotal{k=1}{m}(q_k,q_kx_k, q_ky_k + g_k(q_k)) \in \relint\ttotal{k=1}{m} A^g_k.
\]
This can now be used to construct a point $z\in M\cap\relint\ttotal{k=1}{m} A^g_k$. There are two possibilities, depending on the value of $q$. If $q\ge1$, define $z\coloneqq\tfrac{1}{q}(q,a,b)$. Then clearly $z\in M$ and moreover $z$ can be written as the convex combination
\[
 z = \tfrac{1}{q}(q,a,b) + \big(1-\tfrac{1}{q}\big)(0,0,0)\in\relint\ttotal{k=1}{m} A^g_k
\]
\citep[Theorem~6.1]{rockafellar1997convex}. If $q\in[0,1]$, define
$
 q'_k \coloneqq \tfrac{1}{m}(2-q)>0$ for all $k$ and
\[
 z' \coloneqq \ttotal{k=1}{m}(q'_k,q'_kx_k, q'_ky_k + g_k(q'_k))\in\ttotal{k=1}{m} A^g_k.
\]
Then
$
 z \coloneqq \tfrac{1}{2}(q,a,b) + \tfrac{1}{2}z' \in \relint\ttotal{k=1}{m} A^g_k
$
\citep[Theorem~6.1]{rockafellar1997convex} and $z\in M$ because
$
 \tfrac{1}{2}q + \tfrac{1}{2}\ttotal{k=1}{m}q'_k = 1.
$
\end{proof}

The following result concludes this section.

\begin{proposition} \label{prop:f-continuous-closed}
 The function $f$ in~\eqref{eq:def:f} is continuous on $\dom f$, and the infimum in~\eqref{eq:def:f} is attained for all $x\in\dom f$.
\end{proposition}

\begin{proof}
 It is sufficient to show that $\epigraph f = E_f$, for then $f$ is lower semicontinuous by Proposition~\ref{prop:Ef-closed}, hence continuous on $\dom f$ because it is a closed bounded interval \citep[Theorems~10.2,~20.5]{rockafellar1997convex}. The fact that the infimum in~\eqref{eq:def:f} is attained for all $x\in\dom f$ follows from the properties of $E_f$.
 
 Suppose that $(x,y)\in E_f$. Thus there exist $q_k\in[0,1]$ and $(x_k,y_k)\in\epigraph f_k$ for all $k$ such that

$\ttotal{k=1}{m}q_k =1$, $\ttotal{k=1}{m}q_kx_k =x$ and $\ttotal{k=1}{m}(q_ky_k + g_k(q_k)) =y$.
Then
\[
 y = \ttotal{k=1}{m}(q_ky_k + g_k(q_k)) \ge \ttotal{k=1}{m}(q_kf_k(x_k) + g_k(q_k)) \ge f(x),
\]
and so $(x,y)\in\epigraph f$.

Conversely, suppose that $(x,y)\in\epigraph f$. Then $f(x)<\infty$ and so by~\eqref{eq:def:f} there exists a sequence $(q_{1n},\ldots,x_{mn},x_{1n},\ldots,x_{mn})_{n\ge1}$ such that 
\[
 f(x) = \lim_{n\rightarrow\infty} \ttotal{k=1}{m}(q_{kn}f_k(x_{kn}) + g_k(q_{kn}))
\]
and for all $n\in\mathbb{N}$ we have 
$
 q_{kn}\in[0,1]$ and $x_{kn}\in[b_k,a_k]$ for all $k$, and $\ttotal{k=1}{m} q_{kn} = 1$ and $\ttotal{k=1}{m} q_{kn}x_{kn} = 1$. For each $n\in\mathbb{N}$ and $k=1,\ldots,m$ define
\[
 y_{kn} \coloneqq f_k(x_{kn}) + y-f(x) \ge f_k(x_{kn});
\]
then $(x_{kn},y_{kn})\in\epigraph f_k$. Define moreover for all $n\in\mathbb{N}$
\[
 y_n \coloneqq \ttotal{k=1}{m}(q_{kn}y_{kn} + g_k(q_{kn})) = \ttotal{k=1}{m}(q_{kn}f_k(x_{kn}) + g_k(q_{kn})) + y-f(x);
\]
then $(x,y_n)\in E_f$ and
$
 \lim_{n\rightarrow\infty} y_n = y.
$
This implies that $(x,y)\in\cl E_f = E_f$ by Proposition~\ref{prop:Ef-closed}, which concludes the proof that $\epigraph f = E_f$.
\end{proof}

\subsection{Numerical approximation}
\label{subsec:numerical-approximation}

Computer implementation of the generalised convex hull necessitates a numerical approximation in all but a few special cases. In this section we propose such a numerical approximation, together with error bounds, that will be suitable for use in the dynamic procedure proposed in Section~\ref{sec:Solution-construction}. It is based on approximation of $f_1,\ldots,f_m$ and $f$ by piecewise linear functions. We will refer to this as the \emph{upper approximation} as it approximates the generalised convex hull $f$ from above.

For every $k$, divide $\dom f_k=[b_k,a_k]$ into $n_k$ subintervals. If $b_k=a_k$, then define
$
 \hat{x}_{k0} \coloneqq \hat{x}_{k1} \coloneqq \cdots \coloneqq \hat{x}_{kn_k} \coloneqq a_k,
$
and if $b_k<a_k$, choose any $(\hat{x}_{kl})_{l=0}^{n_k}$ such that
$
 b_k \eqqcolon \hat{x}_{k0} < \cdots < \hat{x}_{kn_k} \coloneqq a_k.
$
Define $\hat{f}_k:\mathbb{R}\rightarrow\{\infty\}$ as
\begin{equation} \label{eq:def:fkhat}
 \hat{f}_k(x) \coloneqq
 \begin{cases}
  f(\hat{x}_{kl}) & \text{if }x=\hat{x}_{kl}\text{ for some }l,\\
  \frac{\hat{x}_{kl}-x}{\hat{x}_{kl}-\hat{x}_{k[l-1]}}\hat{f}_k(\hat{x}_{k[l-1]}) + \frac{x-\hat{x}_{k[l-1]}}{\hat{x}_{kl}-\hat{x}_{k[l-1]}}\hat{f}_k(\hat{x}_{kl}) & \text{if }x\in(\hat{x}_{k[l-1]},\hat{x}_{kl}) \text{ for any }l,\\
  \infty & \text{if }x\in\mathbb{R}\setminus\dom f_k.
 \end{cases}
\end{equation}
Observe that $\hat{f}_k\ge f_k$ by virtue of the convexity of $f_k$.

Let $\hat{g}$ be the generalised convex hull of $\hat{f}_1,\ldots,\hat{f}_m$ and $g_1,\ldots g_m$, in other words,
\begin{multline} \label{eq:def:ghat}
 \hat{g}(x) \coloneqq \inf \left\{\ttotal{k=1}{m} (q_k\hat{f}_k(x_k) + g_k(q_k)):q_k\in[0,1],x_k\in[b_k,a_k] \thinspace\forall k,\right.\\
 \left.\phantom{\hat{f}}
  \ttotal{k=1}{m} q_k = 1, \ttotal{k=1}{m} q_k x_k = x
 \right\}.
\end{multline}
Then $\hat{g}\ge f$ by definition, and it follows from the arguments in the previous subsection that $\hat{g}$ is convex and continuous on its effective domain $\dom \hat{g}=\dom f$, and that the infimum in~\eqref{eq:def:ghat} is attained for all $x\in\dom \hat{g}=\dom f$.

In practical applications, one often needs to approximate $f$ on some subinterval $[b,a]\subset\dom f$. Divide this interval into $n$ subintervals, as follows: if $b=a$, then define
$
 \hat{x}_0 \coloneqq \hat{x}_1 \coloneqq \cdots  \coloneqq \hat{x}_n \coloneqq a_k,
$
and if $b<a$, choose $(\hat{x}_l)_{l=0}^{n}$ such that
$
 b \eqqcolon \hat{x}_0 < \cdots < \hat{x}_n \coloneqq a.
$
Finally, define
\begin{equation} \label{eq:def:fhat}
 \hat{f}(x) \coloneqq
 \begin{cases}
  \hat{g}(\hat{x}_l) & \text{if }x=\hat{x}_l\text{ for some }l,\\
  \frac{\hat{x}_l-x}{\hat{x}_l-\hat{x}_{l-1}}\hat{g}(\hat{x}_{l-1}) + \frac{x-\hat{x}_{l-1}}{\hat{x}_l-\hat{x}_{l-1}}\hat{g}(\hat{x}_l) & \text{if }x\in(\hat{x}_{l-1},\hat{x}_l)\text{ for any }l,\\
  \infty & \text{if }x\in\mathbb{R}\setminus[b,a].
 \end{cases}
\end{equation}
Then $\hat{f}$ is piecewise linear on its effective domain, and moreover $\hat{f}\ge\hat{g}\ge f$. 

Define the mesh size of the approximation as
 \[
  \Delta \coloneqq \max\Big\{\max_{k,l}(\hat{x}_{kl}-\hat{x}_{k[l-1]}),\max_{l}(\hat{x}_l-\hat{x}_{l-1})\Big\}.
 \]
We now have the following result.

\begin{proposition} \label{prop:error-bound:upper-approximation}
 Let $f$ be defined by~\eqref{eq:def:f}, the function $\hat{f}_k$ by~\eqref{eq:def:fkhat} for all $k$, and $\hat{f}$ by~\eqref{eq:def:fhat}. If $[b,a]\subseteq\relint\dom f$ and there exists $c_k\ge0$ for each $k$ such that
 $
  \big\lvert\hat{f}_k(x) - f_k(x)\big\rvert \le c_k\Delta$ for all $x\in\dom f_k$,
 then there exists $c\ge0$ such that
 $
  \big\lvert\hat{f}(x) - f(x)\big\rvert \le c\Delta$ for all $x\in[a,b]$.
\end{proposition}

\begin{proof}
 For any $l=0,\ldots,n$ we have
 \begin{align}
  0 \le \hat{f}(\hat{x}_l) - f(\hat{x}_l)
  &\le \sup\left\{\ttotal{k=1}{m} q_k\big(\hat{f}_k(x_k) - f_k(x_k)\big):q_k\in[0,1],x_k\in[b_k,a_k] \thinspace\forall k, \right. \nonumber\\
  & \qquad \qquad \left.\phantom{\hat{f}}
  \ttotal{k=1}{m} q_k = 1, \ttotal{k=1}{m} q_k x_k = \hat{x}_l
 \right\} \nonumber\\
  & \le \Delta\sup\left\{\ttotal{k=1}{m} q_kc_k:q_k\in[0,1] \thinspace\forall k,\ttotal{k=1}{m} q_k = 1\right\} \nonumber\\
  & = \Delta \max_k c_k. \label{eq:approx:at-mesh-points}
 \end{align}
 
The function $f$ is Lipschitz on $[b,a]$ \citep[Theorem 10.4]{rockafellar1997convex}, and so there exists some $d\ge0$ such that
 \begin{equation} \label{eq:approx:Lipschitz}
   \lvert f(x) - f(y)\rvert \le d\lvert x-y\rvert \text{ for all } x,y\in[a,b].
 \end{equation}
 For any $x\in[b,a]$ such that $\hat{x}_{l-1} < x < \hat{x}_l$ for some $l>0$, choose $l^\ast\in\{l-1,l\}$ such that
 $
  \hat{f}(\hat{x}_{l^\ast}) = \max\big\{\hat{f}(\hat{x}_{l-1}),\hat{f}(\hat{x}_l)\big\}.
 $
 Then
 \begin{align*}
  \lvert\hat{f}(x) - f(x)\rvert \le \lvert\hat{f}(\hat{x}_{l^\ast}) - f(x)\rvert \le \lvert\hat{f}(\hat{x}_{l^\ast}) - f(\hat{x}_{l^\ast})\rvert + \lvert f(\hat{x}_{l^\ast}) - f(x)\rvert
 \end{align*}
 by~\eqref{eq:def:fhat} and the triangle inequality. Combining this with~\eqref{eq:approx:at-mesh-points} and~\eqref{eq:approx:Lipschitz} then gives the desired result after taking $c\coloneqq d + \max_k c_k$. 
\end{proof}

The upper approximation $\hat{f}$ depends on $\hat{g}$ only via the values $\hat{g}(\hat{x}_0),\ldots,\hat{g}(\hat{x}_n)$. It is possible to calculate these values explicitly in the case where $g_k(q)=q\ln\frac{q}{p_k}$ by using standard techniques from calculus \citep[Section 4.3]{Xu2018Pricng}.

The theoretical error bound in Proposition~\ref{prop:error-bound:upper-approximation} ensures that the upper approximation $\hat{f}$ will converge uniformly to $f$ on $[b,a]$ if the mesh size converges to zero. However, it relies on the Lipschitz coefficient of $f$, which is typically unknown in situations that require approximation (and could well be large). We now present a \emph{lower approximation}, which, while slightly less computationally efficient than the upper approximation, can be used in practical applications to estimate the error of the upper approximation.

For each $k$, let $\check{f}_k$ be any convex piecewise linear function with $\dom \check{f}_k=[b_k,a_k]$ and such that $\check{f}_k\le f_k$. Then let $\check{g}$ be the generalised convex hull of $\check{f}_1,\ldots,\check{f}_m$ and $g_1,\ldots g_m$, in other words,
\begin{multline} \label{eq:def:gcheck}
 \check{g}(x) \coloneqq \inf \left\{\ttotal{k=1}{m} (q_k\check{f}_k(x_k) + g_k(q_k)):q_k\in[0,1],x_k\in[b_k,a_k] \thinspace\forall k,\right.\\
 \left.
  \phantom{\check{f}}\ttotal{k=1}{m} q_k = 1, \ttotal{k=1}{m} q_k x_k = x
 \right\}.
\end{multline}
Then $\check{g}$ is clearly convex and continuous on $\dom \check{g} = \dom f$, and the infimum in\eqref{eq:def:gcheck} is attained for all $x\in\dom \check{g}$. Furthermore, $\check{g}\le f\le \hat{g}$.

If $b=a$, then define
\[
 \check{f}(x) \coloneqq
 \begin{cases}
  \check{g}(x) &\text{if }x=a,\\
  \infty &\text{otherwise};
 \end{cases}
\]
then clearly $\check{f}(a)\le f(a) \le \hat{f}(a)$. Assume for the remainder that $b<a$; this implies that $[b,a]\subset\interior\dom f$. Similar to the upper approximation, divide $[b,a]$ into $n-1$ subintervals by choosing $(\breve{x}_l)_{l=1}^n$ such that
$
 b \eqqcolon \breve{x}_1 < \cdots < \breve{x}_n \coloneqq a.
$
Also choose any $\breve{x}_0\in(\min\dom f,b)$ and $\breve{x}_{n+1}\in(\max\dom f,a)$, and consider the function $\breve{f}$ defined by
\begin{equation} \label{eq:def:fbreve}
 \breve{f}(x) \coloneqq
 \begin{cases}
  \check{g}(\breve{x}_l) & \text{if }x=\breve{x}_l\text{ for some }l,\\
  \frac{\breve{x}_l-x}{\breve{x}_l-\breve{x}_{l-1}}\check{g}(\breve{x}_{l-1}) + \frac{x-\breve{x}_{l-1}}{\breve{x}_l-\breve{x}_{l-1}}\check{g}(\breve{x}_l) & \text{if }x\in(\breve{x}_{l-1},\breve{x}_l)\text{ for any }l>0,\\
  \infty & \text{if }x\in\mathbb{R}\setminus[\breve{x}_0,\breve{x}_{n+1}].
 \end{cases}
\end{equation}
It is convex, piecewise linear and $\check{g}(x)\le\breve{f}(x)$ for all $x\in[\breve{x}_0,\breve{x}_{n+1}]$. The graph of $\breve{f}$ consists of $n+1$ line pieces; the $l^\text{th}$ line piece (where $l=0,\ldots,n$) connects the points $(\breve{x}_l,\check{g}(\breve{x}_l))$ and $(\breve{x}_{l+1},\check{g}(\breve{x}_{l+1}))$, and has slope $m_l\coloneqq\frac{\check{g}(\breve{x}_{l+1})-\check{g}(\breve{x}_l)}{\breve{x}_{l+1}-\breve{x}_l}$. These line pieces are now used to determine the lower approximation $\check{f}$ on $[a,b]$. For $l=1,\ldots,n-1$, determine the point $(\check{x}_l,\check{y}_l)$ by extending the $(l-1)^\text{th}$ and $(l+1)^\text{th}$ line pieces and finding their intersection, in other words,
\begin{align*}
 \check{x}_l &\coloneqq \begin{cases}
    \frac{m_{l+1}\breve{x}_{l+1}-m_{l-1}\breve{x}_l+\check{g}(\breve{x}_l)-\check{g}(\breve{x}_{l+1})}{m_{l+1}-m_{l-1}} & \text{if } m_{l-1} < m_{l+1},\\
    \frac{1}{2}(\breve{x}_l + \breve{x}_{l+1}) & \text{if } m_{l-1} = m_{l+1},
                        \end{cases}
 \\ \check{y}_l &\coloneqq m_{l-1} (\check{x}_l-\breve{x}_l) + \check{g}(\breve{x}_l).
\end{align*}
Finally define $\check{x}_0 \coloneqq \breve{x}_1 = b$,  $\check{y}_0 \coloneqq \check{g}(b)$, $\check{x}_n \coloneqq \breve{x}_n = a$ and $\check{y}_n \coloneqq \check{g}(a),$
after which the lower approximation is defined as
\begin{equation} \label{eq:def:fcheck}
 \check{f}(x) \coloneqq
 \begin{cases}
  \check{y}_l & \text{if }x=\check{x}_l\text{ for some }l,\\
  \frac{\check{x}_l-x}{\check{x}_l-\check{x}_{l-1}}\check{y}_{l-1} + \frac{x-\check{x}_{l-1}}{\check{x}_l-\check{x}_{l-1}}\check{y}_l & \text{if }x\in(\check{x}_{l-1},\check{x}_l)\text{ for any }l>0,\\
  \infty & \text{if }x\in\mathbb{R}\setminus[b,a].
 \end{cases}
\end{equation}
The lower approximation $\check{f}$ is piecewise linear. It is also convex due to the convexity of $\breve{f}$. The fact that $\check{f}\le \check{g}$ (whence $\check{f}\le f$) follows from a simple geometric observation: on every interval $[\breve{x}_l,\breve{x}_{l+1}]$, the graph of $\check{f}$ falls below the extensions of both the $(l-1)^\text{th}$ and $(l+1)^\text{th}$ line pieces of $\breve{f}$, and these extended line pieces in turn fall below the graph of $\check{g}$, due to the convexity of $\check{g}$. \citet[Section~5.4]{Xu2018Pricng} contains the full details.

\section{Proofs}
\label{app:proofs}

\begin{proof}[Proof of Proposition \ref{prop:ask-price-rep}]
 A trading strategy $y\in\mathcal{N}^{2\prime}$ superhedges $c$ if and only if $y_T=0$ and the trading strategy $w\in\mathcal{N}^{2\prime}$ defined as $w_{-1} \coloneqq y_{-1}$ and $w_t \coloneqq y_t + \ttotal{s=0}{t}c_s$ for all $t\ge0$
 satisfies $-\Delta w_t \in \mathcal{K}_t$ for all $t$. The result then follows from Theorem~4.4 of \citet{roux2016american} and \eqref{eq:buyer's-superhedging-price}.
\end{proof}

\begin{proof}[Proof of Theorem \ref{th:solution-exists}] The main argument is analogous to existing results \citep[for example,][Theorem~5.1]{pennanen2014optimal} and is therefore presented in outline only. Observe first from~\eqref{eq:V-presentation01-intermediate} that
\[
 V(u) = \inf_{x\in\mathcal{N},y\in\mathcal{N}^{2'}} \mathbb{E}[f(x,y,u)],
\]
where $f:\Omega\times\mathbb{R}^{T+1}\times\mathbb{R}^{2(T+2)}\times\mathbb{R}^{2(T+1)}\rightarrow\mathbb{R}\cup\{\infty\}$ is defined as
\[
 f^\omega(x,y,u) \coloneqq \begin{cases}
                     \ttotal{t=0}{T} v_t(x_t) & \text{if }(x,y,u)\in \mathcal{B}^\omega,\\
                     \infty & \text{if }(x,y,u)\notin \mathcal{B}^\omega,
                    \end{cases}
\]
where $\processdef{x}{t}{0}{T}$, $\processdef{y}{t}{-1}{T}$, $\processdef{u}{t}{0}{T}$ and
\begin{align*}
 \mathcal{B}^\omega &\coloneqq \big\{(x,y,u)\in\mathbb{R}^{T+1}\times\mathbb{R}^{2(T+2)}\times\mathbb{R}^{2(T+1)}\\
 &\qquad: y_{-1}=y_T=0,-\Delta y_t - u_t + (x_t,0) \in\mathcal{K}^\omega_t\ \forall t \big\}, \\
 \mathcal{K}_t^\omega &\coloneqq \big\{z^\omega\in\mathbb{R}^2: z\in\mathcal{K}_t\big\} = \big\{(z^b,z^s)\in\mathbb{R}^2: x^b+x^s S^{b\omega}_t \ge0,x^b+x^s S^{a\omega}_t \ge0\big\}.
\end{align*}

For each $\omega\in\Omega$ the set $\mathcal{B}^\omega$ is a closed convex cone containing the origin $(0,0,0)$. The regret functions $(v_t)_{t=0}^T$ are convex, lower semicontinuous and bounded from below, and so is $(x,y,u)\mapsto f^\omega(x,y,u)$ \citep[Theorems~5.2,~9.3]{rockafellar1997convex}. In particular, $f$ is a normal integrand \citep[Def.~14.27]{rockafellar2009variational} satisfying $f(0,0,0)=0$.

The convexity of $V$ follows from the convexity of $(x,y,u)\mapsto \mathbb{E}(f(x,y,u))$ \citep[Theorem~1]{rockafellar1974conjugate}. Theorem~2 of \citet{pennanen2012stochastic} then establishes the rest of the claim, provided that
\[
 \mathcal{M} \coloneqq \big\{ (x,y)\in\mathcal{N}\times\mathcal{N}^{2'}:f^{\omega\infty} (x^\omega,y^\omega,0) \le 0\ \forall\omega\in\Omega\big\}
\]
is a linear space. For every $\omega\in\Omega$ the recession function $f^{\omega\infty}$ of $f^\omega$ is
\[
 f^{\omega\infty}(x,y,u) = \lim_{\lambda\downarrow0} f(\lambda x, \lambda y, \lambda u) =
 \begin{cases}
  0 &\text{if } (x,y,u)\in B^\omega, x_t \le 0 \ \forall t,\\
  \infty & \text{otherwise}
 \end{cases}
\]
\citep[Corollary~8.5.2]{rockafellar1997convex}, and therefore
\[
 \mathcal{M} = \left\{ (x,y)\in\mathcal{N}\times(\Phi\cap\Psi):-\Delta y_t + (x_t,0) \in\mathcal{K}_t, x_t\le 0\ \forall t\right\}.
\]
The robust no-arbitrage condition implies that $\Phi\cap\Psi$ is linear \citep[Lemma~2.6]{schachermayer2004fundamental}, and so it suffices to show that if $(x,y)\in\mathcal{M}$, then $x_t=0$ for all $t$. To this end, assume by contradiction that 
$\{x_{t^\ast} < 0 \} \neq \emptyset$ for some $t^\ast$ and define $z\in\mathcal{N}^{2'}$ as $z_{-1} \coloneqq 0$, $z_t \coloneqq y_t - \ttotal{s=0}{t}(x_s,0)$ for all $t\ge0$.
Then
$
 \Delta z_t = \Delta y_t - (x_t,0) \in -\mathcal{K}_t$ for all $t\ge0$,
so that $z\in\Phi$. It further follows from $y_T=0$ that
$
 z_T = -\ttotal{t=0}{T}(x_t,0)\neq0,
$
and hence $z$ violates~\eqref{eq:NA}. This contradiction completes the proof.
\end{proof}

\begin{proof}[Proof of Theorem \ref{thm:StrongDuality}]
 For any $x=x\in\mathcal{N}$, there are two possibilities for the second term in the Lagrangian $L_u$. If $x\in\mathcal{A}_u$, then the coefficient of $\lambda$ must be nonpositive, and by taking $\lambda=0$ we obtain
 \[
  \sup_{\lambda\geq0,(\mathbb{Q},S)\in\bar{\mathcal{P}}}L_{u}(x,\lambda,(\mathbb{Q},S)) = \ttotal{t=0}{T}\mathbb{E}[v_t(x_t)].
 \]
 If $x\notin\mathcal{A}_u$, then there exists some $(\mathbb{Q},S)\in\bar{\mathcal{P}}$ for which the second term is positive whenever $\lambda>0$, and by taking $\lambda$ arbitrarily large we obtain
 \[
  \sup_{\lambda\geq0,(\mathbb{Q},S)\in\bar{\mathcal{P}}}L_{u}(x,\lambda,(\mathbb{Q},S)) = \infty.
 \]
 Combining this with \eqref{eq:V-presentation01} gives
 \[
  \inf_{x\in\mathcal{N}}\sup_{\lambda\geq0,(\mathbb{Q},S)\in\bar{\mathcal{P}}}L_{u}(x,\lambda,(\mathbb{Q},S)) = \inf_{x\in\mathcal{A}_u} \ttotal{t=0}{T}\mathbb{E}[v_t(x_t)] = V(u).
\]

Since the function $V$ is lower semicontinuous and convex on $\mathcal{N}^{2}$, it follows that
\begin{equation}
V(u)=\sup_{z\in\mathcal{N}^{2}}\left\{ \ttotal{t=0}{T}\mathbb{E}[u_t\cdot z_t]-V^{\ast}(z)\right\} \text{ for all }u\in\mathcal{N}^2 \label{eq:V-duality}
\end{equation}
\citep[Theorem 5]{rockafellar1974conjugate}, where the conjugate function $V^\ast$ of $V$ is defined as
\[
V^{\ast}(z)\coloneqq\sup_{u\in\mathcal{N}^{2}}\left\{ \ttotal{t=0}{T}\mathbb{E}[u_t\cdot z_t]-V(u)\right\} \text{ for all }z\in\mathcal{N}^{2}.
\]
For every $z\in\mathcal{N}^{2}$, it follows from~\eqref{eq:V-presentation01-intermediate} that
\begin{multline}
 V^{\ast}(z)=\sup\left\{\ttotal{t=0}{T}\mathbb{E}[z_t\cdot u_t-v_t(x_t)]\right.\\
 \left.\phantom{\ttotal{t=0}{T}}:(x,y,u)\in\mathcal{N}\times\Psi\times\mathcal{N}^{2},\Delta y_t+u_t-(x_t,0)\in-\mathcal{K}_t\thinspace\forall t\right\}.
\end{multline}
This optimization problem can be decoupled into three optimization problems over $x$, $y$ and the transformed process $w\in\mathcal{N}^2$ given by
$w_t \coloneqq \Delta y_t + u_t - (x_t,0)$ for all~$t$. Observing that
\begin{align*}
z_t\cdot u_t-v_t(x_t) & =z_t\cdot(w_t-\Delta y_t+(x_t,0))-v_t(x_t) \\
&=z_t\cdot w_t-z_t\cdot\Delta y_t+z_t^bx_t-v_t(x_t)
\end{align*}
for all $t$, it follows that
\begin{multline}
 V^\ast(z) = \sup_{w\in\mathcal{N}^2,w_t\in-\mathcal{K}_t\thinspace\forall t}\ttotal{t=0}{T}\mathbb{E}[z_t\cdot w_t] - \inf_{y\in\Psi}\ttotal{t=0}{T}\mathbb{E}[z_t\cdot \Delta y_t] \\ + \sup_{x\in\mathcal{N}}\ttotal{t=0}{T}\mathbb{E}\big[z^bx_t-v_t(x_t)\big]. \label{eq:Vstar-definition}
\end{multline}
For the first term on the right hand side of~\eqref{eq:Vstar-definition}, define the positive polar of the solvency cone $\mathcal{K}_t$ as
$
\mathcal{K}_t^{+}\coloneqq\left\{ y\in\mathcal{L}_t^{2}:y\cdot x\geq0\text{ for all }x\in\mathcal{K}_t\right\}$.
Then
\begin{equation}
\sup_{w\in\mathcal{N}^2,w_t\in-\mathcal{K}_t\thinspace\forall t}\ttotal{t=0}{T}\mathbb{E}[z_t\cdot w_t] =\begin{cases}
0 & \text{if }z_t\in\mathcal{K}_t^{+}\thinspace\forall t,\\
\infty & \text{otherwise}
\end{cases}\label{eq:Vstar-definition-part1}
\end{equation}
because it holds for all $t$ that
\[
\sup_{w_t\in-\mathcal{K}_t}\mathbb{E}[z_t\cdot w_t]=\begin{cases}
0 & \text{if }z_t\in\mathcal{K}_t^{+},\\
\infty & \text{otherwise}
\end{cases}
\]
For the second term, the property $y_{-1}=y_T=0$ and rearrangement leads to
\[
 \ttotal{t=0}{T} z_t\cdot \Delta y_t = -\ttotal{t=0}{T-1} \Delta z_{t+1}\cdot y_t \text{ for all }y=y\in\Psi.
\]
Moreover, for all $t<T$, the tower property gives
\begin{align*}
\sup_{y_t\in\mathcal{L}_t^{2}}\mathbb{E}[\Delta z_{t+1}\cdot y_t] & =\sup_{y_t\in\mathcal{L}_t^{2}}\mathbb{E}[\mathbb{E}[\Delta z_{t+1}|\mathcal{F}_t]\cdot y_t]=\begin{cases}
0 & \text{if }\mathbb{E}[\Delta z_{t+1}|\mathcal{F}_t]=0,\\
\infty & \text{otherwise},
\end{cases}
\end{align*}
which implies that
\begin{equation}
\inf_{y\in\Psi}\ttotal{t=0}{T}\mathbb{E}[z_t\cdot\Delta y_t]
=
-\ttotal{t=0}{T-1}\sup_{y_t\in\mathcal{L}_t^{2}}\mathbb{E}[\Delta z_{t+1}\cdot y_t] =
\begin{cases}
0 & \text{if }z\text{ is a martingale},\\
-\infty & \text{otherwise}.
\end{cases} \label{eq:VStar-definition-part2}
\end{equation}
Combining~\eqref{eq:Vstar-definition},~\eqref{eq:Vstar-definition-part1} and~\eqref{eq:VStar-definition-part2}, we obtain
\begin{equation}\label{eq:Vstar-simplified}
V^{\ast}(z)=\begin{cases}
\displaystyle\sup_{x\in\mathcal{N}}\ttotal{t=0}{T}\mathbb{E}\big[z^bx_t-v_t(x_t)\big] & \text{if }z\in\bar{\mathcal{C}},\\
\infty & \text{otherwise},
\end{cases}
\end{equation}
where 
\begin{align}
 \bar{\mathcal{C}} 
 &\coloneqq \left\{ z\in\mathcal{N}^{2}:z\text{ a martingale}, z_t\in\mathcal{K}_t^{+}\ \forall t\right\} \nonumber\\
 &= \big\{ \big(\lambda(1,S_t)\Lambda_t^{\mathbb{Q}}\big)_{t=0}^{T}:\lambda\geq0,\thinspace(\mathbb{Q},S)\in\bar{\mathcal{P}}\big\}, \label{eq:Cbar-and-Pbar}
\end{align}
and where the final equality follows by straightforward adaptation of the arguments of \citet[pp.\ 24-25]{schachermayer2004fundamental}. Substituting~\eqref{eq:Vstar-simplified} into~\eqref{eq:V-duality} gives
\begin{align*}
 V(u)
 &=\sup_{z\in\bar{\mathcal{C}}}\left\{ \ttotal{t=0}{T}\mathbb{E}[u_t\cdot z_t] -  \sup_{x\in\mathcal{N}}\ttotal{t=0}{T}\mathbb{E}\big[z^bx_t-v_t(x_t)\big]\right\} \\
 &= \sup_{z\in\bar{\mathcal{C}}}\inf_{x\in\mathcal{N}}\ttotal{t=0}{T}\mathbb{E}\big[v_t(x_t) + u_t\cdot z_t - z_t^bx_t\big]
\end{align*}
for all $u\in\mathcal{N}^2$. The representation~\eqref{eq:Cbar-and-Pbar} then leads to
\begin{align*}
 V(u) &= \sup_{\lambda\geq0,(\mathbb{Q},S)\in\bar{\mathcal{P}}}\inf_{x\in\mathcal{N}}\ttotal{t=0}{T}\mathbb{E}\big[ v_t(x_t) + \lambda \big( u^b_t + u^s_tS_t - x_t\big)\Lambda^\mathbb{Q}_t\big] \\
 &= \sup_{\lambda\geq0,(\mathbb{Q},S)\in\bar{\mathcal{P}}}\inf_{x\in\mathcal{N}}\ttotal{t=0}{T}\big(\mathbb{E}[v_t(x_t)] + \lambda\mathbb{E}_\mathbb{Q}\big[u^b_t + u^s_tS_T - x_t\big] \big) \\
 &= \sup_{\lambda\geq0,(\mathbb{Q},S)\in\bar{\mathcal{P}}}\inf_{x\in\mathcal{N}}L_{u}(x,\lambda,(\mathbb{Q},S)),
\end{align*}
by the tower property of conditional expectation in conjunction with~\eqref{eq:Lambda^Q} and the martingale property of $S$.
\end{proof}

\begin{proof}[Proof of Proposition \ref{prop:infL_dual}]
Fix any $\lambda\geq0$ and $(\mathbb{Q},S)\in\bar{\mathcal{P}}$, and observe from~\eqref{eq:Lagrangian (scalar)}, the definition of $\mathcal{N}$ and the finiteness of $\Omega$ that
\begin{align}
\inf_{x\in\mathcal{N}}L_{u}(x,\lambda,(\mathbb{Q},S))
&=-\sup_{x\in\mathcal{N}}\ttotal{t=0}{T}\mathbb{E}\big[\lambda \Lambda^\mathbb{Q}_t x_t - v_t(x_t)\big] + \lambda\ttotal{t=0}{T}\mathbb{E}_{\mathbb{Q}}\big[u^b_t + u^s_tS_T\big] \nonumber\\
&=-\ttotal{t=0}{T}\sup_{x_t\in\mathcal{L}_t}\mathbb{E}\big[\lambda \Lambda^\mathbb{Q}_t x_t - v_t(x_t)\big] + \lambda\ttotal{t=0}{T}\mathbb{E}_{\mathbb{Q}}\big[u^b_t + u^s_tS_T\big] \nonumber\\
&=-\ttotal{t=0}{T}\mathbb{E}\big[v^\ast_t\big(\lambda \Lambda^\mathbb{Q}_t\big)\big] + \lambda\ttotal{t=0}{T}\mathbb{E}_{\mathbb{Q}}\big[u^b_t + u^s_tS_T\big]. \label{eq:prop:infL_dual}
\end{align}
Here
\begin{equation}
 v_t^\ast(z) 
 \coloneqq \sup_{y\in\mathbb{R}}\{zy-v_t(y)\} =
 \begin{cases}
  \tfrac{z}{\alpha_t}\ln\tfrac{z}{\alpha_t}-\tfrac{z}{\alpha_t}+1 & \text{if }t\in\mathcal{I}, z\ge0,\\
  0 & \text{if } t\notin \mathcal{I}, z\ge0, \\
  \infty & \text{if }z<0
 \end{cases} \label{eq:vstar}
\end{equation}
is the convex conjugate of $v_t$ for all $t$. Note finally that, for each $t\in\mathcal{I}$,
\begin{align*}
\mathbb{E}\big[v^\ast_t\big(\lambda \Lambda^\mathbb{Q}_t\big)\big]
&= \tfrac{\lambda}{\alpha_t}\mathbb{E}_{\mathbb{Q}}\big[\ln\Lambda^\mathbb{Q}_t\big] + \tfrac{\lambda}{\alpha_t}\big(\ln\tfrac{\lambda}{\alpha_t}-1\big) + 1.
\end{align*}
\end{proof}

\begin{proof}[Proof of Theorem \ref{thm:min_disutility}]
The function
\[
 f(\lambda) \coloneqq \lambda K\big(-\ttotal{t=0}{T}u_t\big) + \ttotal{t\in\mathcal{I}}{}\tfrac{\lambda}{\alpha_t}\big(\ln\tfrac{\lambda}{\alpha_t}-1\big) \text{ for all }\lambda\ge0.
\]
is convex and twice continuously differentiable, and attains its unique minimum at the point $\hat{\lambda}_u$ in~\eqref{eq:def_lambda()}. Substituting~\eqref{eq:def_lambda()} into~\eqref{eq:V-ito-lambda} leads to the formula~\eqref{eq:V-ito-hat-lambda}.
\end{proof}

\begin{proof}[Proof of Theorem \ref{thm:indifferenceprices_expdisutility}]
Observe first that~\eqref{eq:formula:piFbi} follows directly from~\eqref{eq:pibF=-piaF} and~\eqref{eq:formula:piFai}. Define 
\[
\hat{\pi}\coloneqq K\left(\ttotal{t=0}{T}w_t\right)-K\left(\ttotal{t=0}{T}(w_t-c_t)\right).
\]
As $\hat{\pi}$ is deterministic, we have
\[
 K\left((\hat{\pi},0)+\ttotal{t=0}{T}(w_t-c_t)\right) = \hat{\pi} + K\left(\ttotal{t=0}{T}(w_t-c_t)\right) = K\left(\ttotal{t=0}{T}w_t\right)
\]
by~\eqref{eq:def_HI} and~\eqref{eq:K_I(X)}. It then follows from~\eqref{eq:def_lambda()} that
\[
\hat{\lambda}_{c-\hat{\pi}\mathbbm{1}-w}=\exp\left\{\left(\ttotal{t\in\mathcal{I}}{}\tfrac{\ln\alpha_t}{\alpha_t}-K\left((\hat{\pi},0)+\ttotal{t=0}{T}(w_t-c_t)\right)\right)\middle/\ttotal{t\in\mathcal{I}}{}\tfrac{1}{\alpha_t}\right\} = \hat{\lambda}_{-w},
\]
and from~\eqref{eq:V-ito-hat-lambda} that
\[
V(c-\hat{\pi}\mathbbm{1}-w)=\hat{\lambda}_{c-\hat{\pi}\mathbbm{1}-w}\ttotal{t\in\mathcal{I}}{}\tfrac{1}{\alpha_t}-\lvert\mathcal{I}\rvert=\hat{\lambda}_{-w}\ttotal{t\in\mathcal{I}}{}\tfrac{1}{\alpha_t}-\lvert\mathcal{I}\rvert=V(-w).
\]
Thus $\pi^{ai}(c;w)\le \hat{\pi}$.

In order to establish~\eqref{eq:formula:piFai}, it suffices to show that $V(c-\pi\mathbbm{1}-w) > V(c-\hat{\pi}\mathbbm{1}-w)$ for any $\pi < \hat{\pi}$. By Theorem \ref{th:solution-exists} there exists for every $\pi < \hat{\pi}$ a process $x^{\pi}\in \mathcal{A}_{c-\pi\mathbbm{1}-w}$ such that
$
V(c-\pi\mathbbm{1}-w) = \ttotal{t=0}{T}\mathbb{E}[v_t(x^{\pi}_t)].
$
Define a new process $x^{\hat{\pi}}\in\mathcal{N}$ as
\[
 x^{\hat{\pi}}_t \coloneqq \begin{cases}
                x^{\pi}_t + \tfrac{1}{\lvert\mathcal{I}\rvert}(\pi-\hat{\pi}) & \text{if } t\in\mathcal{I},\\
                x^{\pi}_t & \text{otherwise}.
               \end{cases}
\]
Then
\begin{align*}
 \ttotal{t=0}{T}(c_t - \hat{\pi}\mathbbm{1}_t - w_t - (x^{\hat{\pi}}_t,0)) 
 &= \ttotal{t=0}{T}(c_t - w_t - (x^{\pi}_t,0)) - (\pi,0) \\
 &= \ttotal{t=0}{T}(c_t - \pi\mathbbm{1}_t - w_t - (x^{\pi}_t,0)),
\end{align*}
and so it follows from~\eqref{eq:def:Au} that $x^{\hat{\pi}}\in \mathcal{A}_{c-\hat{\pi}\mathbbm{1}-w}$. Furthermore, for every $t\in\mathcal{I}$ we have $v_t(x^\pi_t)>v_t(x^{\hat{\pi}}_t)$ so that
\[
 V(c-\pi\mathbbm{1}-w) = \ttotal{t=0}{T}\mathbb{E}[v_t(x^{\pi}_t)] > \ttotal{t=0}{T}\mathbb{E}[v_t(x^{\hat{\pi}}_t)] \geq V(c-\hat{\pi}\mathbbm{1}-w)
\]
by~\eqref{eq:V-presentation01}, as required.
\end{proof}

\begin{proof}[Proof of Theorem \ref{th:bidaskspread}]
  We first show that 
  \begin{equation} \label{eq:ask-price-bound}
   \pi^{ai}(c;w) \le \pi^a(c) \text{ for all }c,w\in\mathcal{N}^2.
  \end{equation}
  Note first that $c-\pi^a(c)\mathbbm{1}\in\mathcal{Z}$ from~\eqref{eq:pi-a-dual} and~\eqref{eq:def-of-Z-extra}. Furthermore, for any $x\in\mathcal{A}_{-w}$, we have $-w - (x_t,0)_{t=0}^T\in\mathcal{Z}$, and since $\mathcal{Z}$ is a convex cone, it follows that $c-\pi^a(c)\mathbbm{1}-w - (x_t,0)_{t=0}^T\in\mathcal{Z}$, so that finally $x\in\mathcal{A}_{c-\pi^a(c)\mathbbm{1}-w}$. Thus $\mathcal{A}_{-w}\subseteq\mathcal{A}_{c-\pi^a(c)\mathbbm{1}-w}$, so that $V(c-\pi^a(c)\mathbbm{1}-w)\le V(-w)$ by~\eqref{eq:V-presentation01}. This in turn implies that $\pi^{ai}(c;w) \le \pi^a(c)$ by~\eqref{eq:seller_indifferenceprice}. 
  
  Combining~\eqref{eq:ask-price-bound} with~\eqref{eq:pibF=-piaF} and~\eqref{eq:buyer's-superhedging-price} immediately gives for all $c,w\in\mathcal{N}^2$ that
  \[
   \pi^{bi}(c;w) = -\pi^{ai}(-c;w) \ge -\pi^a(-c) = \pi^b(c).
  \]
  The remainder of the proof is devoted to showing the convexity of $u\mapsto\pi^{ai}(u;w)$. Once established, it immediately gives that $u\mapsto\pi^{bi}(u;w)$ is concave by~\eqref{eq:pibF=-piaF}. Moreover, combining the convexity with~\eqref{eq:formula:piFai} gives for all $c,w\in\mathcal{N}^2$ that
  \[
    0 = \pi^{ai}(0;w) \le \tfrac{1}{2}\pi^{ai}(c;w) + \tfrac{1}{2}\pi^{ai}(-c;w),
  \]
  whence
  $
   \pi^{bi}(c;w) = -\pi^{ai}(-c;w) \le \pi^{ai}(c;w).
  $ To establish the convexity, fix $w\in\mathcal{N}^2$ and note that
  \[
   C \coloneqq \{x\in\mathcal{N}^2: V(x-w) \le V(-w)\}
  \]
  is convex because, for all $x,y\in C$ and $\lambda\in[0,1]$ we have
  \[
   V(\lambda x + (1-\lambda)y-w) \le \lambda V(x-w) + (1-\lambda)V(y-w) \le V(w)
  \]
  by the convexity of $V$ (Theorem~\ref{eq:formula:piFai}). For any $c,d\in\mathcal{N}^2$ and $\lambda\in[0,1]$ we have
  \begin{align*}
   \lambda\pi^{ai}(c;w) + (1-\lambda)\pi^{ai}(d;w)
   &= \lambda\inf\{\gamma:c-\gamma\mathbbm{1}\in C\} + (1-\lambda)\inf\{\delta:d-\delta\mathbbm{1}\in C\} \\
   &= \inf\{\lambda\gamma + (1-\lambda)\delta:c-\gamma\mathbbm{1}\in C,d-\delta\mathbbm{1}\in C\}.
  \end{align*}
  By the convexity of $C$, the conditions $c-\gamma\mathbbm{1}\in C,d-\delta\mathbbm{1}\in C$ imply that
  \[
   \lambda c + (1-\lambda)d - (\lambda\gamma + (1-\lambda)\delta)\mathbbm{1} = \lambda(c-\gamma\mathbbm{1}) + (1-\lambda)(d-\delta\mathbbm{1}) \in C,
  \]
  whence
  \begin{align*}
   \lambda\pi^{ai}(c;w) + (1-\lambda)\pi^{ai}(d;w)
   &\ge \inf\{\varepsilon:\lambda c + (1-\lambda)d - \varepsilon\mathbbm{1}\in C\} \\
   &= \pi^{ai}(\lambda c + (1-\lambda)d;w).
  \end{align*}
  This establishes the convexity of $u\mapsto\pi^{ai}(u;w)$ and completes the proof.  
  \end{proof}
    
\begin{proof}[Proof of Proposition \ref{prop:H_formula}]
 Observe from~\eqref{eq:Lambda^Q-and-transition-probabilities} that
 \[
  \ttotal{\nu\in\mu^+}{} q^\nu_t\ln\Lambda^{\mathbb{Q}\nu}_t = \ln\Lambda^{\mathbb{Q}\mu}_{t-1} + \ttotal{\nu\in\mu^+}{} q^\nu_t\ln\tfrac{q^\nu_t}{p^\nu_t} \text{ for all }t>0,\mu\in\Omega^\mathbb{Q}_{t-1},\nu\in\mu^+.
 \]
 Using the nodes in $\Omega_{t-1}$ to partition $\Omega$, and noting that $\mathbb{Q}$ and $\Lambda^{\mathbb{Q}}_t$ are nonzero only on the nodes in $\Omega_{t-1}^{\mathbb{Q}}$, leads to
 \begin{align*}
  \mathbb{E}_{\mathbb{Q}}\big[\ln\Lambda^\mathbb{Q}_t\big]
  &= \ttotal{\mu\in\Omega_{t-1}^{\mathbb{Q}}}{} \mathbb{Q}(\mu)\ttotal{\nu\in\mu^+}{} q^\nu_t\ln\Lambda^{\mathbb{Q}\nu}_t\\
  &=\ttotal{\mu\in\Omega_{t-1}^{\mathbb{Q}}}{} \mathbb{Q}(\mu)\ln\Lambda^{\mathbb{Q}\mu}_{t-1} + \ttotal{\mu\in\Omega_{t-1}^{\mathbb{Q}}}{} \mathbb{Q}(\mu)\ttotal{\nu\in\mu^+}{} q^\nu_t\ln\tfrac{q^\nu_t}{p^\nu_t} \\
  &= \mathbb{E}_{\mathbb{Q}}\big[\ln\Lambda_{t-1}^{\mathbb{Q}}\big] + \ttotal{\mu\in\Omega_{t-1}^{\mathbb{Q}}}{} \mathbb{Q}(\mu)\ttotal{\nu\in\mu^+}{} q^\nu_t\ln\tfrac{q^\nu_t}{p^\nu_t}.
 \end{align*}
 Observing that $\mathbb{E}_\mathbb{Q}\big[\ln\Lambda_0^{\mathbb{Q}}\big]=0$, and introducing a telescoping sum, leads to
 \begin{align*}
  \mathbb{E}_{\mathbb{Q}}\big[\ln\Lambda^\mathbb{Q}_t\big] 
  &= \ttotal{k=1}{t}\ttotal{\mu\in\Omega_{k-1}^{\mathbb{Q}}}{} \mathbb{Q}(\mu)\ttotal{\nu\in\mu^+}{} q^\nu_k\ln\tfrac{q^\nu_k}{p^\nu_k}.
 \end{align*}
 Then, after collecting like terms, it follows that
 \begin{align*}
  \ttotal{t\in\mathcal{I}}{}\tfrac{1}{\alpha_t}\mathbb{E}_{\mathbb{Q}}\big[\ln\Lambda^\mathbb{Q}_t\big]
  &= \ttotal{t\in\mathcal{I}\backslash\{0\}}{}\tfrac{1}{\alpha_t}\mathbb{E}_{\mathbb{Q}}\big[\ln\Lambda^\mathbb{Q}_t\big] \\
  &= \ttotal{t=0}{T-1}a_{t+1}\ttotal{\mu\in\Omega_t^{\mathbb{Q}}}{}\mathbb{Q}(\mu)\ttotal{\nu\in\mu^{+}}{}q_{t+1}^\nu\ln\tfrac{q_{t+1}^\nu}{p_{t+1}^\nu}.
 \end{align*}
 The result follows from~\eqref{eq:def_HI} after using the nodes in $\Omega_{T-1}$ to partition $\Omega$ and observing that
 \[
  \mathbb{E}_{\mathbb{Q}}\big[X^b+X^sS_T\big] = \ttotal{\mu\in\Omega_{T-1}^{\mathbb{Q}}}{}\mathbb{Q}(\mu)\ttotal{\nu\in\mu^{+}}{}q_T^\nu\big(X^{b\nu} + X^{s\nu}S_T^\nu\big).
 \]
\end{proof}

\begin{proof}[Proof of Proposition \ref{prop:J-minimizes-H}]
The properties of the $J_t$'s are proved by backward induction. The convexity, continuity and boundedness properties of $J_T^\nu$ is self-evident from~\eqref{eq:def_JT}. For every $t<T$, suppose that $J^\nu_t$ is convex, bounded from below and continuous on its effective domain $\dom J^\nu_t\subseteq[S_t^{b\nu},S_t^{a\nu}]$ for all $\nu\in\Omega_{t+1}$. Define
\[
 g^\nu(q) \coloneqq 
 \begin{cases}
    a_{t+1} q\ln\tfrac{q}{p_{t+1}^\nu} & \text{if } q\in[0,1],\\
    \infty & \text{otherwise}
 \end{cases}
\]
for all $\nu\in\Omega_{t+1}$; then $g^\nu$ is convex, bounded from below and continuous on its effective domain $\dom g^\nu=[0,1]$. Propositions~\ref{prop:f-convex} and~\ref{prop:f-continuous-closed} then give that $f_t^\mu$ is convex, bounded from below and continuous on its effective domain for every \mbox{$\mu\in\Omega_t$}, and that the infimum in~\eqref{eq:def_ft} is attained for all $x\in\dom f_t^\mu$. It is then clear from~\eqref{eq:def_Jt} that $J_t^\mu$ has the properties claimed. This concludes the inductive step.

To establish~\eqref{eq:J0S0-formula}, fix any $(\mathbb{Q},S)\in\bar{\mathcal{P}}$. We show first by backward induction that
\begin{multline}\label{eq:backward_induction}
\inf_{(\bar{\mathbb{Q}},\bar{S})\in\bar{\mathcal{P}}_{t+1}(\mathbb{Q},S)}H((\bar{\mathbb{Q}},\bar{S});X)
 = \ttotal{k=0}{t}a_{k+1}\ttotal{\mu\in\Omega_k^{\mathbb{Q}}}{} \mathbb{Q}(\mu)\ttotal{\nu\in\mu^{+}}{}q_{k+1}^\nu\ln\tfrac{q_{k+1}^\nu}{p_{k+1}^\nu} \\
 + \ttotal{\mu\in\Omega_t^{\mathbb{Q}}}{}\mathbb{Q}(\mu)\ttotal{\nu\in\mu^{+}}{}q_{t+1}^\nu J^\nu_{t+1}(S^\nu_{t+1})
\end{multline}
for all $t<T$, where
\begin{equation} \label{eq:def-PtQS}
\bar{\mathcal{P}}_t(\mathbb{Q},S)\coloneqq\{ (\bar{\mathbb{Q}},\bar{S})\in\bar{\mathcal{P}}:\bar{\mathbb{Q}}=\mathbb{Q}\text{ on }\mathcal{F}_t,\bar{S}_k=S_k\thinspace\forall k\le t\} 
\end{equation}
is the collection of martingale pairs that coincide with $(\mathbb{Q},S)$ up to time $t$. When $t=T-1$, we have $\bar{\mathcal{P}}_T(\mathbb{Q},S)=\{(\mathbb{Q},S)\}$, so that~\eqref{eq:backward_induction} follows from~\eqref{eq:H_formula} and~\eqref{eq:def_JT}. Assume now that~\eqref{eq:backward_induction} holds for some $t=1,\ldots,T-1$. Rearrangement gives
\begin{multline*}
\inf_{(\bar{\mathbb{Q}},\bar{S})\in\bar{\mathcal{P}}_{t+1}(\mathbb{Q},S)}H((\bar{\mathbb{Q}},\bar{S});X)
 = \ttotal{k=0}{t-1}a_{k+1}\ttotal{\mu\in\Omega_k^{\mathbb{Q}}}{}\mathbb{Q}(\mu)\ttotal{\nu\in\mu^{+}}{}q_{k+1}^\nu\ln\tfrac{q_{k+1}^\nu}{p_{k+1}^\nu} \\
 + \ttotal{\mu\in\Omega_t^{\mathbb{Q}}}{}\mathbb{Q}(\mu)\ttotal{\nu\in\mu^{+}}{}q_{t+1}^\nu\left(a_{t+1} \ln\tfrac{q_{t+1}^\nu}{p_{t+1}^\nu}+J_{t+1}^\nu(S_{t+1}^\nu)\right),
\end{multline*}
after which we obtain from~\eqref{eq:def_Pbar},~\eqref{eq:def-PtQS} and~\eqref{eq:def_Jt} that
\begin{multline*}
\inf_{(\bar{\mathbb{Q}},\bar{S})\in\bar{\mathcal{P}}_t(\mathbb{Q},S)}H((\bar{\mathbb{Q}},\bar{S});X)\\
\begin{aligned}
&= \ttotal{k=0}{t-1}a_{k+1}\ttotal{\mu\in\Omega_k^{\mathbb{Q}}}{}\mathbb{Q}(\mu)\ttotal{\nu\in\mu^{+}}{}q_{k+1}^\nu\ln\tfrac{q_{k+1}^\nu}{p_{k+1}^\nu} + \ttotal{\mu\in\Omega_t^{\mathbb{Q}}}{}\mathbb{Q}(\mu)J^\mu_t(S^\mu_t)\\
&= \ttotal{k=0}{t-1}a_{k+1}\ttotal{\mu\in\Omega_k^{\mathbb{Q}}}{}\mathbb{Q}(\mu)\ttotal{\nu\in\mu^{+}}{}q_{k+1}^\nu\ln\tfrac{q_{k+1}^\nu}{p_{k+1}^\nu} + \ttotal{\mu\in\Omega_{t-1}^{\mathbb{Q}}}{}\mathbb{Q}(\mu)\ttotal{\nu\in\mu^{+}}{}q_t^\nu J^\mu_t(S^\mu_t).
\end{aligned}
\end{multline*}
This concludes the inductive step. 

Finally, when $t=0$, the equation~\eqref{eq:backward_induction} reduces to
\begin{align*}
 \inf_{(\bar{\mathbb{Q}},\bar{S})\in\bar{\mathcal{P}}_1(\mathbb{Q},S)}H((\bar{\mathbb{Q}},\bar{S});X)
 = a_1\ttotal{\nu\in\Omega_1}{}q_1^\nu\ln\tfrac{q_1^\nu}{p_1^\nu} + \ttotal{\nu\in\Omega_1}{}q_1^\nu J^\nu_1(S^\nu_1),
\end{align*}
and again combining~\eqref{eq:def_Pbar},~\eqref{eq:def-PtQS} and~\eqref{eq:def_Jt} yields
\[
 \inf_{(\bar{\mathbb{Q}},\bar{S})\in\bar{\mathcal{P}},\bar{S}_0=S_0}H((\bar{\mathbb{Q}},\bar{S});X)
 = \inf_{(\bar{\mathbb{Q}},\bar{S})\in\bar{\mathcal{P}}_0(\mathbb{Q},S)}H((\bar{\mathbb{Q}},\bar{S});X)
 = J_0(S_0).
\]
This completes the proof.
\end{proof}

\begin{proof}[Proof of Theorem \ref{thm:opt_(QS)=opt(qx)}]
Standard arguments \citep[Theorem 5.25]{cutland2012derivative} can be used to show that $\hat{\mathbb{Q}}$ is a probability measure. The process $\hat{S}$ is a martingale under $\hat{\mathbb{Q}}$ by~\eqref{eq:Shat-martingale}, whence $(\hat{\mathbb{Q}},\hat{S})\in\bar{\mathcal{P}}$. Furthermore, recursive expansion of~\eqref{eq:QShat-and-Js} gives
\begin{multline*}
 J_0(\hat{S}_0)
 = \ttotal{t=0}{T-1}a_{t+1}\ttotal{\mu\in\Omega_t^{\hat{\mathbb{Q}}}}{}\hat{\mathbb{Q}}(\mu)\ttotal{\nu\in\mu^{+}}{}\hat{q}_{t+1}^\nu\ln\tfrac{\hat{q}_{t+1}^\nu}{p_{t+1}^\nu} \\
 + \ttotal{\mu\in\Omega_{T-1}^{\hat{\mathbb{Q}}}}{}\hat{\mathbb{Q}}(\mu)\ttotal{\nu\in\mu^{+}}{}\hat{q}_T^\nu J^\nu_T(\hat{S}_T^\nu) = H((\hat{\mathbb{Q}},\hat{S});X)
\end{multline*}
from~\eqref{eq:def_HI} and~\eqref{eq:def_JT}. Then~\eqref{eq:hatS-minimizes-J0}, Proposition~\ref{prop:J-minimizes-H} and~\eqref{eq:K_I(X)} combine to give
\begin{align*}
 J_0(\hat{S}_0) &= \min_{(\mathbb{Q},S)\in\bar{\mathcal{P}}}H((\mathbb{Q},S);X) = K(X).
\end{align*}

We now show that $(\hat{\mathbb{Q}},\hat{S})\in\mathcal{P}$. Suppose
by contradiction that $(\hat{\mathbb{Q}},\hat{S})\in\bar{\mathcal{P}}\backslash\mathcal{P}$, in other words, $\Lambda_t^{\hat{\mathbb{Q}}}(\omega)=0$ for some $t=0,\ldots,T$ and $\omega\in\Omega$. Fix any $(\mathbb{Q},S)\in\mathcal{P}$, and define
\[
 \epsilon \coloneqq \tfrac{1}{2}\exp\left\{\big(H((\hat{\mathbb{Q}},\hat{S});X) - H((\mathbb{Q},S);X)\big)\middle/\ttotal{t\in\mathcal{I}}{}\tfrac{1}{\alpha_t}\mathbb{Q}\big(\Lambda_t^{\hat{\mathbb{Q}}}=0\big)\right\}.
\]
Observe that $\epsilon\in[0,1)$ because
$
 H((\hat{\mathbb{Q}},\hat{S});X) = J_0(\hat{S}_0) \le J_0(S_0) \le H((\mathbb{Q},S);X).
$
Define a new probability measure $\bar{\mathbb{Q}}:\mathcal{F}\rightarrow[0,1]$ and stochastic process $\bar{S}\in\mathcal{N}$ as
\begin{align}
\bar{\mathbb{Q}} & \coloneqq\epsilon\mathbb{Q} + (1-\epsilon)\hat{\mathbb{Q}}, \label{eq:Qepsilon}\\
\bar{S}_t & \coloneqq\epsilon S_t\mathbb{E}\left[\tfrac{d\mathbb{Q}}{d\bar{\mathbb{Q}}}\middle|\mathcal{F}_t\right] + (1-\epsilon)\hat{S}_t\mathbb{E}\left[\tfrac{d\hat{\mathbb{Q}}}{d\bar{\mathbb{Q}}}\middle|\mathcal{F}_t\right]\text{ for all }t. \label{eq:Sepsilon}
\end{align}
Then $(\bar{\mathbb{Q}},\bar{S})\in\mathcal{P}$ \citep[Lemma 7.2]{roux2008options}, after which~\eqref{eq:def_HI} gives
\begin{multline}
 H((\bar{\mathbb{Q}},\bar{S});X) - H((\hat{\mathbb{Q}},\hat{S});X)
 = \ttotal{t\in\mathcal{I}}{}\tfrac{1}{\alpha_t}\mathbb{E}\big[\Lambda_t^{\bar{\mathbb{Q}}}\ln\Lambda_t^{\bar{\mathbb{Q}}}-\Lambda_t^{\hat{\mathbb{Q}}}\ln\Lambda_t^{\hat{\mathbb{Q}}}\big] 
 \\+ \epsilon\big(\mathbb{E}_{\mathbb{Q}}\big[X^b+X^sS_T\big] - \mathbb{E}_{\hat{\mathbb{Q}}}\big[X^b+X^s\hat{S}_T\big]\big). \label{eq:HIe}
\end{multline}
The mapping $x\mapsto x\ln x$ is convex on $[0,\infty)$, and so
\begin{align} \label{eq:convexity-Lambda}
\Lambda_t^{\bar{\mathbb{Q}}}\ln\Lambda_t^{\bar{\mathbb{Q}}} - \Lambda_t^{\hat{\mathbb{Q}}}\ln\Lambda_t^{\hat{\mathbb{Q}}}
& \le \epsilon \big(\Lambda_t^{\mathbb{Q}}\ln\Lambda_t^{\mathbb{Q}} - \Lambda_t^{\hat{\mathbb{Q}}}\ln\Lambda_t^{\hat{\mathbb{Q}}}\big)\text{ for all }t.
\end{align}
Furthermore, on the set $\big\{\Lambda_t^{\hat{\mathbb{Q}}}=0\big\}$, and recalling the convention $0\ln0=0$, we have
\begin{align*}
 \Lambda_t^{\bar{\mathbb{Q}}}\ln\Lambda_t^{\bar{\mathbb{Q}}} - \Lambda_t^{\hat{\mathbb{Q}}}\ln\Lambda_t^{\hat{\mathbb{Q}}}
 = \epsilon\Lambda_t^{\mathbb{Q}}\ln\epsilon\Lambda_t^{\mathbb{Q}} 
 &= \epsilon \big(\Lambda_t^{\mathbb{Q}}\ln\Lambda_t^{\mathbb{Q}} - \Lambda_t^{\hat{\mathbb{Q}}}\ln\Lambda_t^{\hat{\mathbb{Q}}}\big) + \epsilon\Lambda_t^{\mathbb{Q}}\ln\epsilon.
\end{align*}
Substituting this into~\eqref{eq:HIe} gives
\begin{multline*}
 H((\bar{\mathbb{Q}},\bar{S});X) - H((\hat{\mathbb{Q}},\hat{S});X) \\
 \le \epsilon \left(H((\mathbb{Q},S);X) - H((\hat{\mathbb{Q}},\hat{S});X) + \ln\epsilon\ttotal{t\in\mathcal{I}}{}\tfrac{1}{\alpha_t}\mathbb{Q}\big(\Lambda_t^{\hat{\mathbb{Q}}}=0\big)\right).
\end{multline*}
The choice of $\epsilon$ implies that $H((\bar{\mathbb{Q}},\bar{S});X) < H((\hat{\mathbb{Q}},\hat{S});X)$, which is a contradiction. Hence $\hat{\mathbb{Q}}(\omega)>0$ for all $\omega\in\Omega$, so that $(\hat{\mathbb{Q}},\hat{S})\in\mathcal{P}$.

The proof is complete upon establishing the uniqueness of $\hat{\mathbb{Q}}$ on the nodes in $\mathcal{I}$. To this end, suppose by contradiction that there exists another pair $(\mathbb{Q},S)\in\mathcal{P}$ such that $H((\hat{\mathbb{Q}},\hat{S});X)=H((\mathbb{Q},S);X)$ and $\hat{\mathbb{Q}}(\nu')\neq\mathbb{Q}(\nu')$ for some $t'\in\mathcal{I}$ and $\nu'\in\Omega_{t'}$. The argument now proceeds along similar lines as above: take any $\epsilon\in(0,1)$, and use \eqref{eq:Qepsilon}--\eqref{eq:Sepsilon} to define a new pair $(\bar{\mathbb{Q}},\bar{S})\in\mathcal{P}$. This immediately leads to~\eqref{eq:HIe} and~\eqref{eq:convexity-Lambda}, noting in \eqref{eq:convexity-Lambda} that $\Lambda^{\hat{\mathbb{Q}}}_t(\nu')\neq\Lambda^{\mathbb{Q}}_t(\nu')$ gives
\[
\Lambda_{t'}^{\bar{\mathbb{Q}}}\ln\Lambda_{t'}^{\bar{\mathbb{Q}}} - \Lambda_{t'}^{\hat{\mathbb{Q}}}\ln\Lambda_{t'}^{\hat{\mathbb{Q}}}
< \epsilon \big(\Lambda_{t'}^{\mathbb{Q}}\ln\Lambda_{t'}^{\mathbb{Q}} - \Lambda_{t'}^{\hat{\mathbb{Q}}}\ln\Lambda_{t'}^{\hat{\mathbb{Q}}}\big) \text{ on }\nu'.
\]
Substituting into~\eqref{eq:def_HI}, it follows that
\begin{multline*}
 H((\bar{\mathbb{Q}},\bar{S});X) - H((\hat{\mathbb{Q}},\hat{S});X)
 \\
 \begin{aligned}
 &< \epsilon\ttotal{t\in\mathcal{I}}{}\tfrac{1}{\alpha_t}\mathbb{E}\big[\Lambda_t^{\mathbb{Q}}\ln\Lambda_t^{\mathbb{Q}}-\Lambda_t^{\hat{\mathbb{Q}}}\ln\Lambda_t^{\hat{\mathbb{Q}}}\big] + \epsilon\left(\mathbb{E}_{\mathbb{Q}}\big[X^b+X^sS_T\big] - \mathbb{E}_{\hat{\mathbb{Q}}}\big[X^b+X^s\hat{S}_T\big]\right) \\
 &= \epsilon(H((\mathbb{Q},S);X) - H((\hat{\mathbb{Q}},\hat{S});X)) = 0,
 \end{aligned}
\end{multline*}
in other words,
$
 H((\bar{\mathbb{Q}},\bar{S});X) < H((\hat{\mathbb{Q}},\hat{S});X).
$
This contradicts the assumption that $(\hat{\mathbb{Q}},\hat{S})$ is a solution to the optimization problem~\eqref{eq:K_I(X)}.
\end{proof}

\begin{proof}[Proof of Proposition \ref{prop:Optimal_Investment}]
The partial uniqueness property of $\hat{\mathbb{Q}}$ in Theorem~\ref{thm:opt_(QS)=opt(qx)} ensures that $\hat{x}$ is well defined and unique, irrespective of the minimiser $(\hat{\mathbb{Q}},\hat{S})$ chosen. Straightforward calculation and~\eqref{eq:V-ito-hat-lambda} also gives that
\[
 \ttotal{t=0}{T}\mathbb{E}\left[v_{t}(\hat{x}_t)\right] = \hat{\lambda}_u\ttotal{t\in\mathcal{I}}{}\tfrac{1}{\alpha_t}-\lvert\mathcal{I}\rvert = V(u).
\]
It then remains only to show that $\hat{x}\in\mathcal{A}_u$, and that $\hat{x}$ is the unique minimiser in~\eqref{eq:Problem 1'}. To this end, it suffices to show that any minimiser $\bar{x}\in\mathcal{A}_u$ in~\eqref{eq:Problem 1'} satisfies
\begin{equation}
v_{t}^{\ast}(\hat{\lambda}_{u}\Lambda_{t}^{\hat{\mathbb{Q}}})=\hat{\lambda}_{u}\Lambda_{t}^{\hat{\mathbb{Q}}}\bar{x}_{t}-v_{t}(\bar{x}_{t})\text{ for all }t,\label{eq:dual solution}
\end{equation}
where $v^\ast_t$ is the convex conjugate of $v_t$; see~\eqref{eq:vstar}. This system of equations has a unique solution in $\mathcal{N}$, namely $\hat{x}$. This means that $\bar{x}=\hat{x}$, which concludes the proof.

Let $\bar{x}\in\mathcal{A}_u$ be any minimiser in~\eqref{eq:Problem 1'}; its existence is guaranteed by Theorem~\ref{th:solution-exists}. Observing from \eqref{eq:def-of-Z-extra} that
$
 \ttotal{t=0}{T}\mathbb{E}_{\mathbb{Q}}[u_{t}^{b}+u_{t}^{s}S_{T}-\bar{x}_{t}]\leq0$  for all $(\mathbb{Q},S)\in\bar{\mathcal{P}},
$
it then follows from \eqref{eq:Lagrangian (scalar)} that
\[
L_{u}(\bar{x},\hat{\lambda}_u,(\mathbb{Q},S)) \le \sup_{\lambda\geq0,(\mathbb{Q},S)\in\bar{\mathcal{P}}}L_{u}(\bar{x},\lambda,(\mathbb{Q},S))=\ttotal{t=0}{T}\mathbb{E}\left[v_{t}(\bar{x}_t)\right]=V(u).
\]
Furthermore, as $(\hat{\lambda}_{u},(\hat{\mathbb{Q}},\hat{S}))$ maximises
\eqref{eq:problem 1' dual}, we have 
\[
L_{u}(\bar{x},\hat{\lambda}_u,(\mathbb{Q},S)) \ge \inf_{x\in\mathcal{N}}L_{u}(x,\hat{\lambda}_{u},(\hat{\mathbb{Q}},\hat{S}))=V(u).
\]
Taken together with \eqref{eq:prop:infL_dual}, this gives
\begin{align*}
L_{u}(\bar{x},\hat{\lambda}_{u},(\hat{\mathbb{Q}},\hat{S})) 
&= \inf_{x\in\mathcal{N}}L_{u}(x,\hat{\lambda}_{u},(\hat{\mathbb{Q}},\hat{S})) \\
&= \ttotal{t=0}{T}\big(-\mathbb{E}\big[v^\ast_t\big(\hat{\lambda}_u \Lambda^{\hat{\mathbb{Q}}}_t\big)\big] + \hat{\lambda}_u\mathbb{E}_{\hat{\mathbb{Q}}}\big[u^b_t + u^s_t\hat{S}_T\big]\big).
\end{align*}
Combining with \eqref{eq:Lagrangian (scalar)} and rearranging, we obtain
 \[
\ttotal{t=0}{T}\mathbb{E}\big[v_{t}^{\ast}(\hat{\lambda}_{u}\Lambda_{t}^{\hat{\mathbb{Q}}})+v_{t}(\bar{x}_{t})-\hat{\lambda}_u\Lambda_{t}^{\hat{\mathbb{Q}}}\bar{x}_{t}\big]=0.
\]
This is the sum of expectations of nonnegative random variables, and the conclusion is \eqref{eq:dual solution}.
\end{proof}

\begin{proof}[Proof of Proposition \ref{prop:shadow-price-properties}]
Item \ref{prop:shadow-price-properties:1}: Suppose that $\hat{y}\in\Psi$ solves \eqref{eq:Problem 1} in the friction-free model with price process $\hat{S}$ and it satisfies \eqref{eq:TradeAtSpread}. Then \eqref{eq:TradeAtSpread}
gives  
\[
\ttotal{t=0}{T}\mathbb{E}\big[v_{t}(\phi_{t}(\Delta\hat{y}_{t}+u_{t}))\big]=\ttotal{t=0}{T}\mathbb{E}\big[v_{t}(\Delta\hat{y}^b_{t}+u^b_{t} + (\Delta\hat{y}^s_{t}+u^s_{t})\hat{S}_t)\big]=V(u)
\]
by~\eqref{eq:EqualDisutilitty_ShadowPrice}. Thus $\hat{y}$ solves~\eqref{eq:Problem 1} in the market model
with bid-ask spread~$[S^{b},S^{a}]$.

Item \ref{prop:shadow-price-properties:2}: Suppose that $\hat{y}\in\Psi$ solves \eqref{eq:Problem 1} in the model with bid-ask spread $[S^b,S^a]$. Proposition \ref{prop:Optimal_Investment} guarantees that the optimisation problem \eqref{eq:Problem 1'} has a unique solution $\hat{x}\in\mathcal{N}$ with $\hat{x}_t=0$ for all $t\notin\mathcal{I}$, and
$
 \phi_t(\Delta\hat{y}_{t}+u_{t}) = \hat{x}_t \text{ for all }t.
$
It then follows from \eqref{eq:EqualDisutilitty_ShadowPrice} that
\begin{align}
 \ttotal{t=0}{T}\mathbb{E}\big[v_{t}(\phi_{t}(\Delta\hat{y}_{t}+u_{t}))\big]
 &=\inf_{y\in\Psi}\ttotal{t=0}{T}\mathbb{E}\big[v_{t}(\Delta y^b_{t}+u^b_{t} + (\Delta y^s_{t}+u^s_{t})\hat{S}_t)\big] \nonumber \\
 &=\ttotal{t=0}{T}\mathbb{E}\big[v_{t}(\Delta\hat{y}^b_{t}+u^b_{t} + (\Delta\hat{y}^s_{t}+u^s_{t})\hat{S}_t)\big], \label{eq:equality-y-shadow}
\end{align}
where the last equality comes from the fact that \eqref{def:phi} and $S^b_t\le \hat{S}_t\le S^a_t$ gives
\[
 \hat{x}_t = \phi_t(\Delta\hat{y}_{t}+u_{t}) \ge \Delta\hat{y}^b_{t}+u^b_{t} + (\Delta\hat{y}^s_{t}+u^s_{t})\hat{S}_t \text{ for all }t.
\]
This means that $\hat{y}$ solves~\eqref{eq:Problem 1} in the model with stock price process $\hat{S}$. 

Lack of arbitrage in the friction-free model with stock price process $\hat{S}$ implies that the results in this paper apply directly to that model. In particular, Proposition~\ref{prop:Optimal_Investment} guarantees that the optimisation problem \eqref{eq:Problem 1'} has a unique solution $\bar{x}\in\mathcal{N}$ with $\bar{x}_t=0$ for all $t\notin\mathcal{I}$. This means that
$
 \Delta\hat{y}^b_{t}+u^b_{t} + (\Delta\hat{y}^s_{t}+u^s_{t})\hat{S}_t = \bar{x}_t \text{ for all }t.
$
It immediately follows that
\[
 \phi_t(\Delta\hat{y}_{t}+u_{t}) = \hat{x}_t = \bar{x}_t = \Delta\hat{y}^b_{t}+u^b_{t} + (\Delta\hat{y}^s_{t}+u^s_{t})\hat{S}_t \text{ for all }t\notin\mathcal{I}.
\]
Suppose by contradiction that there exists some $t\in\mathcal{I}$ and $\nu\in\Omega_t$ such that $\hat{x}^\nu_t>\bar{x}^\nu_t$. Then $v_t(\hat{x}^\nu_t) > v_t(\bar{x}^\nu_t)$, so that
\[ \ttotal{t=0}{T}\mathbb{E}\big[v_{t}(\phi_{t}(\Delta\hat{y}_{t}+u_{t}))\big] > \ttotal{t=0}{T}\mathbb{E}\big[v_{t}(\Delta\hat{y}^b_{t}+u^b_{t} + (\Delta\hat{y}^s_{t}+u^s_{t})\hat{S}_t)\big]. \]
This contradicts \eqref{eq:equality-y-shadow}, and hence $\hat{y}$ satisfies \eqref{eq:TradeAtSpread}.
\end{proof}

\begin{proof}[Proof of Proposition \ref{prop:Constr-optimal-hedging}]
  Let $\process{J}{t}{0}{T}$ be the sequence of functions from Construction~\ref{constr:Jt} with $X=-\ttotal{t=0}{T}u_t$, and let $(\hat{\mathbb{Q}},\hat{S})$ be the pair from Construction~\ref{constr:optQS}. Recursive expansion of~\eqref{eq:QShat-and-Js} gives
\begin{align}
 J_t(\hat{S}_t)
 &= \mathbb{E}_{\hat{\mathbb{Q}}}\left[-\ttotal{s=0}{T}(u^b_s+u^s_s\hat{S}_T) + \ttotal{s=t}{T-1} a_{s+1}\ln\tfrac{\hat{q}_{s+1}}{p_{s+1}} \middle|\mathcal{F}_t\right] \text{ for all }t<T.\label{eq:recursive-J}
\end{align}
Let $\hat{x}$ be defined by \eqref{eq:xhat}. It follows from Remark \ref{remark:xhat} that
  \begin{equation} \label{eq:sum-xhat}
   \ttotal{t=0}{T} \hat{x}_t = \ttotal{t=0}{T-1} a_{t+1}\ln\tfrac{\hat{q}_{t+1}}{p_{t+1}} + \ttotal{t\in\mathcal{I}}{} \tfrac{1}{\alpha_{t}}\ln\tfrac{\hat{\lambda}_{u}}{\alpha_{t}} = \ttotal{t=0}{T-1} a_{t+1}\ln\tfrac{\hat{q}_{t+1}}{p_{t+1}} - J_0(\hat{S}_0).
  \end{equation}
  
  The first step in the proof is to show that the collection $\mathcal{W}_T$ in Construction \ref{alg:OptStockPosition} is non-empty. Theorem~\ref{th:solution-exists} guarantees the existence of a minimiser $\hat{y}\in\Psi$ for \eqref{eq:Problem 1}, and by Proposition \ref{prop:shadow-price-properties}\ref{prop:shadow-price-properties:2} it is also a minimiser in the friction-free model with stock price process $\hat{S}$. Combining this further with the uniqueness of $\hat{x}$, it follows that $\hat{y}$ satisfies \eqref{eq:TradeAtSpread} and
\begin{align}
 y_{-1}&=y_T=0, & \Delta y^b_t + u^b_t + (\Delta y^s_t + u^s_t)\hat{S}_t &= \hat{x}_{t}\text{ for all }t\ge0. \label{eq:yhat-minimiser}
\end{align}
The trading strategy $w\in\mathcal{N}^{2\prime}$ defined by
\begin{align} \label{eq:FlozStra_PaymentStra}
 w_{-1}&=0, & w_t &\coloneqq y_t + \ttotal{s=0}{t}(u^b_s-\hat{x}_s,u^s_s)\text{ for all }t=0,\ldots,T
\end{align}
satisfies 
\begin{align}
 (\Delta w^s_t)_+S^a_t - (\Delta w^s_t)_-S^b_t = \Delta w^s_t\hat{S}_t\text{ for all }t, & & w^s_T&=\ttotal{t=0}{T}u^s_t \label{eq:what-minimiser-trading}
\end{align}
by definition and by \eqref{eq:sum-xhat}
\begin{align}
 w^b_T = \ttotal{t=0}{T}u^b_t - \ttotal{t=0}{T-1} a_{t+1}\ln\tfrac{\hat{q}_{t+1}}{p_{t+1}} + J_0(\hat{S}_0). \label{eq:what-minimiser-final}
\end{align}
Moreover \eqref{eq:yhat-minimiser} gives the self-financing condition
\begin{equation} \label{eq:replicatePayoff_2}
 \Delta w^b_t + \Delta w^s_t \hat{S}_t = 0\text{ for all }t\ge0.
\end{equation}
Combining \eqref{eq:replicatePayoff_2} with the fact that $\hat{S}$ is a martingale under $\hat{\mathbb{Q}}$, it follows from standard arguments \cite[cf.][Th.~5.40]{cutland2012derivative} that
\begin{align} \label{eq:LinearEqForHedging}
   w^b_t+w^s_t\hat{S}_{t+1} &= \mathbb{E}_{\hat{\mathbb{Q}}}\big[w^b_T+w^s_T\hat{S}_T\big|\mathcal{F}_{t+1}\big]\text{ for all }t<T.
\end{align}
For every $t<T$, substituting \eqref{eq:what-minimiser-trading}, \eqref{eq:what-minimiser-final} and \eqref{eq:recursive-J} leads to
\begin{align*}
 w^b_t+w^s_t\hat{S}_{t+1} 
 &= \mathbb{E}_{\hat{\mathbb{Q}}}\left[\ttotal{s=0}{T}(u^b_s+u^s_s\hat{S}_T) - \ttotal{s=0}{T-1} a_{s+1}\ln\tfrac{\hat{q}_{t+1}}{p_{s+1}} + J_0(\hat{S}_0)\middle|\mathcal{F}_{t+1}\right] \\
 &= -J_{t+1}(\hat{S}_{t+1}) - \ttotal{s=0}{t} a_{s+1}\ln\tfrac{\hat{q}_{s+1}}{p_{s+1}} + J_0(\hat{S}_0).
 \end{align*}
After defining the stochastic process $\process{x^b}{t}{-1}{T}$ as
\[
 z^b_t \coloneqq 
 \begin{cases}
 0 & \text{if }t=-1,\\
 w^b_0 - J_0(\hat{S}_0) & \text{if }t=0,\\
 w^b_t + \ttotal{s=0}{t-1} a_{s+1}\ln\tfrac{\hat{q}_{s+1}}{p_{s+1}} - J_0(\hat{S}_0), &\text{if }t>0,
 \end{cases}
\]
this can be rewritten as
\begin{align*}
 z^b_t+w^s_t\hat{S}_{t+1} 
 &= -J_{t+1}(\hat{S}_{t+1}) - a_{t+1}\ln\tfrac{\hat{q}_{t+1}}{p_{t+1}}.
 \end{align*}
When combined with \eqref{eq:what-minimiser-trading}--\eqref{eq:what-minimiser-final}, this means that $(z^b_t,w^s_t)_{t=-1}^T\in\mathcal{W}_T$ and hence $\mathcal{W}_T\neq\emptyset$.

Now let $\mathcal{W}_T$ and $\mathcal{Y}$ be the collections of processes from Construction~\ref{alg:OptStockPosition}. By Proposition \ref{prop:shadow-price-properties}\ref{prop:shadow-price-properties:1} it suffices to show that every $\hat{y}\in\mathcal{Y}$ satisfies \eqref{eq:yhat-minimiser} and \eqref{eq:TradeAtSpread}, in other words, it minimises \eqref{eq:Problem 1} in the friction-free model with stock price process $\hat{S}$ and trades only at the spread. As $\hat{y}\in\mathcal{Y}$, there exists some $w\in\mathcal{W}_T$ satisfying \eqref{eq:constr:yhatb:Delta}--\eqref{eq:constr:yhats:Delta}. Taking the sum over all $t$ in \eqref{eq:constr:yhatb:Delta}--\eqref{eq:constr:yhats:Delta} and substituting \eqref{eq:sum-xhat} gives that $\hat{y}_T=0$. 
Turning to the properties of $w$, it satisfies \eqref{eq:what-minimiser-trading} by construction, which immediately gives \eqref{eq:TradeAtSpread}. Moreover,
\begin{equation} \label{eq:hatw-current}
 w^b_t + w^s_t\hat{S}_t=-J_t(\hat{S}_t) \text{ for all }t.
\end{equation}
For $t=T$ this comes from \eqref{eq:def_JT} and \eqref{eq:OpStockConditions-T}. For $t<T$ it is obtained by taking conditional expectation in \eqref{eq:OpStockEquation} with respect to $\hat{\mathbb{Q}}$ and $\mathcal{F}_t$, and substituting \eqref{eq:QShat-and-Js}. Combining \eqref{eq:hatw-current} with \eqref{eq:OpStockEquation} furthermore gives
\begin{equation}
 \Delta w^b_t + \Delta w^s_t\hat{S}_t = a_t\ln\tfrac{\hat{q}_t}{p_t} \text{ for all }t>0. \label{eq:Delta-w}
\end{equation}
The equalities \eqref{eq:hatw-current} for $t=0$ (recall $w_{-1}=0$) and \eqref{eq:Delta-w} for $t>0$ now combine with \eqref{eq:constr:yhatb:Delta}--\eqref{eq:constr:yhats:Delta} to give \eqref{eq:yhat-minimiser}, as required.
\end{proof}

\bibliographystyle{agsm}
\bibliography{regret_minimisation}

\end{document}